\documentclass[a4paper, english]{article}

\usepackage{lscape}
\usepackage{pdflscape}

\usepackage{stmaryrd}
\usepackage{mathdots}

\usepackage[geometry]{ifsym}
\usepackage{cmll}

\usepackage{comment}
\usepackage{amscd}
\usepackage{amssymb}
\usepackage{amsmath}
\usepackage{amsthm}
\usepackage{mathrsfs}

\usepackage{color}
\usepackage{xspace}
\usepackage{bussproofs}
\EnableBpAbbreviations

\usepackage{mathtools}
\usepackage{bpextra}

\usepackage{enumerate}

\usepackage{centernot}

\usepackage{hyperref}

\usepackage{cleveref}

\usepackage{caption}

\usepackage{multirow}

\usepackage{array}
\newcolumntype{C}[1]{>{\centering\arraybackslash}p{#1}}
\newcolumntype{L}[1]{>{\arraybackslash}p{#1}}

\usepackage{color}
\usepackage{cancel}
\usepackage{verbatim,cmll}
\usepackage{setspace}
\usepackage{lscape}
\usepackage{latexsym,relsize}
\usepackage{multicol}
\usepackage[position=top]{subfig}
\usepackage{leqno}
\usepackage{tikz}

\usepackage{txfonts}

\usepackage{multicol}
\usetikzlibrary{arrows}
\usetikzlibrary{matrix}
\usetikzlibrary{patterns}
\usetikzlibrary{shapes}

\def\aol{\rule[0.5865ex]{1.38ex}{0.1ex}}

\newcommand{\FH}{\hat{f}}

\newcommand{\GC}{\check{g}}

\newcommand{\FHS}{\hat{f}^{\,\sharp}}
\newcommand{\FCS}{\check{f}^{\,\sharp}}
\newcommand{\GHF}{\hat{g}^{\,\flat}}
\newcommand{\GCF}{\check{g}^{\,\flat}}

\newcommand{\mand}{\otimes}

\newcommand{\mor}{\oplus}

\newcommand{\MAND}{\,\hat{\otimes}\,}
\newcommand{\MRARR}{\,\,\check{\backslash}\,\,}

\newcommand{\MOR}{\,\check{\oplus}\,}
\newcommand{\MDLARR}{\,\hat{\varobslash}\,}

\newcommand{\aatop}{\ensuremath{\top}\xspace}
\newcommand{\abot}{\ensuremath{\bot}\xspace}
\newcommand{\aand}{\ensuremath{\wedge}\xspace}
\newcommand{\aor}{\ensuremath{\vee}\xspace}
\newcommand{\ararr}{\ensuremath{\rightarrow}\xspace}
\newcommand{\alarr}{\ensuremath{\leftarrow}\xspace}
\newcommand{\adrarr}{\ensuremath{\,{>\mkern-7mu\raisebox{-0.065ex}{\rule[0.5865ex]{1.38ex}{0.1ex}}}\,}\xspace}
\newcommand{\adlarr}{\ensuremath{\rotatebox[origin=c]{180}{$\,>\mkern-8mu\raisebox{-0.065ex}{\aol}\,$}}\xspace}
\newcommand{\ATOP}{\hat{\top}}
\newcommand{\ABOT}{\ensuremath{\check{\bot}}\xspace}
\newcommand{\AAND}{\ensuremath{\:\hat{\wedge}\:}\xspace}
\newcommand{\AOR}{\ensuremath{\:\check{\vee}\:}\xspace}
\newcommand{\ARARR}{\ensuremath{\:\check{\rightarrow}\:}\xspace}
\newcommand{\ALARR}{\ensuremath{\:\check{\leftarrow}\:}\xspace}
\newcommand{\ADRARR}{\ensuremath{\hat{{\:{>\mkern-7mu\raisebox{-0.065ex}{\rule[0.5865ex]{1.38ex}{0.1ex}}}\:}}}\xspace}
\newcommand{\ADLARR}{\ensuremath{\:\hat{\rotatebox[origin=c]{180}{$\,>\mkern-8mu\raisebox{-0.065ex}{\aol}\,$}}\:}\xspace}

\newcommand{\wbox}{\ensuremath{\Box}\xspace}
\newcommand{\wdia}{\ensuremath{\Diamond}\xspace}
\newcommand{\bbox}{\ensuremath{\blacksquare}\xspace}
\newcommand{\bdia}{\ensuremath{\Diamondblack}\xspace}

\newcommand{\WBOX}{\ensuremath{\check{\Box}}\xspace}
\newcommand{\WDIA}{\ensuremath{\hat{\Diamond}}\xspace}
\newcommand{\BBOX}{\ensuremath{\check{\blacksquare}}\xspace}
\newcommand{\BDIA}{\ensuremath{\hat{\Diamondblack}}\xspace}

\renewcommand{\epsilon}{\varepsilon}

\newcommand{\ox}{\overline{x}}
\newcommand{\oy}{\overline{y}}
\newcommand{\oz}{\overline{z}}

\newcommand{\blhd}{\blacktriangleleft}
\newcommand{\brhd}{\blacktriangleright}

\newcommand{\bba}{\mathbb{A}}
\newcommand{\bbA}{\mathbb{A}}
\newcommand{\bbL}{\mathbb{L}}

\newcommand{\AATOP}{\hat{\top}}

\newcommand{\marginnote}[1]{\marginpar{\raggedright\tiny{#1}}}

\tikzset{
	treenode/.style = {align=center, inner sep=0pt, text centered},
	Ske/.style = {treenode, ellipse, double, draw=black,
		minimum width=6pt, thick},% arbre rouge noir, noeud noir
	PIA/.style = {treenode, ellipse, black, draw=black,
		minimum width=6pt},% arbre rouge noir, noeud rouge
	Crit/.style = {treenode, rectangle, draw=black,
		minimum width=0.5em, minimum height=0.5em}% arbre rouge noir, nil
}

\theoremstyle{plain}
\newtheorem{thm}{Theorem}[section]

\newtheorem{cor}[thm]{Corollary}
\newtheorem{prop}[thm]{Proposition}

\newtheorem{lemma}[thm]{Lemma}
\newtheorem{notation}[thm]{Notation}

\theoremstyle{definition}

\newtheorem{definition}[thm]{Definition}

\newtheorem{example}[thm]{Example}

\newtheorem{rem}[thm]{Remark}

\newtheorem{remark}[thm]{Remark}

\def\fCenter{{\mbox{$\ \vdash\ $}}}

\EnableBpAbbreviations

\newcommand{\fns}{\footnotesize}
\newcommand{\mc}{\multicolumn}

\newcommand{\commment}[1]{}
\numberwithin{equation}{section}

\def\pdra{\mbox{$\,>\mkern-8mu\raisebox{-0.065ex}{\aol}\,$}}
\def\pdla{\mbox{\rotatebox[origin=c]{180}{$\,>\mkern-8mu\raisebox{-0.065ex}{\aol}\,$}}}
\newcommand{\pand}{\wedge}

\newcommand{\por}{\vee}

\newcommand{\pra}{\rightarrow}

\newcommand{\cfDLE}{\underline{D.LE}}
\newcommand{\cfDDLE}{\underline{D.DLE}}

\newcommand{\bp}{\textcolor{blue}{p}}
\newcommand{\bx}{\textcolor{blue}{x}}
\newcommand{\bu}{\textcolor{blue}{u}}
\newcommand{\bz}{\textcolor{blue}{z}}

\newcommand{\bi}{\textcolor{blue}{i}}

\newcommand{\bk}{\textcolor{blue}{k}}

\newcommand{\obp}{\overline{\textcolor{blue}{p}}}
\newcommand{\obx}{\overline{\textcolor{blue}{x}}}
\newcommand{\obz}{\overline{\textcolor{blue}{z}}}
\newcommand{\rrq}{\textcolor{red}{q}}
\newcommand{\ru}{\textcolor{red}{u}}
\newcommand{\ry}{\textcolor{red}{y}}

\newcommand{\rw}{\textcolor{red}{w}}

\newcommand{\rj}{\textcolor{red}{j}}

\newcommand{\rh}{\textcolor{red}{h}}
\newcommand{\orq}{\overline{\textcolor{red}{q}}}
\newcommand{\ory}{\overline{\textcolor{red}{y}}}
\newcommand{\orw}{\overline{\textcolor{red}{w}}}
\newcommand{\ba}{\textcolor{blue}{\alpha}}
\newcommand{\oba}{\overline{\textcolor{blue}{\alpha}}}
\newcommand{\bpsi}{\textcolor{blue}{\psi}}

\newcommand{\bgamma}{\textcolor{blue}{\gamma}}

\newcommand{\bS}{\textcolor{blue}{S}}
\newcommand{\bsigma}{\textcolor{blue}{\sigma}}
\newcommand{\orb}{\overline{\textcolor{red}{\beta}}}
\newcommand{\rxi}{\textcolor{red}{\xi}}

\newcommand{\rdelta}{\textcolor{red}{\delta}}

\newcommand{\rU}{\textcolor{red}{U}}
\newcommand{\rtau}{\textcolor{red}{\tau}}

\usepackage[left=2cm,top=3cm,right=2cm,bottom=3cm]{geometry}
\usepackage{authblk}

\title{Syntactic completeness of proper display calculi}

\author[1,2]{Jinsheng Chen}
\author[2]{Giuseppe Greco}
\author[2,3]{Alessandra Palmigiano\thanks{The research of the second and third author has been funded in part by the NWO grant KIVI.2019.001.}}
\author[2]{Apostolos Tzimoulis}
\affil[1]{Department of Philosophy, Zhejiang University}
\affil[2]{SBE, Vrije Universiteit Amsterdam}
\affil[3]{Department of Mathematics and Applied Mathematics, University of Johannesburg, South Africa}

\date{}

\begin{document}

\maketitle

\begin{abstract}
A recent strand of research in structural proof theory aims at exploring the notion of {\em analytic calculi} (i.e.~those calculi that support general and modular proof-strategies for cut elimination), and at identifying classes of logics that can be captured in terms of these calculi.  In this context, Wansing introduced the notion of {\em proper display calculi} as one possible design framework for proof calculi in which the analiticity desiderata are realized in a particularly transparent way. Recently, the theory of {\em properly displayable} logics (i.e.~those logics that can be equivalently presented with some proper display calculus) has been developed in connection with generalized Sahlqvist theory (aka unified correspondence). Specifically, properly displayable logics  have been syntactically characterized as those axiomatized by {\em analytic inductive axioms}, which can be equivalently and algorithmically transformed into analytic structural rules so that the resulting proper display calculi enjoy a set of basic properties: soundness, completeness, conservativity, cut elimination and subformula property. In this context, the proof that the given calculus is {\em complete} w.r.t.~the original logic is usually carried out {\em syntactically}, i.e.~by showing that a (cut free) derivation exists of each given axiom of the logic in the basic system  to which the analytic structural rules algorithmically generated from the given axiom have been added. However, so far this proof strategy for {\em syntactic completeness} has been implemented on a case-by-case base, and not in general. In this paper, we address this gap by proving syntactic completeness for properly displayable logics in any normal (distributive) lattice expansion signature. Specifically, we show that for every analytic inductive axiom a cut free derivation can be effectively generated which has a specific shape, referred to as {\em pre-normal form}.\\
%
%effectively transformed into analytic rules, i.e.~rules that preserve cut elimination when added to a calculus with cut elimination.   LE-logics are logics canonically associated with varieties of lattice expansions; the LE-logics that are {\em properly displayable} (i.e.~those LE-logics that can be equivalently presented with a proper display calculus, according to Wansing's definition) have been syntactically characterized as those axiomatized by {\em analytic inductive axioms}. The general theory of properly displayable logics guarantees that analytic inductive axioms can be equivalently and algorithmically transformed into analytic rules so that the resulting proper display calculi enjoy a set of basic properties: soundness, completeness, conservativity, cut elimination and subformula property.  As to proving that the given calculus is complete w.r.t.~the original logic Let $\mathsf{L}$ be a properly displayable LE-logic and let D.$\mathsf{L}$ be a proper display calculus obtained by adding to the basic display calculus in the relevant signature the (ALBA-generated) analytic (structural) rule(s) equivalent to the axioms of $\mathsf{L}$. In this article we show that D.$\mathsf{L}$ has what we call \emph{syntactical completeness}, i.e.~every theorem (or derivable sequent) of $\mathsf{L}$ has a D.$\mathsf{L}$-derivation $\pi$ that is (i) effectively generated, (ii) cut-free, and (iii) in {\em pre-normal form}, namely a specific form that we can syntactically characterize. \\
%

\end{abstract}

\tableofcontents
	
	\section{Introduction}
In recent years, research in structural proof theory has focused on {\em analytic calculi} \cite{negri2005proof,ciabattoni2008axioms,GMPTZ,Belnap,Wa98,Wan02}, understood as those calculi supporting a {\em robust} form of cut elimination, i.e.~one which is preserved by adding rules of a specific shape (the analytic rules). 
Important results on analytic calculi have been obtained 
in the context of various proof-theoretic formalisms:  (classes of) axioms have been  identified for which equivalent correspondences with analytic rules have been established algorithmically or semi-algorithmically. Without claiming to be exhaustive, we briefly review this strand of research as it has been developed in the context of \emph{sequent and labelled calculi} \cite{NegVonPla98,Neg03,negri2005proof}, \emph{sequent and hypersequent calculi} \cite{ciabattoni2008axioms,lahav2013frame,lellmann2014axioms}, and \emph{(proper) display calculi} \cite{Kracht,CiRa14,GMPTZ}).%lellmann2013correspondence%%ciabattoni2012algebraic%%CiRa13%%gore2011correspondence

In \cite{NegVonPla98},  a methodology is established, sometimes referred to as \emph{axioms-as-rules}, for transforming \emph{universal axioms} in the language of first order classical (or intuitionistic) logics into analytic sequent rules. As remarked in the same paper, this methodology has a precursor in \cite{Neg99} for the intuitionistic theories of apartness and order. The rules so generated are then used to expand the sequent calculus \textbf{G3c} for first order classical logic. In \cite{Neg03}, the axioms-as-rules methodology is generalized so as to capture the so-called \emph{geometric implications} in the language of first order classical logic, i.e.~formulas of the form $\forall \overline{z} (A \rightarrow B)$ where $A$ and $B$ are \emph{geometric formulas} (i.e.~first-order formulas not containing $\rightarrow$ or $\forall$). In \cite{negri2005proof},  the axioms-as-rules methodology is applied to capture various normal modal logic axioms via equivalent analytic labelled-calculi rules over the basic labelled calculus \textbf{G3K} for the modal logic K; moreover, following the standard methods as for the \textbf{G3}-style sequent calculi, the admissibility of cut, substitution and contraction is established. Although these calculi do not satisfy the full subformula property, decidability is established thanks to their enjoying the so-called \emph{subterm property} (requiring all the terms in minimal derivations to occur in the endsequent) and height-preserving admissibility of contraction. 
%The book \cite{Vig00} xxx. \marginnote{Qui si deve completare. Va bene come ho modificato questo paragrafo?}

In \cite{ciabattoni2008axioms}, a \emph{hierarchy} (sometimes referred to as \emph{substructural hierarchy}) is defined of classes of substructural formulas,  %of \textbf{FLe}, i.e.~the (unital, associative) full Lambek calculus with exchange, 
and it is shown how to translate substructural axioms %in this language 
up to level $\mathcal{N}_2$ of the hierarchy into equivalent rules of a \emph{Gentzen-style sequent calculus},  and axioms up to a subclass of  level $\mathcal{P}_3$ into equivalent rules of a \emph{hypersequent calculus}; the rules so generated are then transformed into equivalent analytic rules whenever they satisfy an additional condition or the base calculus admits weakening; cut elimination is proved via a semantic argument extending the semantic proof of \cite{Oka02} to hypersequent calculi (and in \cite{ciabattoni2012algebraic}, this approach is generalized  to multi-conclusions hypersequents, and  a heuristic is proposed to go beyond $\mathcal{P}_3$ axioms). In \cite{lahav2013frame}, $n$-simple formulas, a particularly well-behaved proper subset of geometric formulas  \cite{negri2005proof}, are identified, and  a method is introduced which transforms $n$-simple formulas into equivalent hypersequent rules for a variety of normal modal logics extending the modal logics \textbf{K}, \textbf{K4}, or \textbf{KB}; cut admissibility is proved for $n$-simple extensions of \textbf{K} and \textbf{K4}, and decidability (via standard sub-formula property) is established for $n$-simple extensions of \textbf{KB}. In \cite{lellmann2014axioms},  the format of hypersequent rules with \emph{context restrictions} is introduced, and transformations are studied between rules and modal axioms on a classical or intuitionistic base; decidability and complexity results are proved for a variety of modal logics, as well as uniform cut elimination extending the proof in \cite{ciabattoni2008axioms}. In \cite{CiaLanRam19},  hypersequent calculi are studied capturing analytic extensions of the full Lambek calculus \textbf{FLe}, and a procedure is introduced  for translating structural rules into equivalent formulas in  \emph{disjunction form}. %\marginnote{disjunction form is related to definite, we should probably add a footnote when we talk about definite} 
This approach is also applied to some normal modal logics on a classical base. The main goal of \cite{CiaLanRam19} is to show that cut-free derivations in hypersequent calculi can be transformed into derivations in sequent calculi satisfying various weaker versions of the subformula property which still guarantee decidability (although not necessarily cut elimination). Specifically, \cite[Theorem 12(i)]{CiaLanRam19}  shows how to construct a derivation in  hypersequent calculi of formulas in disjunction form which are equivalent to structural rules. %This result would be the counterpart of what we call \emph{syntactic completeness} (see the discussion in the last paragraph and Section \ref{sec:syntactic completeness}), with the proviso that here we handle logics in arbitrary (D)LE signatures and we are not concerned with the specific `special' or `normal' form of the analytic formulas under consideration.
%\marginnote{G: maybe we want to reconsider this last sentence on the comparison between the two results.\\ AP: I think this comparison is not pertinent to this paper; we should have done it in GMPTZ. But if we want to say something we should do so in the conclusions, not here}

%\cite{gore2011correspondence} shows how to transform axioms in the language of modal tense logic into rules of a display calculus (and rules of a nested calculus). 

In \cite{Kracht}, the syntactic shape of \emph{primitive axioms} in the language of tense modal logic on a classical base is characterized as the one which can be equivalently captured as analytic structural rules extending the minimal display calculus for tense logic. In \cite{CiRa14}, an analogous characterization is provided in a more general setting for a given but not fixed display calculus, by introducing a procedure for transforming axioms into analytic structural rules and showing the converse direction whenever the calculus satisfies additional conditions.

In \cite{GMPTZ}, which is the contribution in the line of research described above to which the results of the present paper most directly connect, a characterization, analogous to the one of \cite{CiRa14},\footnote{For a comparison between the characterizations in \cite{CiRa14} and in  \cite{GMPTZ}, see \cite[Section 9]{GMPTZ}. 
%\cite{GMPTZ} extends the characterization of \cite{Kracht} to any distributive lattice expansion. 
} of the property of being properly displayable\footnote{The adjective `proper' singles out a subclass of Belnap's display calculi  \cite{Belnap} identified  by Wansing in \cite[Section 4.1]{Wa98}. A display calculus is  {\em proper} if every structural rule is closed under uniform substitution. This requirement strengthens Belnap's conditions C$_6$ and C$_7$. In \cite{Multitrends}, this requirement is extended  to  \emph{multi-type} display calculi. A logic is {\em (properly) displayable} if it can be captured by some (proper) display calculus (see \cite[Section 2.2]{GMPTZ}).} is obtained for arbitrary normal (D)LE-logics\footnote{Normal (D)LE-logics are those logics algebraically captured by varieties of normal (distributive) lattice expansions, i.e.~(distributive) lattices endowed with additional operations that are  finitely join-preserving or meet-reversing in each coordinate, or are finitely meet-preserving or join-reversing in each coordinate.} via a systematic connection between proper display calculi and generalized Sahlqvist correspondence theory (aka unified correspondence \cite{CoGhPa13,ConPal12,CoPa11,DeRPal20}). Thanks to this connection, general meta-theoretic results are established for properly displayable (D)LE-logics. In particular, in \cite{GMPTZ}, the properly displayable (D)LE-logics are syntactically characterized as the logics axiomatised by {\em analytic inductive axioms} (cf.~Definition \ref{def:type5}); moreover, the same algorithm ALBA which computes the first-order correspondent of (analytic) inductive (D)LE-axioms can be used to effectively compute their corresponding analytic structural rule(s). 

\medskip 

%The present paper continues the line of investigation initiated in \cite{GMPTZ},
The {\em semantic} equivalence between each analytic inductive axiom $\varphi\vdash \psi$ and its corresponding analytic structural rule(s) $R_1, \ldots, R_n$ is an immediate consequence of the soundness of the rules of ALBA on perfect normal (distributive) lattice expansions (cf.~Footnote \ref{def:can:ext}). However, on the {\em syntactic} side, an effective procedure was still missing for building {\em cut-free} derivations of $\varphi\vdash \psi$ in the proper display calculus obtained by adding $R_1,\ldots,R_n$ to the  basic proper display calculus $\mathrm{D.LE}$ (resp.~$\mathrm{D.DLE}$) of the basic normal (D)LE-logic. Such an effective procedure would establish, via syntactic means, that for any  properly displayable (D)LE-logic $\mathsf{L}$, the proper display calculus for $\mathsf{L}$---i.e.~the calculus obtained by adding the analytic structural rules corresponding to the axioms of $\mathsf{L}$ to the basic calculus $\mathrm{D.LE}$ (resp.~$\mathrm{D.DLE}$)---derives all the theorems (or derivable sequents) of $\mathsf{L}$. This is what we refer to as the {\em syntactic completeness} of the proper display calculus for $\mathsf{L}$ with respect to any analytic (D)LE-logic $\mathsf{L}$. This syntactic completeness result for all properly displayable logics in arbitrary (D)LE-signatures is the main contribution of the present paper. It is perhaps worth to emphasize that we do not just show that any analytic inductive axiom is derivable in its corresponding proper display calculus, but we also provide an algorithm to generate a {\em cut-free} derivation of a particular shape that we refer to as being  \emph{in pre-normal form} (see Section \ref{ssec: canonical form}).

		\paragraph{Structure of the paper} 
		In Section \ref{subset:language:algsemantics}, we collect the necessary preliminaries about (D)LE-logics, their language, their basic presentation and notational conventions, algebraic semantics, basic proper display calculi, and analytic inductive LE-inequalities.
		In Section \ref{sec:properties}, we prove a series of technical properties of the basic proper display calculi which will be needed for achieving our main result, which is then proven in  Section \ref{sec:syntactic completeness}. We conclude
		in Section \ref{sec:conclusions}.
	%%%%%%%%%%%%%%%%%%%%%%%%%%%

\section{Preliminaries}
\label{subset:language:algsemantics}

The present section adapts material from \cite[Section 2]{conradie2019algorithmic}, \cite[Section 2]{GMPTZ}, \cite[Section 2]{greco2018algebraic}, and \cite[Section 2]{CPT-Goldblatt}.

\subsection{Basic normal $\mathrm{LE}$-logics and their algebras}
Our base language is an unspecified but fixed language $\mathcal{L}_\mathrm{LE}$, to be interpreted over  lattice expansions of compatible similarity type. This setting uniformly accounts for many well known logical systems, such as full Lambek calculus and its axiomatic extensions, full Lambek-Grishin calculus, and other lattice-based logics.

In our treatment, we make use of the following auxiliary definition: an {\em order-type} over $n\in \mathbb{N}$
%
%PJ: This footnote is not really used or needed.
%\footnote{Throughout the paper, order-types will be typically associated with arrays of variables $\vec p: = (p_1,\ldots, p_n)$. When the order of the variables in $\vec p$ is not specified, we will sometimes abuse notation and write $\varepsilon(p) = 1$ or $\varepsilon(p) = \partial$.} 
is an $n$-tuple $\epsilon\in \{1, \partial\}^n$. 
For every order type $\epsilon$, we denote its {\em opposite} order type by $\epsilon^\partial$, that is, $\epsilon^\partial_i = \epsilon^\partial (i) = 1$ iff $\epsilon_i = \epsilon(i)=\partial$ for every $1 \leq i \leq n$, and $\epsilon^\partial_i = \epsilon^\partial (i) = \partial$ iff $\epsilon_i = \epsilon(i)=1$ for every $1 \leq i \leq n$. For any lattice $\bba$, we let $\bba^1: = \bba$ and $\bba^\partial$ be the dual lattice, that is, the lattice associated with the converse partial order of $\bba$. For any order type $\varepsilon$, we let $\bba^\varepsilon: = \Pi_{i = 1}^n \bba^{\varepsilon_i}$.

The language $\mathcal{L}_\mathrm{LE}(\mathcal{F}, \mathcal{G})$ (from now on abbreviated as $\mathcal{L}_\mathrm{LE}$) takes as parameters: a denumerable set of proposition letters $\mathsf{AtProp}$, elements of which are denoted $p,q,r$, possibly with indexes, and disjoint sets of connectives $\mathcal{F}$ and $\mathcal{G}$.\footnote{\label{footnote: f and g operators}
	%It will be clear from the treatment in the present and the following sections that 
	The connectives in $\mathcal{F}$ (resp.\ $\mathcal{G}$) correspond to those referred to as {\em positive} (resp.\ {\em negative}) connectives in \cite{ciabattoni2008axioms}. This terminology is not adopted in the present paper to avoid confusion with positive and negative nodes in signed generation trees, defined later in this section.
	%is explained later on in Footnote \ref{footnote: why not adopt terminology of CiGaTe}. 
	Our assumption that the sets $\mathcal{F}$ and $\mathcal{G}$ are disjoint is motivated by the desideratum of generality and modularity. Indeed, for instance, the order theoretic properties of Boolean negation $\neg$ guarantee that this connective belongs both to $\mathcal{F}$ and to $\mathcal{G}$. In such cases we prefer to define two copies $\neg_\mathcal{F}\in\mathcal{F}$ and $\neg_\mathcal{G}\in\mathcal{G}$, and introduce structural rules which encode the fact that these two copies coincide. Another possibility is to admit a non empty intersection of the sets $\mathcal{F}$ and $\mathcal{G}$. Notice that only unary connectives can be both left and right adjoints. Whenever a connective belongs both to $\mathcal{F}$ and to $\mathcal{G}$ a completely standard solution in the display calculi literature is also available (c.f.~Remark \ref{rem: overloading the notation} and \ref{rem: both f and g operators}).} Each $f\in \mathcal{F}$ and $g\in \mathcal{G}$ has arity $n_f\in \mathbb{N}$ (resp.\ $n_g\in \mathbb{N}$) and is associated with some order-type $\varepsilon_f$ over $n_f$ (resp.\ $\varepsilon_g$ over $n_g$). 
%\footnote{
Unary $f \in \mathcal{F}$ (resp.\ $g \in \mathcal{G}$) are sometimes denoted  $\Diamond$ (resp.\ $\Box$) if their order-type is 1, and $\lhd$ (resp.\ $\rhd$) if their order-type is $\partial$.\footnote{The adjoints of the unary connectives $\Box$, $\Diamond$, $\lhd$ and $\rhd$ are denoted $\Diamondblack$, $\blacksquare$, $\blhd$ and $\brhd$, respectively.} The terms (formulas) of $\mathcal{L}_\mathrm{LE}$ are defined recursively as follows:
\[
\varphi ::= p \mid \bot \mid \top \mid \varphi \wedge \varphi \mid \varphi \vee \varphi \mid f(\varphi_1, \ldots, \varphi_{n_f}) \mid g(\varphi_1, \ldots, \varphi_{n_g})
\] 
where $p \in \mathsf{AtProp}$. Terms in $\mathcal{L}_\mathrm{LE}$ are denoted either by $s,t$, or by lowercase Greek letters such as $\varphi, \psi, \gamma$. 

\begin{definition}
	\label{def:LE}
	For any tuple $(\mathcal{F}, \mathcal{G})$ of disjoint sets of function symbols as above, a {\em  lattice expansion} (abbreviated as LE) is a tuple $\bba = (\bbL, \mathcal{F}^\bbA, \mathcal{G}^\bbA)$ such that $\bbL$ is a bounded  lattice, $\mathcal{F}^\bbA = \{f^\bbA\mid f\in \mathcal{F}\}$ and $\mathcal{G}^\bbA = \{g^\bbA\mid g\in \mathcal{G}\}$, such that every $f^\bbA\in\mathcal{F}^\bbA$ (resp.\ $g^\bbA\in\mathcal{G}^\bbA$) is an $n_f$-ary (resp.\ $n_g$-ary) operation on $\bbA$. An LE $\mathbb{A}$ is {\em normal} if every $f^\bbA\in\mathcal{F}^\bbA$ (resp.\ $g^\bbA\in\mathcal{G}^\bbA$) preserves finite -- hence also empty -- joins (resp.\ meets) in each coordinate with $\epsilon_{f}(i)=1$ (resp.\ $\epsilon_{g}(i)=1$) and reverses finite -- hence also empty -- meets (resp.\ joins) in each coordinate with $\epsilon_{f}(i)=\partial$ (resp.\ $\epsilon_{g}(i)=\partial$).\footnote{\label{footnote:LE vs DLO} Normal LEs are sometimes referred to as {\em  lattices with operators} (LOs). This terminology comes from the setting of Boolean algebras with operators, in which operators are operations which preserve finite joins in each coordinate. However, this terminology is somewhat ambiguous in the lattice setting, in which primitive operations are typically maps which are operators if seen as $\bbA^\epsilon\to \bbA^\eta$ for some order-type $\epsilon$ on $n$ and some order-type $\eta\in \{1, \partial\}$.
		% Rather than speaking of lattices with $(\varepsilon, \eta)$-operators, we then speak of normal LEs.
	} Let $\mathbb{LE}$ be the class of LEs. Sometimes we will refer to certain LEs as $\mathcal{L}_\mathrm{LE}$-algebras when we wish to emphasize that these algebras have a compatible signature with the logical language we have fixed.
\end{definition}
In the remainder of the paper, %\marginnote{Apostolos: I replaced the double appearance of ``in the remainder of the paper''; hope it is ok.},
we will often simplify notation and write e.g.\ $f$ for $f^\bbA$, $n$ for $n_f$ and $\varepsilon_i$ for $\varepsilon_{f}(i)$. We also extend the $\{1,\partial\}$-notation to the symbols $\vee,\wedge,\bot,\top,\le,\vdash$ by stipulating that the superscript $^1$ denotes the identity map, defining 
$$
\vee^\partial=\wedge,\qquad \wedge^\partial=\vee,\qquad \bot^\partial=\top,\qquad \top^\partial=\bot,\qquad {\le^\partial}={\ge},
$$
and stipulating that $\varphi\vdash^\partial \psi$ stands for $\psi\vdash \varphi$.

%Normal LEs constitute the main semantic environment of the present paper. 
Henceforth, %every LE is assumed to be normal, so 
the adjective `normal' will typically be dropped. The class of all LEs is equational, and can be axiomatized by the usual lattice identities and the following equations for any $f\in \mathcal{F}$, $g\in \mathcal{G}$ and $1\leq i\leq n$:
%	\begin{itemize}
\[f(p_1\ldots, q\vee^{\epsilon_f(i)} r,\ldots p_{n_f}) = f(p_1\ldots, q,\ldots p_{n_f})\vee f(p_1 \ldots, r,\ldots p_{n_f}),\quad
f(p_1\ldots, \bot^{\epsilon_f(i)},\ldots p_{n_f}) = \bot,
\]
\[g(p_1\ldots, q\wedge^{\epsilon_g(i)} r,\ldots p_{n_g}) = g(p_1 \ldots, q,\ldots p_{n_g})\wedge g(p_1\ldots, r,\ldots p_{n_g}),\quad
g(p_1\ldots, \top^{\epsilon_g(i)},\ldots p_{n_g}) = \top.
\]
%		\item $f(p_1,\ldots, \bot^{\epsilon_i},\ldots,p_{n}) = \bot$,
%		\item $g(p_1,\ldots, \top^{\epsilon_i},\ldots,p_{n}) = \top$.
%	\end{itemize}
%	\begin{itemize}
%		\item if $\varepsilon_f(i) = 1$, then $f(p_1,\ldots, p\vee q,\ldots,p_{n_f}) = f(p_1,\ldots, p,\ldots,p_{n_f})\vee f(p_1,\ldots, q,\ldots,p_{n_f})$ and $f(p_1,\ldots, \bot,\ldots,p_{n_f}) = \bot$,
%		\item if $\varepsilon_f(i) = \partial$, then $f(p_1,\ldots, p\wedge q,\ldots,p_{n_f}) = f(p_1,\ldots, p,\ldots,p_{n_f})\vee f(p_1,\ldots, q,\ldots,p_{n_f})$ and $f(p_1,\ldots, \top,\ldots,p_{n_f}) = \bot$,
%		\item if $\varepsilon_g(j) = 1$, then $g(p_1,\ldots, p\wedge q,\ldots,p_{n_g}) = g(p_1,\ldots, p,\ldots,p_{n_g})\wedge g(p_1,\ldots, q,\ldots,p_{n_g})$ and $g(p_1,\ldots, \top,\ldots,p_{n_g}) = \top$,
%		\item if $\varepsilon_g(j) = \partial$, then $g(p_1,\ldots, p\vee q,\ldots,p_{n_g}) = g(p_1,\ldots, p,\ldots,p_{n_g})\wedge g(p_1,\ldots, q,\ldots,p_{n_g})$ and $g(p_1,\ldots, \bot,\ldots,p_{n_g}) = \top$.
%	\end{itemize}
Each language $\mathcal{L}_\mathrm{LE}$ is interpreted in the appropriate class of LEs. In particular, for every LE $\bba$, each operation $f^\bba\in \mathcal{F}^\bbA$ (resp.\ $g^\bba\in \mathcal{G}^\bbA$) is finitely join-preserving (resp.\ meet-preserving) in each coordinate when regarded as a map $f^\bba: \bba^{\varepsilon_f}\to \bba$ (resp.\ $g^\bba: \bba^{\varepsilon_g}\to \bba$).

The generic LE-logic is not equivalent to a sentential logic. Hence the consequence relation of these logics cannot be uniformly captured in terms of theorems, but rather in terms of sequents, which motivates the following definition:
\begin{definition}
	\label{def:LE:logic:general}
	For any language $\mathcal{L}_\mathrm{LE} = \mathcal{L}_\mathrm{LE}(\mathcal{F}, \mathcal{G})$, the {\em basic}, or {\em minimal} $\mathcal{L}_\mathrm{LE}$-{\em logic} is a set of sequents $\varphi\vdash\psi$, with $\varphi,\psi\in\mathcal{L}_\mathrm{LE}$, which contains as axioms the following sequents for lattice operations and additional connectives: 
	%		\begin{itemize}
	%			\item Sequents for lattice operations:\footnote{In what follows we will use the turnstile symbol $\vdash$ both as sequent separator and also as the consequence relation of the logic.}
	\begin{align*}
	&p\vdash p, && \bot\vdash p, && p\vdash \top,&&\\
	&p\vdash p\vee q  && q\vdash p\vee q, && p\wedge q\vdash p, && p\wedge q\vdash q,
	\end{align*}
	$$
	f(p_1\ldots, q\vee^{\epsilon_f(i)} r,\ldots p_{n_f}) \vdash f(p_1\ldots, q,\ldots p_{n_f})\vee f(p_1 \ldots, r,\ldots p_{n_f}),\quad f(p_1,\ldots, \bot^{\epsilon_f(i)},\ldots,p_{n_f}) \vdash \bot,
	$$
	$$
	g(p_1 \ldots, q,\ldots p_{n_g})\wedge g(p_1 \ldots, r,\ldots p_{n_g})\vdash g(p_1 \ldots, q\wedge^{\epsilon_g(i)} r,\ldots p_{n_g}),\quad \top\vdash g(p_1,\ldots, \top^{\epsilon_g(i)},\ldots,p_{n_g}),
	$$
	and is closed under the following inference rules (note that $\varphi\vdash^\partial\psi$ means $\psi\vdash\varphi$):
	\begin{displaymath}
	\frac{\varphi\vdash \chi\quad \chi\vdash \psi}{\varphi\vdash \psi}
	\qquad
	\frac{\varphi\vdash \psi}{\varphi(\chi/p)\vdash\psi(\chi/p)}
	\qquad
	\frac{\chi\vdash\varphi\quad \chi\vdash\psi}{\chi\vdash \varphi\wedge\psi}
	\qquad
	\frac{\varphi\vdash\chi\quad \psi\vdash\chi}{\varphi\vee\psi\vdash\chi}
	\end{displaymath}
	\begin{displaymath}
	\frac{\ \ \ \ \varphi\vdash^{\epsilon_{f}(i)}\psi}{f(p_1,\ldots,\varphi,\ldots,p_n)\vdash f(p_1,\ldots,\psi,\ldots,p_n)} \qquad
	\frac{\ \ \ \ \varphi \vdash^{\epsilon_{g}(i)}\psi}{g(p_1,\ldots,\varphi,\ldots,p_n)\vdash g(p_1,\ldots,\psi,\ldots,p_n)}.
	\end{displaymath}
	
	\medskip
	
	We let $\mathbf{L}_\mathrm{LE}$ denote the minimal $\mathcal{L}_\mathrm{LE}$-logic. We typically drop reference to the parameters when they are clear from the context. By an {\em $\mathrm{LE}$-logic} we understand any axiomatic extension of $\mathbf{L}_\mathrm{LE}$ in the language $\mathcal{L}_{\mathrm{LE}}$. If all the axioms in the extension are analytic inductive (cf.~Definition \ref{def:type5}) we say that the given $\mathrm{LE}$-logic is {\em analytic}.
\end{definition}

%For every LE $\bba$, the symbol $\vdash$ is interpreted as the lattice order $\leq$. 
A sequent $\varphi\vdash\psi$ is valid in  an LE $\bba$ if $h(\varphi)\leq h(\psi)$ for every homomorphism $h$ from the $\mathcal{L}_\mathrm{LE}$-algebra of formulas over $\mathsf{AtProp}$ to $\bba$. The notation $\mathbb{LE}\models\varphi\vdash\psi$ indicates that $\varphi\vdash\psi$ is valid in every LE of the appropriate signature. Then, by means of a routine Lindenbaum-Tarski construction, it can be shown that the minimal LE-logic $\mathbf{L}_\mathrm{LE}$ is sound and complete with respect to its corresponding class of algebras $\mathbb{LE}$, i.e.\ that any sequent $\varphi\vdash\psi$ is provable in $\mathbf{L}_\mathrm{LE}$ iff $\mathbb{LE}\models\varphi\vdash\psi$. %Moreover, it is not hard to see that every consistent LE-logic is characterized by the class of algebras for it.

\subsection{The fully residuated language $\mathcal{L}_\mathrm{LE}^*$}
\label{ssec:expanded language}
Any given language $\mathcal{L}_\mathrm{LE} = \mathcal{L}_\mathrm{LE}(\mathcal{F}, \mathcal{G})$ can be associated with the language $\mathcal{L}_\mathrm{LE}^* = \mathcal{L}_\mathrm{LE}(\mathcal{F}^*, \mathcal{G}^*)$, where $\mathcal{F}^*\supseteq \mathcal{F}$ and $\mathcal{G}^*\supseteq \mathcal{G}$ are obtained by expanding $\mathcal{L}_\mathrm{LE}$ with the following connectives: 
\begin{enumerate}
	\item the $n_f$-ary connective $f^\sharp_i$ for $0\leq i\leq n_f$, the intended interpretation of which is the right residual of $f\in\mathcal{F}$ in its $i$th coordinate if $\varepsilon_{f}(i) = 1$ (resp.\ its Galois-adjoint if $\varepsilon_{f}(i) = \partial$);
	\item the $n_g$-ary connective $g^\flat_i$ for $0\leq i\leq n_g$, the intended interpretation of which is the left residual of $g\in\mathcal{G}$ in its $i$th coordinate if $\varepsilon_{g}(i) = 1$ (resp.\ its Galois-adjoint if $\varepsilon_{g}(i) = \partial$).
	% $ g^\flat_j$ for each and $g\in \mathcal{G}$, where and $0\leq j\leq n_g$ ($f^\sharp_i$ is the right residual of $f$ in the $i$-th coordinate, and $g^\flat_j$ is the left residual of $g$ in the $j$-th coordinate).
	%\footnote{The adjoints of the unary connectives $\Box$, $\Diamond$, $\downarrow$ and $\uparrow$ are denoted $\Diamondblack$, $\blacksquare$, $\blhd$ and $\brhd$, respectively.}
\end{enumerate}
We stipulate that 
$f^\sharp_i\in\mathcal{G}^*$ if $\varepsilon_{f}(i) = 1$, and $f^\sharp_i\in\mathcal{F}^*$ if $\varepsilon_{f}(i) = \partial$. Dually, $g^\flat_i\in\mathcal{F}^*$ if $\varepsilon_{g}(i) = 1$, and $g^\flat_i\in\mathcal{G}^*$ if $\varepsilon_{g}(i) = \partial$. The order-type assigned to the additional connectives is predicated on the order-type of their intended interpretations. That is, for any $f\in \mathcal{F}$ and $g\in\mathcal{G}$, 
%each $g^\flat_j\in\mathcal{F}$, for each coordinate $i$ in $f$ or $g$,
	\begin{enumerate}
		\item if $\epsilon_f(i) = 1$, then $\epsilon_{f_i^\sharp}(i) = 1$ and $\epsilon_{f_i^\sharp}(j) = \epsilon_f^\partial(j)$ for any $j\neq i$.
		\item if $\epsilon_f(i) = \partial$, then $\epsilon_{f_i^\sharp}(i) = \partial$ and $\epsilon_{f_i^\sharp}(j) = \epsilon_f(j)$ for any $j\neq i$.
		\item if $\epsilon_g(i) = 1$, then $\epsilon_{g_i^\flat}(i) = 1$ and $\epsilon_{g_i^\flat}(j) = \epsilon_g^\partial(j)$ for any $j\neq i$.
		\item if $\epsilon_g(i) = \partial$, then $\epsilon_{g_i^\flat}(i) = \partial$ and $\epsilon_{g_i^\flat}(j) = \epsilon_g(j)$ for any $j\neq i$.
	\end{enumerate}

	For instance, if $f$ and $g$ are binary connectives such that $\varepsilon_f = (1, \partial)$ and $\varepsilon_g = (\partial, 1)$, then $\varepsilon_{f^\sharp_1} = (1, 1)$, $\varepsilon_{f^\sharp_2} = (1, \partial)$, $\varepsilon_{g^\flat_1} = (\partial, 1)$ and $\varepsilon_{g^\flat_2} = (1, 1)$.\footnote{Note that this notation depends on the connective which is taken as primitive, and needs to be carefully adapted to well known cases. For instance, consider the  `fusion' connective $\circ$ (which, when denoted  as $f$, is such that $\varepsilon_f = (1, 1)$). Its residuals
	$f_1^\sharp$ and $f_2^\sharp$ are commonly denoted $/$ and
	$\backslash$ respectively. However, if $\backslash$ is taken as the primitive connective $g$, then $g_2^\flat$ is $\circ = f$, and
	$g_1^\flat(x_1, x_2): = x_2/x_1 = f_1^\sharp (x_2, x_1)$. This example shows
	that, when identifying $g_1^\flat$ and $f_1^\sharp$, the conventional order of the coordinates is not preserved, and depends on which connective
	is taken as primitive.}

\begin{definition}\label{def:tense lattice logic}
	For any language $\mathcal{L}_\mathrm{LE}(\mathcal{F}, \mathcal{G})$, its associated basic $\mathcal{L}_\mathrm{LE}^\ast$-{\em logic} is defined by specializing Definition \ref{def:LE:logic:general} to the language $\mathcal{L}_\mathrm{LE}^* = \mathcal{L}_\mathrm{LE}(\mathcal{F}^*, \mathcal{G}^*)$ %a set of sequents $\varphi\vdash\psi$ with $\varphi,\psi\in\mathcal{L}_\mathrm{LE}^*$, which contains the axioms of the LE-logic $\mathbb{L}_\mathrm{LE}$, and is closed under rules for LE-logics plus
	and closing under the following additional residuation rules for $f\in \mathcal{F}$ and $g\in \mathcal{G}$:
	$$
	\begin{array}{cc}
	\AxiomC{$f(\varphi_1,\ldots,\varphi,\ldots, \varphi_{n_f}) \vdash \psi$}
	\doubleLine
	%			\RightLabel{$( = 1)$}
	\UnaryInfC{$\varphi\vdash^{\epsilon_{f}(i)} f^\sharp_i(\varphi_1,\ldots,\psi,\ldots,\varphi_{n_f})$}
	\DisplayProof
	&\qquad\qquad
	\AxiomC{$\varphi \vdash g(\varphi_1,\ldots,\psi,\ldots,\varphi_{n_g})$}
	\doubleLine
	%			\RightLabel{$(\epsilon_g(i) = 1)$}
	\UnaryInfC{$g^\flat_i(\varphi_1,\ldots, \varphi,\ldots, \varphi_{n_g})\vdash^{\epsilon_{g}(i)} \psi$}
	\DisplayProof
	\end{array}
	$$
	%			$$
	%			\begin{array}{cc}
	%			\AxiomC{$f(\varphi_1,\ldots,\varphi,\ldots, \varphi_{n}) \vdash \psi$}
	%			\doubleLine
	%			\RightLabel{$(\epsilon_f(i) = \partial)$}
	%			\UnaryInfC{$f^\sharp_i(\varphi_1,\ldots,\psi,\ldots,\varphi_{n})\vdash \varphi$}
	%			\DisplayProof
	%			&\quad
	%			\AxiomC{$\varphi \vdash g(\varphi_1,\ldots,\psi,\ldots,\varphi_{n})$}
	%			\doubleLine
	%			\RightLabel{($\epsilon_g(i) = \partial)$}
	%			\UnaryInfC{$\psi\vdash g^\flat_i(\varphi_1,\ldots, \varphi,\ldots, \varphi_{n})$}
	%			\DisplayProof
	%			\end{array}
	%			$$	
	The double line in each rule above indicates that the rule is invertible (i.e., bidirectional).
	Let $\mathbf{L}_\mathrm{LE}^*$ be the basic $\mathcal{L}^\ast_\mathrm{LE}$-logic.
	%PJ: This footnote can perhaps be omitted
	%\footnote{\label{ftn: definition basic logic for expanded language} Hence, for any language $\mathcal{L}_\mathrm{LE}$, there are in principle two logics associated with the expanded language $\mathcal{L}_\mathrm{LE}^*$, namely the {\em minimal} $\mathcal{L}_\mathrm{LE}^*$-logic, which we denote by $\mathbb{L}_\mathrm{LE}^{\underline{*}}$, and which is obtained by instantiating Definition \ref{def:LE:logic:general} to the language $\mathcal{L}_\mathrm{LE}^*$, and the fully residuated logic $\mathbb{L}_\mathrm{LE}^*$, defined above. The logic $\mathbb{L}_\mathrm{LE}^*$ is the natural logic on the language $\mathcal{L}_\mathrm{LE}^*$, however it is useful to introduce a specific notation for $\mathbb{L}_\mathrm{LE}^{\underline{*}}$, given that all the results holding for the minimal logic associated with an arbitrary LE-language can be instantiated to the expanded language $\mathcal{L}_\mathrm{LE}^*$ and will then apply to $\mathbb{L}_\mathrm{LE}^{\underline{*}}$.} For any LE-language $\mathcal{L}_{\mathrm{LE}}$, by a {\em fully residuated $\mathrm{LE}$-logic} we understand any axiomatic extension of the basic fully residuated  $\mathcal{L}_{\mathrm{LE}}$-logic in $\mathcal{L}^*_{\mathrm{LE}}$.
\end{definition}

The algebraic semantics of $\mathbf{L}_\mathrm{LE}^*$ is given by the class of  fully residuated $\mathcal{L}_\mathrm{LE}$-algebras, defined as tuples $\bba = (\mathbb{L}, \mathcal{F}^*, \mathcal{G}^*)$ such that $\mathbb{L}$ is a lattice and moreover,
\begin{enumerate}
	
	\item for every $f\in \mathcal{F}$ with $n_f\geq 1$, all $a_1,\ldots,a_{n_f},b\in L$ and $1\leq i\leq n_f$,
	$$f(a_1,\ldots,a_i,\ldots, a_{n_f})\leq b \quad \iff \quad a_i\leq^{\epsilon_f(i)} f^\sharp_i(a_1,\ldots,b,\ldots,a_{n_f}),$$
	\item for every $g\in \mathcal{G}$ with $n_g\geq 1$, all $a_1,\ldots,a_{n_g},b\in L$ and $1\leq i\leq n_g$,
	$$b\leq g(a_1,\ldots,a_i,\ldots, a_{n_g}) \quad \iff \quad g^\flat_i(a_1,\ldots,b,\ldots,a_{n_g})\leq^{\epsilon_g(i)} a_i.$$
\end{enumerate}

It is also routine to prove using the Lindenbaum-Tarski construction that $\mathbb{L}_\mathrm{LE}^*$ (as well as any of its canonical axiomatic extensions) is sound and complete with respect to the class of  fully residuated $\mathcal{L}_\mathrm{LE}$-algebras (or a suitably defined equational subclass, respectively). % and that every consistent LE$^*$-logic is characterized by its algebras.

\begin{thm}
	\label{th:conservative extension}%\marginnote{Edit the footnote, replace the BDL with LE and mention the reference to LE. Reference to footnote about perfect LEs.}
	The logic $\mathbf{L}_\mathrm{LE}^*$ is a conservative extension of $\mathbf{L}_\mathrm{LE}$, i.e.~every $\mathcal{L}_\mathrm{LE}$-sequent $\varphi\vdash\psi$ is derivable in $\mathbf{L}_\mathrm{LE}$ if and only if $\varphi\vdash\psi$ is derivable in $\mathbf{L}_\mathrm{LE}^*$. 
\end{thm}
\begin{proof}
	We only outline the proof.
	Clearly, every $\mathcal{L}_\mathrm{LE}$-sequent which is $\mathbf{L}_\mathrm{LE}$-derivable is also $\mathbf{L}_\mathrm{LE}^*$-derivable. Conversely, if an $\mathcal{L}_\mathrm{LE}$-sequent $\varphi\vdash\psi$ is not $\mathbf{L}_\mathrm{LE}$-derivable, then by the completeness of $\mathbf{L}_\mathrm{LE}$ with respect to the class of $\mathcal{L}_\mathrm{LE}$-algebras, there exists an $\mathcal{L}_\mathrm{LE}$-algebra $\bba$ and a variable assignment $v$ under which $\varphi^\bba\not\leq \psi^\bba$. Consider the canonical extension $\bba^\delta$ of $\bba$.\footnote{ \label{def:can:ext}
		The \emph{canonical extension} of a bounded lattice $L$ is a complete lattice $L^\delta$ with $L$ as a sublattice, satisfying
		\emph{denseness}: every element of $L^\delta$ can be expressed both as a join of meets and as a meet of joins of elements from $L$, and
		\emph{compactness}: for all $S,T \subseteq L$, if $\bigwedge S \leq \bigvee T$ in $L^\delta$, then $\bigwedge F \leq \bigvee G$ for some finite sets $F \subseteq S$ and $G\subseteq T$.
		It is well known that the canonical extension of $L$ is unique up to isomorphism fixing $L$ (cf.~e.g.~\cite[Section 2.2]{GehNagVen05}), and that the canonical extension is a \emph{perfect} bounded lattice, i.e.\ a complete  lattice which is completely join-generated by its completely join-irreducible elements and completely meet-generated by its completely meet-irreducible elements  (cf.~e.g.~\cite[Definition 2.14]{GehNagVen05})\label{canext bdl is perfect}. The canonical extension of an
		$\mathcal{L}_\mathrm{LE}$-algebra $\bbA = (L, \mathcal{F}^\bbA, \mathcal{G}^\bbA)$ is the perfect  $\mathcal{L}_\mathrm{LE}$-algebra
		$\bbA^\delta: = (L^\delta, \mathcal{F}^{\bbA^\delta}, \mathcal{G}^{\bbA^\delta})$ such that $f^{\bbA^\delta}$ and $g^{\bbA^\delta}$ are defined as the
		$\sigma$-extension of $f^{\bbA}$ and as the $\pi$-extension of $g^{\bbA}$ respectively, for all $f\in \mathcal{F}$ and $g\in \mathcal{G}$ (cf.\ \cite{sofronie2000duality1, sofronie2000duality2}).}
	Since $\bba$ is a subalgebra of $\bba^\delta$, the sequent $\varphi\vdash\psi$ is not satisfied in $\bba^\delta$ under the variable assignment $\iota \circ v$ ($\iota$ denoting the canonical embedding $\bba\hookrightarrow \bba^\delta$). Moreover, since $\bba^\delta$ is a perfect $\mathcal{L}_\mathrm{LE}$-algebra, it is naturally endowed with a structure of   $\mathcal{L}^\ast_\mathrm{LE}$-algebra. Thus, by the completeness of $\mathbf{L}_\mathrm{LE}^*$ with respect to the class of  $\mathcal{L}^\ast_\mathrm{LE}$-algebras, the sequent $\varphi\vdash\psi$ is not derivable in $\mathbf{L}_\mathrm{LE}^*$, as required.
\end{proof} 
Notice that the algebraic completeness of the logics $\mathbf{L}_\mathrm{LE}$ and $\mathbf{L}_\mathrm{LE}^*$ and the canonical embedding of LEs into their canonical extensions immediately give completeness of $\mathbf{L}_\mathrm{LE}$ and $\mathbf{L}_\mathrm{LE}^*$ with respect to the appropriate class of perfect LEs. 

\subsection{Analytic inductive LE-inequalities}\label{Inductive:Fmls:Section}
In this section we  recall the definitions of inductive LE-inequalities introduced in \cite{conradie2019algorithmic} and their corresponding `analytic' restrictions introduced in \cite{GMPTZ} in the distributive setting and then generalized to the setting of  LEs of arbitrary signatures in \cite{greco2018algebraic}.  Each inequality in any of these classes 
is canonical  and elementary (cf.~\cite[Theorems 7.1 and 6.1]{conradie2019algorithmic}).  %We will not give a direct proof that all inductive inequalities are elementary and canonical, but this will follow from the facts that they are all reducible by the {\sf ALBA}-algorithm and that all inequalities so reducible are elementary and canonical.

%\subsection{Inductive inequalities}\label{ssec:InductiveInequalities}
%A DLE-inequality %(resp.\ DLE$^*$-inequality)
%is an expression of the form $s\leq t$ where $s,t\in\mathcal{L}_\mathrm{DLE}$ %(resp.\ $s,t\in\mathcal{L}_\mathrm{DLE}^*$),
%which is essentially a sequent in algebraic form.

%	In the present subsection, we define the {\em inductive} $\mathcal{L}_\mathrm{LE}$-inequalities on which the algorithm ALBA defined in Section \ref{Spec:Alg:Section} will be shown to succeed. %equivalently transform into one (or the conjunction of more) pure quasi-inequalities in an expanded language. For more details,

\begin{definition}[\textbf{Signed Generation Tree}]
	\label{def: signed gen tree}
	The \emph{positive} (resp.\ \emph{negative}) {\em generation tree} of any $\mathcal{L}_\mathrm{LE}$-term $s$ is defined by labelling the root node of the generation tree of $s$ with the sign $+$ (resp.\ $-$), and then propagating the labelling on each remaining node as follows:
	\begin{itemize}
		%\item The root node $+s$ (resp.\ $-s$) is the root node of the positive (resp.\ negative) generation tree of $s$ signed with + (resp.\ $-$).
		\item For any node labelled with $ \lor$ or $\land$, assign the same sign to its children nodes.
		%\item If a node is labelled with $\lhd$, $\rhd$, assign the opposite sign to its child node.
		\item For any node labelled with $h\in \mathcal{F}\cup \mathcal{G}$ of arity $n_h\geq 1$, and for any $1\leq i\leq n_h$, assign the same (resp.\ the opposite) sign to its $i$th child node if $\varepsilon_h(i) = 1$ (resp.\ if $\varepsilon_h(i) = \partial$).
	\end{itemize}
	Nodes in signed generation trees are \emph{positive} (resp.\ \emph{negative}) if are signed $+$ (resp.\ $-$).
\end{definition}

Signed generation trees will be mostly used in the context of term inequalities $s\leq t$. In this context we will typically consider the positive generation tree $+s$ for the left-hand side and the negative one $-t$ for the right-hand side.\footnote{\label{footnote: precedent succedent} In the context of sequents $s\vdash t$, signed generation trees $+s$ and $-t$ can be also used to specify when subformulas of $s$ (resp.~$t$) occur in precedent or succedent position. Specifically, a given occurrence of  formula $\gamma$ is in {\em precedent} (resp.~{\em succedent}) position in $s\vdash t$ iff $+\gamma\prec +s$ or $+\gamma \prec -t$ (resp.~$-\gamma\prec +s$ or $-\gamma \prec -t$).} We will also say that a term-inequality $s\leq t$ is \emph{uniform} in a given variable $p$ if all occurrences of $p$ in both $+s$ and $-t$ have the same sign, and that $s\leq t$ is $\epsilon$-\emph{uniform} in a (sub)array $\overline{r}$ of its variables if $s\leq t$ is uniform in $r$, occurring with the sign indicated by $\epsilon$, for every $r$ in $\overline{r}$. \footnote{\label{footnote:uniformterms}If a term inequality $(s\leq t)[\overline{p}/!\overline{x}, \overline{q}]$ is $\epsilon$-uniform in all variables in $\overline{p}$ (cf.\ discussion after Definition \ref{def: signed gen tree}), then the validity of $s\leq t$ is equivalent to the validity of $(s\leq t)[\overline{\top^{\epsilon(i)}}/!\overline{x},\overline{q}]$, where $\top^{\epsilon(i)}=\top$ if $\epsilon(i)=1$ and $\top^{\epsilon(i)}=\bot$ if $\epsilon(i)=\partial$.}% \marginnote{G: I turned every $\vec{p}$ into $\obp$, given that later on we use this notation for vectors of variables.}

For any term $s(p_1,\ldots p_n)$, any order type $\epsilon$ over $n$, and any $1 \leq i \leq n$, an \emph{$\epsilon$-critical node} in a signed generation tree of $s$ is a leaf node $+p_i$ with $\epsilon(i) = 1$ or $-p_i$ with $\epsilon (i) = \partial$. An $\epsilon$-{\em critical branch} in the tree is a branch from an $\epsilon$-critical node. Variable occurrences corresponding to $\epsilon$-critical nodes are those used in the runs of the various versions of the algorithm ALBA (cf.~\cite{ConPal12, CoGhPa13, conradie2019algorithmic, lmcs:6694}) to compute the minimal valuations. 
For every term $s(p_1,\ldots p_n)$ and every order type $\epsilon$, we say that $+s$ (resp.\ $-s$) {\em agrees with} $\epsilon$, and write $\epsilon(+s)$ (resp.\ $\epsilon(-s)$), if every leaf in the signed generation tree of $+s$ (resp.\ $-s$) is $\epsilon$-critical.
%In other words, $\epsilon(+s)$ (resp.\ $\epsilon(-s)$) means that all variable occurrences corresponding to leaves of $+s$ (resp.\ $-s$) are  $\epsilon$-critical.
 We will also write $+s'\prec \ast s$ (resp.\ $-s'\prec \ast s$) to indicate that the subterm $s'$ inherits the positive (resp.\ negative) sign from the signed generation tree $\ast s$. Finally, we will write $\epsilon(\gamma) \prec \ast s$ (resp.\ $\epsilon^\partial(\gamma) \prec \ast s$) to indicate that the signed subtree $\gamma$, with the sign inherited from $\ast s$, agrees with $\epsilon$ (resp.\ with $\epsilon^\partial$). %\marginnote{G: Shall we emphasize this paragraph using Remark/Notation? Later on we refer to these conventions simply saying: ``see the discussion before Definition 6". Moreover, notice that we did not explain here what the expression ``are to be solved for'' means.}

\begin{notation}\label{notation: placeholder variables}
In what follows, we will often need to use {\em placeholder variables} to e.g.~specify the occurrence of a subformula within a given formula. In these cases, we will write e.g.~$\varphi(!z)$ (resp.~$\varphi(!\overline{z})$) to indicate that the variable $z$ (resp.~each variable $z$ in  vector $\overline{z}$) occurs exactly once in $\varphi$. Accordingly, we will write $\varphi[\gamma / !z]$  (resp.~$\varphi[\overline{\gamma}/!\overline{z}]$   to indicate the formula obtained from $\varphi$ by substituting $\gamma$ (resp.~each formula $\gamma$ in $\overline{\gamma}$) for the unique occurrence of (its corresponding variable) $z$ in $\varphi$. Also, in what follows, we will find sometimes useful to group placeholder variables together according to certain assumptions we make about them. So, for instance, we will sometimes write e.g.~$\varphi(!\overline{x}, !\overline{y})$ to indicate that $\epsilon(x) \prec \ast \varphi$ for all variables $x$ in  $\overline{x}$ and $\epsilon^\partial(y) \prec \ast \varphi$ for all variables $y$ in  $\overline{y}$, or we will write e.g.~$f(!\overline{x}, !\overline{y})$ to indicate that $f$ is monotone (resp.~antitone) in the coordinates corresponding to every variable $x$ in  $\overline{x}$ (resp.~$y$ in  $\overline{y}$).  We will provide further explanations as to the intended meaning of these groupings whenever required. Finally, we will also extend these conventions to inequalities or sequents, and thus write e.g.~$(\phi\leq \psi) [\overline{\gamma}/!\overline{z}, \overline{\delta}/!\overline{w}] $ to indicate  the inequality obtained from $\varphi\leq \psi$ by substituting each formula $\gamma$ in $\overline{\gamma}$ (resp.~$\delta$ in $\overline{\delta}$) for the unique occurrence of its corresponding variable $z$ (resp.~$w$) in $\varphi\leq \psi$.
\end{notation}

\begin{definition}
	\label{def:good:branch}
	Nodes in signed generation trees will be called \emph{$\Delta$-adjoints}, \emph{syntactically left residuals (SLR)}, \emph{syntactically right residuals (SRR)}, and \emph{syntactically right adjoints (SRA)}, according to the specification given in Table \ref{Join:and:Meet:Friendly:Table}.
	A branch in a signed generation tree $\ast s$, with $\ast \in \{+, - \}$, is called a \emph{good branch} if it is the concatenation of two paths $P_1$ and $P_2$, one of which may possibly be of length $0$, such that $P_1$ is a path from the leaf consisting (apart from variable nodes) only of PIA-nodes, and $P_2$ consists (apart from variable nodes) only of Skeleton-nodes. %\footnote{These classes are grouped together into the super-classes \emph{Skeleton} and \emph{PIA} as indicated in the table. This organization is motivated and discussed in \cite{CFPS} and \cite{CoGhPa13}.}
	%A branch is \emph{excellent} if it is good and in $P_1$ there are only SRA-nodes. 
	A good branch is \emph{Skeleton} if the length of $P_1$ is $0$, %(hence Skeleton branches are excellent), 
	and  is {\em SLR}, or {\em definite}, if  $P_2$ only contains SLR nodes.
	\begin{table}[h]
		\begin{center}
			\bgroup
			\def\arraystretch{1.2}%  1 is the default, change whatever you need
			\begin{tabular}{| c | c |}
				\hline
				Skeleton &PIA\\
				\hline
				$\Delta$-adjoints & Syntactically Right Adjoint (SRA) \\
				\begin{tabular}{ c c c c c c}
					$+$ &$\vee$ &\\
					$-$ &$\wedge$ \\
				\end{tabular}
				&
				\begin{tabular}{c c c c }
					$+$ &$\wedge$ &$g$ & with $n_g = 1$ \\
					$-$ &$\vee$ &$f$ & with $n_f = 1$ \\

				\end{tabular}
				\\ \hline
				Syntactically Left Residual (SLR) &Syntactically Right Residual (SRR)\\
				\begin{tabular}{c c c c }
					$+$ &  &$f$ & with $n_f \geq 1$\\
					$-$ &  &$g$ & with $n_g \geq 1$ \\
				\end{tabular}
				&\begin{tabular}{c c c c}
					$+$ & &$g$ & with $n_g \geq 2$\\
					$-$ &  &$f$ & with $n_f \geq 2$\\
				\end{tabular}
				\\
				\hline
			\end{tabular}
			\egroup
		\end{center}
		\caption{Skeleton and PIA nodes for $\mathrm{LE}$-languages.}\label{Join:and:Meet:Friendly:Table}
		\vspace{-1em}
	\end{table}
\end{definition}
We refer to \cite[Remark 3.3]{conradie2019algorithmic} and \cite{vanbenthem2005} for a discussion of the notational conventions and terminology. We refer to \cite[Section 3.2]{conradie2019algorithmic} and \cite[Section 1.7.2]{CoGhPa13} for a comparison with \cite{ConPal12} and \cite{GehNagVen05} where the nodes of the signed generation tree were classified according to the \emph{choice} and \emph{universal} terminology.	%\marginnote{il remark e' commentato; ci sono referenze da riportare alla luce. G: I checked the remark and added all the relevant references.}

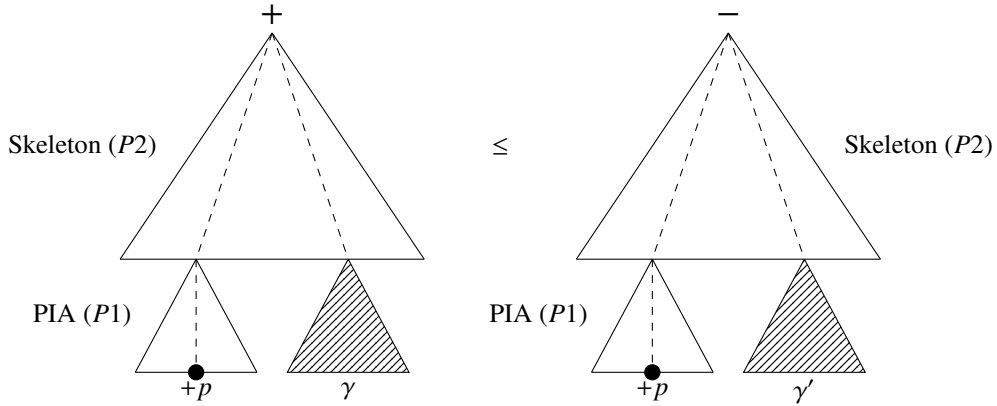
\begin{center}
	\begin{tikzpicture}
	\draw (-5,-1.5) -- (-3,1.5) node[above]{\Large$+$} ;
	\draw (-5,-1.5) -- (-1,-1.5) ;
	\draw (-3,1.5) -- (-1,-1.5);
	\draw (-5.5,0) node{Skeleton ($P2$)} ;
	\draw[dashed] (-3,1.5) -- (-4,-1.5);
	\draw[dashed] (-3,1.5) -- (-2,-1.5);
	\draw (-4,-1.5) --(-4.8,-3);
	\draw (-4.8,-3) --(-3.2,-3);
	\draw (-3.2,-3) --(-4,-1.5);
	\draw[dashed] (-4,-1.5) -- (-4,-3);
	\draw[fill] (-4,-3) circle[radius=.1] node[below]{$+p$};
	\draw
	(-2,-1.5) -- (-2.8,-3) -- (-1.2,-3) -- (-2,-1.5);
	\fill[pattern=north east lines]
	(-2,-1.5) -- (-2.8,-3) -- (-1.2,-3);
	\draw (-2,-3.25)node{$\gamma$};
	\draw (-5.5,-2.25) node{PIA ($P1$)} ;
	\draw (0,0) node{$\leq$};
	\draw (5,-1.5) -- (3,1.5) node[above]{\Large$-$} ;
	\draw (5,-1.5) -- (1,-1.5) ;
	\draw (3,1.5) -- (1,-1.5);
	\draw (5.5,0) node{Skeleton ($P2$)} ;
	\draw[dashed] (3,1.5) -- (4,-1.5);
	\draw[dashed] (3,1.5) -- (2,-1.5);
	\draw (2,-1.5) --(2.8,-3);
	\draw (2.8,-3) --(1.2,-3);
	\draw (1.2,-3) --(2,-1.5);
	\draw[dashed] (2,-1.5) -- (2,-3);
	\draw[fill] (2,-3) circle[radius=.1] node[below]{$+p$};
	\draw
	(4,-1.5) -- (4.8,-3) -- (3.2,-3) -- (4, -1.5);
	\fill[pattern=north east lines]
	(4,-1.5) -- (4.8,-3) -- (3.2,-3) -- (4, -1.5);
	\draw (4,-3.25)node{$\gamma'$};
	\draw (0.5,-2.25) node{PIA ($P1$)} ;
	\end{tikzpicture}
\end{center}

\begin{definition}[Inductive inequalities]\label{Inducive:Ineq:Def}
	For any order type $\epsilon$ and any irreflexive and transitive relation (i.e.\ strict partial order) $\Omega$ on $p_1,\ldots p_n$, the signed generation tree $*s$ $(* \in \{-, + \})$ of a term $s(p_1,\ldots p_n)$ is \emph{$(\Omega, \epsilon)$-inductive} if
	\begin{enumerate}
		\item for all $1 \leq i \leq n$, every $\epsilon$-critical branch with leaf $p_i$ is good (cf.\ Definition \ref{def:good:branch});
		\item every $m$-ary SRR-node occurring in the critical branch is of the form $ \circledast(\gamma_1,\dots,\gamma_{j-1},\beta,\gamma_{j+1}\ldots,\gamma_m)$, where for any $h\in\{1,\ldots,m\}\setminus j$: %$\gamma \in \{\gamma_1,\ldots,\gamma_{j-1},\gamma_{j+1},\ldots,\alpha_m\}$
		\begin{enumerate}
			\item $\epsilon^\partial(\gamma_h) \prec \ast s$ (cf.\ discussion before Definition \ref{def:good:branch}), and
			%\item $\epsilon^\partial(\ast \gamma)$, and
			%
			\item $p_k <_{\Omega} p_i$ for every $p_k$ occurring in $\gamma_h$ and for every $1\leq k\leq n$.
		\end{enumerate}
	\end{enumerate}
	
	We will refer to $<_{\Omega}$ as the \emph{dependency order} on the variables. An inequality $s \leq t$ is \emph{$(\Omega, \epsilon)$-inductive} if the signed generation trees $+s$ and $-t$ are $(\Omega, \epsilon)$-inductive. An inequality $s \leq t$ is \emph{inductive} if it is $(\Omega, \epsilon)$-inductive for some $\Omega$ and $\epsilon$.
\end{definition}

In what follows, we refer to formulas $\varphi$ such that only PIA nodes occur in $+\varphi$ (resp.\ $-\varphi$) as {\em positive} (resp.\ {\em negative}) {\em PIA-formulas}, and to formulas $\xi$ such that only Skeleton nodes occur in $+\xi$ (resp.\ $-\xi$) as {\em positive} (resp.\ {\em negative}) {\em Skeleton-formulas}\label{page: positive negative PIA}. PIA formulas $\ast \varphi$ in which no binary SRA-nodes (i.e.~$+\wedge$ and $-\vee$)
occur are referred to as {\em definite}. Skeleton formulas $\ast \xi$ in which no $\Delta$-adjoint nodes (i.e.~$-\wedge$ and $+\vee$) occur 
are referred to as {\em definite}. Hence, $\ast \xi$ (resp.~$\ast \varphi$) is definite Skeleton (resp.~definite PIA) iff  all nodes of $\ast \xi$ (resp.~$\ast \varphi$) are SLR (resp.~SRR or unary SRA).

\begin{lemma} 
\label{lemma: reduction to definite} For every LE-language $\mathcal{L}$,
\begin{enumerate}
\item if $\gamma$ is a positive PIA (resp.~negative Skeleton) $\mathcal{L}$-formula,  then $\gamma$ is equivalent to $\bigwedge_{i\in I}\gamma_i$ for some finite set of definite positive PIA (resp.~negative Skeleton) formulas $\gamma_i$;
\item if $\delta$ is a negative PIA (resp.~positive Skeleton) $\mathcal{L}$-formula,  then $\delta$ is equivalent to $\bigvee_{j\in j}\delta_j$ for some finite set of definite negative PIA (resp.~positive Skeleton) formulas $\delta_j$.
\end{enumerate}
\end{lemma}
\begin{proof}
By simultaneous induction on $\gamma$ and $\delta$. The base cases are immediately true. If $\delta: = f(\overline{\delta'}, \overline{\gamma'})$, then by induction hypothesis on each $\delta'$ in $\overline{\delta'}$ and each $\gamma'$ in $\overline{\gamma'}$, the formula $\delta$ is equivalent to $f(\overline{\bigvee_{j\in J}\delta'_j}, \overline{\bigwedge_{i\in I}\gamma'_i})$ for some finite sets of definite positive PIA (resp.~negative Skeleton) formulas $\gamma'_i$ and of definite positive Skeleton (resp.~negative PIA)  formulas $\delta'_j$.
 By the coordinatewise distribution properties of every $f\in \mathcal{F}$, the term $f(\overline{\bigvee_{j\in J}\delta'_j}, \overline{\bigwedge_{i\in I}\gamma'_i})$ is equivalent to  $\bigvee_{j\in J} \bigvee_{i\in I} f(\overline{\delta'_j}, \overline{\gamma'_i})$ with each $f(\overline{\delta'_j}, \overline{\gamma'_i})$ being definite positive Skeleton (resp.~negative PIA), as required. The remaining cases are omitted.
\end{proof}
%\begin{definition}\label{Sahlqvist:Ineq:Def}
%	For an order type $\epsilon$, the signed generation tree $\ast s$, $\ast \in \{-, + \}$, of a term $s(p_1,\ldots p_n)$ is \emph{$\epsilon$-Sahlqvist} if every $\epsilon$-critical branch is excellent. An inequality $s \leq t$ is \emph{$\epsilon$-Sahlqvist} if the trees $+s$ and $-t$ are both $\epsilon$-Sahlqvist.  An inequality $s \leq t$ is \emph{Sahlqvist} if it is $\epsilon$-Sahlqvist for some $\epsilon$.
%\end{definition}

%

\begin{definition}[Analytic inductive inequalities]
	\label{def:type5}
	For every order type $\epsilon$ and every irreflexive and transitive relation $\Omega$ on the variables $p_1,\ldots p_n$,
	the signed generation tree $\ast s$ ($\ast\in \{+, -\}$) of a term $s(p_1,\ldots p_n)$ is \emph{analytic $(\Omega, \epsilon)$-inductive}  if
	
	\begin{enumerate}%\marginnote{We still need to sort out the occurences of top and bottom}
		\item $\ast s$ is $(\Omega, \epsilon)$-inductive (cf.\ Definition \ref{Inducive:Ineq:Def});
		\item every branch of $\ast s$ is good (cf.\ Definition \ref{def:good:branch}).
	\end{enumerate}	
	
	An inequality $s \leq t$ is \emph{analytic $(\Omega, \epsilon)$-inductive}   if $+s$ and $-t$ are both analytic  $(\Omega, \epsilon)$-inductive. An inequality $s \leq t$ is \emph{analytic inductive} if is analytic $(\Omega, \epsilon)$-inductive  for some $\Omega$ and $\epsilon$. An analytic inductive inequality is {\em definite} if no $\Delta$-adjoint nodes (i.e.~$-\wedge$ and $+\vee$) occur in its Skeleton.
	\end{definition}	
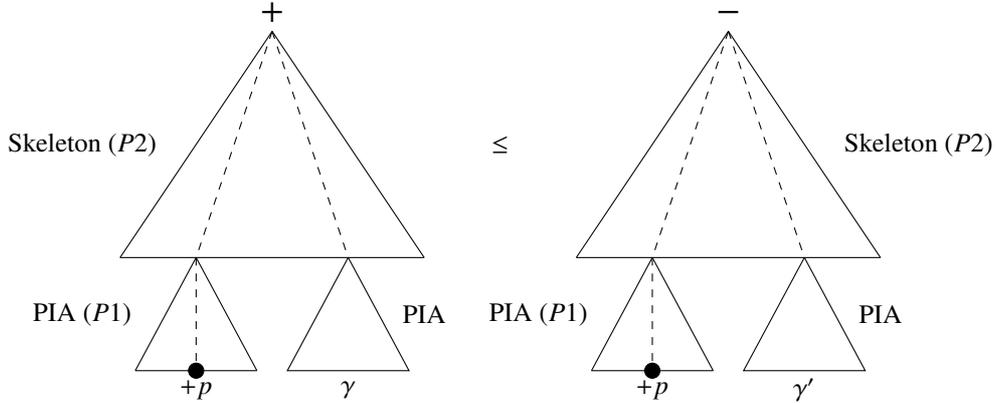
\begin{center}
	\begin{tikzpicture}
	\draw (-5,-1.5) -- (-3,1.5) node[above]{\Large$+$} ;
	\draw (-5,-1.5) -- (-1,-1.5) ;
	\draw (-3,1.5) -- (-1,-1.5);
	\draw (-5.5,0) node{Skeleton ($P2$)} ;
	\draw[dashed] (-3,1.5) -- (-4,-1.5);
	\draw[dashed] (-3,1.5) -- (-2,-1.5);
	\draw (-4,-1.5) --(-4.8,-3);
	\draw (-4.8,-3) --(-3.2,-3);
	\draw (-3.2,-3) --(-4,-1.5);
	\draw[dashed] (-4,-1.5) -- (-4,-3);
	\draw[fill] (-4,-3) circle[radius=.1] node[below]{$+p$};
	\draw
	(-2,-1.5) -- (-2.8,-3) -- (-1.2,-3) -- (-2,-1.5);
	%\fill[pattern=north east lines](-2,-1.5) -- (-2.8,-3) -- (-1.2,-3);
	\draw (-2,-3.25)node{$\gamma$};
	\draw (-5.5,-2.25) node{PIA ($P1$)} ;
	\draw (0,0) node{$\leq$};
	\draw (5,-1.5) -- (3,1.5) node[above]{\Large$-$} ;
	\draw (5,-1.5) -- (1,-1.5) ;
	\draw (3,1.5) -- (1,-1.5);
	\draw (5.5,0) node{Skeleton ($P2$)} ;
	\draw[dashed] (3,1.5) -- (4,-1.5);
	\draw[dashed] (3,1.5) -- (2,-1.5);
	\draw (2,-1.5) --(2.8,-3);
	\draw (2.8,-3) --(1.2,-3);
	\draw (1.2,-3) --(2,-1.5);
	\draw[dashed] (2,-1.5) -- (2,-3);
	\draw[fill] (2,-3) circle[radius=.1] node[below]{$+p$};
	\draw
	(4,-1.5) -- (4.8,-3) -- (3.2,-3) -- (4, -1.5);
	%\fill[pattern=north east lines](4,-1.5) -- (4.8,-3) -- (3.2,-3) -- (4, -1.5);
	\draw (4,-3.25)node{$\gamma'$};
	\draw (0.5,-2.25) node{PIA ($P1$)} ;
	\draw (-1,-2.25) node{PIA} ;
         \draw (5,-2.25) node{PIA} ;
	\end{tikzpicture}
\end{center}

%\marginnote{there are some inconsistencies or at the very least we need to explain what the colors stand for and why the alphas and betas are coloured and the gammas and deltas are only partially so, and the placeholder variables are not etc}
\begin{notation}\label{notation: analytic inductive}
	We will sometimes represent $(\Omega, \epsilon)$-analytic inductive inequalities/sequents as follows: \[(\varphi\leq \psi)[\oba/!\obx, \orb/!\ory,\overline{\bgamma}/!\obz, \overline{\rdelta}/!\orw] \quad\quad (\varphi\vdash \psi)[\oba/!\obx, \orb/!\ory,\overline{\bgamma}/!\obz, \overline{\rdelta}/!\orw],\] where $(\varphi\leq \psi)[!\obx, !\ory,!\obz, !\orw]$ is  the skeleton of the given inequality, $\overline{\alpha}$ (resp.~$\overline{\beta}$) denotes the vector of positive (resp.~negative) maximal PIA-subformulas, i.e.~each $\alpha$ in $\overline{\alpha}$ and $\beta$ in $\overline{\beta}$ contains at least one $\varepsilon$-critical occurrence of some propositional variable, and moreover:
	\begin{enumerate}
		\item for each $\alpha$ in $\overline{\alpha}$, either 
		$+\alpha\prec +\varphi$ or $+\alpha\prec -\psi$;
		\item for each $\beta$ in $\overline{\beta}$, either 
		$-\beta\prec +\varphi$ or $-\beta\prec -\psi$,
	\end{enumerate}
	and $\overline{\gamma}$ (resp.~$\overline{\delta}$) denotes  the vector of positive (resp.~negative) maximal $\varepsilon^{\partial}$-subformulas, and moreover:
	\begin{enumerate}		
		\item for each $\gamma$ in $\overline{\gamma}$, either $+\gamma\prec +\varphi$ or $+\gamma\prec -\psi$;
		\item for each $\delta$ in $\overline{\delta}$, either $-\delta\prec +\varphi$ or $-\delta\prec -\psi$.
	\end{enumerate}
	For the sake of a more compact notation, in what follows we sometimes write e.g.~$(\varphi\leq \psi)[\oba, \orb,\overline{\bgamma}, \overline{\rdelta}]$ in place of $(\varphi\leq \psi)$$[\oba/!\obx, \orb/!\ory,\overline{\bgamma}/!\obz, \overline{\rdelta}/!\orw]$.
	The colours are intended to help in identifying which subformula occurrences are in precedent (blue) or succedent (red) position (cf.~Footnote \ref{footnote: precedent succedent}).\footnote{The use of colours in this  notational convention is inspired by, but different from, the one introduced in \cite{GreRicMooTzi20}, where the blue (resp.~red) colour identifies the logical connectives algebraically  interpreted as right (resp.~left) adjoints or residuals. However, when restricted to the analytic inductive LE- inequalities, these two conventions coincide, since the main connective of a (non-atomic) positive (resp.~negative) maximal PIA-subformula is a right (resp.~left) adjoint/residual. Interestingly,  the so-called (strong) \emph{focalization} property of the focalized sequent calculi introduced in  \cite{GreRicMooTzi20} can be equivalently formulated in terms of maximal PIA-subtrees. }%\marginnote{G: I added this lat paragraph: it is not fundamental, but it makes our notational convention more meaningful (and it points to a different area of research that is surprisingly related to our analysis here). What do you think?\\
	\end{notation}
\begin{lemma}
\label{lemma: from inductive to definite sequents} For any LE-language $\mathcal{L}$,  any analytic inductive $\mathcal{L}$-sequent $(\varphi\vdash \psi)[\oba/!\obx, \orb/!\ory,\overline{\bgamma}/!\obz, \overline{\rdelta}/!\orw]$ is equivalent to the conjunction of definite analytic inductive $\mathcal{L}$-sequents $(\varphi_j\vdash \psi_i)[\oba/!\obx, \orb/!\ory,\overline{\bgamma}/!\obz, \overline{\rdelta}/!\orw]$.
\end{lemma}	
\begin{proof}
Since by assumption $\varphi(!\overline{x}, !\overline{y},!\overline{z}, !\overline{w})$ is positive Skeleton  and $\psi(!\overline{x}, !\overline{y},!\overline{z}, !\overline{w})$ is negative Skeleton, by Lemma \ref{lemma: reduction to definite}, the given sequent is equivalent to $(\bigvee_{j\in J}\varphi_j\vdash \bigwedge_{i\in I}\psi_i)[\oba/!\obx, \orb/!\ory,\overline{\bgamma}/!\obz, \overline{\rdelta}/!\orw]$, where every $\psi_i$ is definite negative Skeleton and  every $\varphi_j$ is definite positive Skeleton, from which the statement readily follows.
\end{proof}
	
%\marginnote{AP: ok, but better to move this notation closer to the first of these pictures. but I'm not sure that the best environment is the Notation environment, maybe we can add it to the captions of the pictures?}
\begin{notation}\label{notation: representations of signed generation trees}
We adopt the convention that in graphical representations of signed generation trees the squared variable occurrences are the $\epsilon$-critical ones, the doubly circled nodes are the Skeleton ones and the single-circle ones are PIA-nodes.
\end{notation}

	\begin{example}\label{Ex:Church-Rosser et al}
	\label{example:inductive and analytic inductive}
	Let $\mathcal{L}: = \mathcal{L}(\mathcal{F}, \mathcal{G})$, where $\mathcal{F}: = \{\Diamond\}$ and $\mathcal{G}: = \{\Box,\mor,\ararr\}$ with the usual arity and order-type.
The	$\mathcal{L}$-inequality $p\le \Diamond\Box p$ is $\epsilon$-Sahlqvist for $\epsilon(p) = 1$, but is not analytic inductive for any order-type, because the negative generation tree of $\Diamond \Box p$, which has only one branch, is not good. The Church-Rosser inequality $\Diamond \Box p \le \Box \Diamond p$ is analytic $\epsilon$-Sahlqvist for every order-type. 

The inequality $p \ararr (q \ararr r) \le ((p\ararr q) \ararr (\wbox p \ararr r))\mor \Diamond r$  is not Sahlqvist for any order-type: indeed, both the positive and the negative occurrence of $q$ occur under the scope of an SRR-connective. However, it is an analytic $(\Omega, \epsilon)$-inductive inequality, e.g.\ for $p <_\Omega q  <_\Omega r $ and $\epsilon(p,q,r)=(1,1,\partial)$. 

Below, we represent the signed generation trees pertaining to the inequalities above (see Notation \ref{notation: representations of signed generation trees}): %\marginnote{Jinsheng, I think that in order for all the pictures to be coherent, in the signed generation tree of $p\leq \Diamond \Box p$ the occurrrence of $+p$ needs to be into a square, not a circle. G: I fixed this. Moreover, I added a curly bracket pointing to $\alpha_p$?}
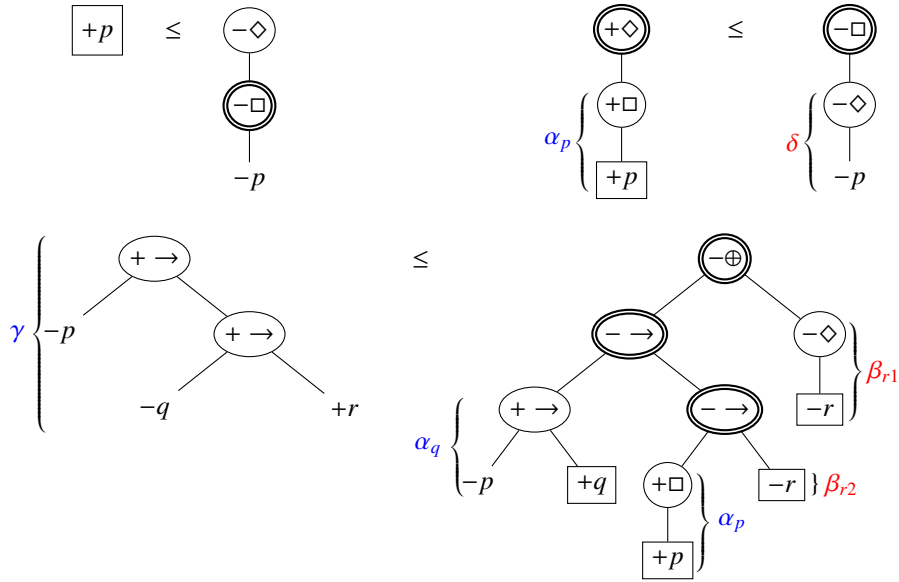
	\begin{center}
		\begin{tikzpicture}
%%%first inequality
		\node at(-3,0){
			\begin{tikzpicture}
			\tikzstyle{level 1}=[level distance=1cm, sibling distance=2.5cm]
			\tikzstyle{level 2}=[level distance=1cm, sibling distance=1.5cm]
			\tikzstyle{level 3}=[level distance=1 cm, sibling distance=1.5cm]
			\node[draw] at (-1,0) {$\begin{aligned} +p \end{aligned}$}
			;
			\node at (0,0) {$\le$}; 
			
			\node[PIA] at (1,0) {$\begin{aligned} -\wdia \end{aligned}$}
			child{node[Ske]{$\begin{aligned} -\wbox \end{aligned}$}
				child{node{$-p$}}
			}
			;
			\end{tikzpicture}
		};
%%%Curch-Rosser inequality
		\node at(4,0){
			\begin{tikzpicture}
			\tikzstyle{level 1}=[level distance=1cm, sibling distance=2.5cm]
			\tikzstyle{level 2}=[level distance=1cm, sibling distance=1.5cm]
			\tikzstyle{level 3}=[level distance=1 cm, sibling distance=1.5cm]
			\node[Ske] at (-1.5,0) {$\begin{aligned} +\wdia \end{aligned}$}
			child{node[PIA]{$\begin{aligned} +\wbox \end{aligned}$}
				child{node[draw]{$+p$}}
			}
			;
			\node at (0,0) {$\le$}; 
			
			\node[Ske] at (1.5,0) {$\begin{aligned} -\wbox \end{aligned}$}
			child{node[PIA]{$\begin{aligned} -\wdia \end{aligned}$}
				child{node{$-p$}}
			}
			;
			\node[rotate = -90] at (1, -1.5) {$\underbrace{\hspace{1.3cm}}$};
			\node at (0.75,-1.5) {$\rdelta$};
			% \draw[help lines] (-4,-4) grid (4,4);
			\node[rotate = -90] at (-2, -1.5) {$\underbrace{\hspace{1.3cm}}$};
			\node at (-2.33,-1.5) {$\textcolor{blue}{\alpha_p}$};
			% \draw[help lines] (-4,-4) grid (4,4);
			\end{tikzpicture}
		};
		%\caption{Signed generation tree for $\Box (p \rightarrow  q) \to \Box (\Box p\to\Box q)$}
		%\label{fig:fisher-servi}
		\node at (0.7,-4){
			\begin{tikzpicture}
			\tikzstyle{level 1}=[level distance=1cm, sibling distance=2.5cm]
			\tikzstyle{level 2}=[level distance=1cm, sibling distance=2.5cm]
			\tikzstyle{level 3}=[level distance=1 cm, sibling distance=1.5cm]
			\node[PIA] at (-3.5,0) {$\begin{aligned} +\ararr \end{aligned}$}
			child{node{$-p$}}          
			child{node[PIA]{$\begin{aligned} +\ararr \end{aligned}$}
				child{node{$-q$}}
				child{node{$+r$}}
			};
			\node at (0,0) {$\le$}; 
			
			\node[Ske] at (4,0) {$\begin{aligned} -\mor \end{aligned}$}
			child {node[Ske] {$\begin{aligned} -\ararr \end{aligned}$}
				child{node[PIA]{$\begin{aligned} +\ararr \end{aligned}$}
					child{node{$-p$}}
					child{node[draw]{$+q$}}
				}
				child{node[Ske]{$\begin{aligned} -\ararr \end{aligned}$}
					child{node[PIA]{$\begin{aligned} +\wbox \end{aligned}$}
						child{node[draw]{$+p$}}
					}
					child{node[draw]{$-r$}}
				}            
			}
			child{node[PIA]{$\begin{aligned} -\wdia \end{aligned}$}
				child{node[draw]{$-r$}}
			}
			;
			\node[rotate = +90] at (3.7, -3.5) {$\underbrace{\hspace{1.3cm}}$};
			\node at (4.1,-3.5) {$\textcolor{blue}{\alpha_p}$};
			\node[rotate = +90] at (5.7, -1.5) {$\underbrace{\hspace{1.3cm}}$};
			\node at (6.1,-1.5) {$\textcolor{red}{\beta_{r1}}$};
			\node[rotate = -90] at (0.5, -2.5) {$\underbrace{\hspace{1.3cm}}$};
			\node at (0.1,-2.5) {$\textcolor{blue}{\alpha_q}$};
			\node[rotate = -90] at (-5, -1) {$\underbrace{\hspace{2.6cm}}$};
			\node at (-5.3,-1) {$\bgamma$};
			
			\node at (5.4,-3) {$\ \}\ \textcolor{red}{\beta_{r2}}$};
			
			%\draw[help lines] (-4,-5) grid (6,6);
			%\node[rotate = +90] at (3.3, -2) {$\underbrace{\hspace{2.6cm}}$};
			%\node at (3.6,-2) {$\beta$};
			\end{tikzpicture}
			
		};
		
	\end{tikzpicture}
	%\caption{Signed generation tree for $\Box (p \rightarrow  q) \to \Box (\Box p\to\Box q)$}
	%\label{fig:fisher-servi}
\end{center}
%In the pictured inequalities above, the squared variable occurrences are the $\epsilon$-critical ones, the doubly circled nodes are Skeleton, and the single-circle ones are PIA. 		
\end{example}

The following auxiliary definition was introduced in \cite[Definition 48]{GMPTZ} as a simplified version of \cite[Definition 5.1]{CFPS}, and serves to calculate effectively the residuals of definite positive and negative PIA formulas (cf.\ \cite{GMPTZ}, discussion after Definition  \ref{Inducive:Ineq:Def}) w.r.t.\ a given variable occurrence $x$. The intended meaning of symbols such as $\varphi(!x, \oz)$ is that the variable $x$ occurs {\em exactly once} in the formula $\varphi$ (cf.~Notation \ref{notation: placeholder variables}). In the context of the following definition, the variable $x$ is used (and referred to) as the {\em pivotal} variable, i.e.~the variable that is displayed by effect of the recursive residuation procedure.

\begin{definition} 
\label{def: RA and LA}

For every definite positive PIA $\mathcal{L}_{\mathrm{LE}}$-formula $\psi = \psi(!x, \oz)$, and any definite negative PIA $\mathcal{L}_{\mathrm{LE}}$-formula $\xi = \xi(!x, \oz)$ such that $x$ occurs in them exactly once, the $\mathcal{L}^\ast_\mathrm{LE}$-formulas $\mathsf{la}(\psi)(u, \oz)$ and $\mathsf{ra}(\xi)(u, \oz)$ (for $u \in Var - (x \cup \oz)$) are defined by simultaneous recursion as follows: %\marginnote{AP:  %G: I would suggest to give a concrete example in a given signature immediately. Or, at least, we may refer to the specification we provide for distributive lattice operators and their residuals after example 14. In any case, what about connectives that are monotone in each coordinate or antitone in each coordinate? In this Def $x$ is arbitrary no? Is it fine like this or shall we use another variable (e.g. $w$)? I would also suggest to specify the meaning of $\overline{\psi_{-j}}$ before the equations and not after them. One more thing: I would suggest to mention that for denoting arbitrary connectives we put the monotone coordinates first... we need a convention here and I like this, but this is not the convention for the usual symbol $\ararr$ for instance, so it can be confusing.\\
%Please  remove colors, which we only use in the context of sequents.}

\begin{center}
	\begin{tabular}{r c l}
		
		$\mathsf{la}(x)$ &= &$u$;\\
		%	$\mathsf{LA}(\Box \varphi(x, \oz))$ &= &$\mathsf{LA}(\varphi)(\Diamondblack u, \overline{z})$;\\
		%	$\mathsf{LA}(\psi(\oz) \rightarrow \varphi(x, \oz))$ &= &$\mathsf{LA}(\varphi)(u \wedge \psi(\oz), \oz)$;\\
		%	$\mathsf{LA}(\varphi_1(\oz) \vee \varphi_2(x, \oz))$ &= &$\mathsf{LA}(\varphi_2)(u - \varphi_1(\oz), \oz)$;\\
		%	$\mathsf{LA}(\psi(x, \oz)\rightarrow\varphi(\oz))$ &= &$\mathsf{RA}(\psi)(u \rightarrow \varphi(\oz), \oz)$;\\
		$\mathsf{la}(g(\overline{\psi_{-j}(\oz)},\psi_j(x,\oz), \overline{\xi(\oz)}))$ &= &$\mathsf{la}(\psi_j)(g^{\flat}_{j}(\overline{\psi_{-j}(\oz)},u, \overline{\xi(\oz)} ), \oz)$;\\
		$\mathsf{la}(g(\overline{\psi(\oz)}, \overline{\xi_{-j}(\oz)},\xi_j(x,\oz)))$ &= &$\mathsf{ra}(\xi_j)(g^{\flat}_{j}(\overline{\psi(\oz)}, \overline{\xi_{-j}(\oz)},u), \oz)$;\\
		&&\\
		$\mathsf{ra}(x)$ &= &$u$;\\
		%$\mathsf{RA}(\Diamond \psi(x, \oz))$ &= &$\mathsf{RA}(\psi)(\blacksquare u, \overline{z})$;\\
		%$\mathsf{RA}(\psi(x, \oz) - \varphi(\oz))$ &= &$\mathsf{RA}(\psi)(\varphi(\oz) \vee u, \oz)$;\\
		%$\mathsf{RA}(\psi_1(\oz) \wedge \psi_2(x, \oz))$ &= &$\mathsf{RA}(\psi_2)(\psi_1(\oz) \rightarrow u, \oz)$;\\
		%$\mathsf{RA}(\psi(\oz) - \varphi(x, \oz))$ &= &$\mathsf{LA}(\varphi)(\psi(\oz) - u, \oz)$;\\
		$\mathsf{ra}(f(\overline{\xi_{-j}(\oz)},\xi_j(x,\oz), \overline{\psi(\oz)}))$ &= &$\mathsf{ra}(\xi_j)(f^{\sharp}_{j}(\overline{\xi_{-j}(\oz)},u, \overline{\psi(\oz)} ), \oz)$;\\
		$\mathsf{ra}(f(\overline{\xi(\oz)}, \overline{\psi_{-j}(\oz)},\psi_j(x,\oz)))$ &= &$\mathsf{la}(\psi_j)(f^{\sharp}_{j}(\overline{\xi(\oz)}, \overline{\psi_{-j}(\oz)},u), \oz)$.\\
	\end{tabular}
\end{center}
Above, symbols such as $\overline{\psi_{-j}}$ denote the vector obtained by removing the $j$th coordinate of the vector $\overline{\psi}$.
\end{definition}
\begin{example}\label{ex:adjunction} Let $\mathcal{L}: = \mathcal{L}(\mathcal{F}, \mathcal{G})$, where $\mathcal{F}: = \{\pdra\}$ and $\mathcal{G}: = \{\Box,\rightarrow\}$ with the usual arity and order-type. Let $\mathcal{F}^\ast: = \{\pdra, \Diamondblack, \otimes\}$ and $\mathcal{G}^\ast: = \{\Box,\rightarrow, \oplus\}$, where $\pdra$ (resp.~$\rightarrow$) has $\pdra'$ (resp.~$\rightarrow'$)  as residual in its first coordinate and $\oplus$ (resp.~$\otimes$) as residual in its second coordinate.
Consider the definite positive PIA $\mathcal{L}$-formula $\psi(!w, !y, !z): = \Box((w\pdra y) \rightarrow z)$. If  $x: = w$, then $\mathsf{la}(\psi)(u, !y, !z) = (\Diamondblack u\rightarrow' z)\pdra' y$. If $x: = y$, then $\mathsf{la}(\psi)(u, !w, !z) = (\Diamondblack u\rightarrow' z)\oplus w$. If $x: = z$, then $\mathsf{la}(\psi)(u, !w, !y) = (w\pdra y) \otimes\Diamondblack u$. 
\end{example}

\subsection{Display calculi for basic normal LE-logics} \label{ssec:syntactic frames associated w algebras}

In this section we define the proper display calculus $\mathrm{D.LE}$ for  the basic normal $\mathcal{L}_{\mathrm{LE}}$-logic in a fixed but arbitrary LE-signature $\mathcal{L} = \mathcal{L}(\mathcal{F}, \mathcal{G})$ (cf.~Section \ref{subset:language:algsemantics}.1). 
%In this section we let $\mathcal{L} = \mathcal{L}(\mathcal{F}, \mathcal{G})$ be a fixed but arbitrary LE-signature (cf.~Section \ref{subset:language:algsemantics}.1) and define the
 %display calculus $\mathrm{D.LE}$ for the basic normal $\mathcal{L}_{LE}$-logic and the display calculus $\mathrm{D.LE^\ast}$ for the fully residuated normal $\rbL^\ast_{LE}$ generated by the basic normal $\mathcal{L}_{LE}$-logic simultaneously. Their cut-free counterparts are denoted by $\mathrm{\cfDLE}$ and  $\mathrm{\cfDLE^\ast}$, respectively. %display calculus $\mathrm{D.LE^\ast}$ for the fully residuated normal $\rbL^\ast_{LE}$ generated by the basic normal $\mathcal{L}_{LE}$-logic and and its cut-free counterpart $\mathrm{\cfDLE^\ast}$. The display calculus $\mathrm{D.LE}$ for the basic normal $\mathcal{L}_{LE}$-logic is obtained from $\mathrm{D.LE^\ast}$ by deleting the logical (also called operational) rules for the operators that do not occur in the original LE-signature. 
 %The display calculus $\mathrm{D.LE^\ast}$ is obtained from $\mathrm{D.LE}$ by adding the logical (also called operational) rules for the operators that do not occur in the original LE-signature. 
 %As is usual of existing logical systems which the present framework intends to capture (e.g.~intuitionistic and bi-intuitionistic logics, or modal and tense logics \cite{gore1998substructural}), the languages manipulated by these calculi are built up using {\em one and the same} set of structural terms, and differ only in the set of operational (also called logical) term constructors. 
 Let $S_{\mathcal{F}}: = \{\FH \mid f\in \mathcal{F}^*\}$ and $S_{\mathcal{G}}: = \{\GC \mid g\in \mathcal{G}^*\}$ be the sets of structural connectives associated with  $\mathcal{F}^*$ and $ \mathcal{G}^*$ respectively (cf.~Section \ref{ssec:expanded language}). Each such structural connective comes with an arity and an order type which coincide with those of its associated operational connective in $ \mathcal{F}^*$ and $\mathcal{G}^*$.

\begin{remark} \label{rem: overloading the notation} If $f\in \mathcal{F}$ and $g\in \mathcal{G}$ form a dual pair,\footnote{Examples of dual pairs are $(\top, \bot)$, $(\wedge, \vee)$, $(\pdra, \pra)$, $(\pdla, \leftarrow)$, and $(\Diamond, \Box)$ where $\Diamond$ is defined as $\neg\Box\neg$.} %\footnote{The connectives $f\in \mathcal{F}$ and $g\in \mathcal{G}$ form a {\em dual pair} if $n_f = n_g = n\geq 1$, and for every DLE-relational structure in which $f$ and $g$ are interpreted by means of the $n+1$-ary relations $R$ and $S$ respectively, $R^{-1}[X_1,\ldots, X_n] = (S^{-1}[X_1^c,\ldots, X_n^c])^c$ for every $n$-tuple $(X_1,\ldots, X_n)$ of potential interpretants of proposition variables. For any set $W$ and any $n+1$-ary relation $R$ on $W$, we let $R^{-1}[X_1,\ldots, X_n]: = \{y\mid R(y, x_1,\ldots, x_n)$ for some $x_i\in X_i, 1\leq i\leq n\}$.},\marginnote{check definition of dual pair in the footnote. This is very cumbersome, but the thing is that I don't want to give it in terms of boolean negation, since e.g.\ it is also true in the distributive case}
	then $n_f = n_g$ and $\varepsilon_f = \varepsilon_g$. Then $f$ and $g$ can be assigned one and the same structural operator $H$, which is interpreted as $f$ when occurring in precedent position and as $g$ when occurring in succedent position (cf.~Footnote \ref{footnote: precedent succedent}):	
	\begin{center}
		\begin{tabular}{|r|c|c|}
			\hline
			\scriptsize{Structural symbols} & \mc{2}{c|}{$H$} \\
			\hline
			\scriptsize{Operational symbols} & $f$ & $g$ \\
			\hline
		\end{tabular}
	\end{center}
	Moreover, for any $1\leq i\leq n_f = n_g$, the residuals $f_i^\sharp$ and $g_i^\flat$ are dual to one another. Hence they can also be assigned one and the same structural connective as follows:
	
	\begin{center}
		\begin{tabular}{|r|c|c|c|c|}
			\hline
			\scriptsize{Order type}              & \mc{2}{c|}{$\varepsilon_f(i) = \varepsilon_g(i) = 1$} & \mc{2}{c|}{$\varepsilon_f(i) = \varepsilon_g(i) = \partial$} \\
			\hline
			\scriptsize{Structural symbols} & \mc{2}{c|}{$H_i$} & \mc{2}{c|}{$H_i$} \\
			\hline
			\scriptsize{Operational symbols} & \ \,\,$(g_i^\flat)$\ \,\, & $\rule[-1.2ex]{0pt}{0ex}(f_i^\sharp)\rule{0pt}{2.5ex}$ & \ \,\,$(f_i^\sharp)$\ \,\, & $(g_i^\flat)$\\
			\hline
		\end{tabular}
	\end{center}
	%Notice that ($f, g$) and ($g^\flat, f^\sharp$) are (coordinate-wise) \emph{dual pairs}, while the operators ($f, f^\sharp$) and ($g^\flat, g$) are (coordinatewise) \emph{residual pairs} as $f \dashv f^\sharp$ and $g^\flat \dashv g$.
 This observation has made it possible to associate one structural connective with two logical connectives, which has become common in the display calculi literature. In this paper, we prefer to maintain a strict one-to-one correspondence between operational and structural symbols. 
 
If we admit that the sets $\mathcal{F}$ and $\mathcal{G}$ have a non empty intersection (cf.~Footnote \ref{footnote: f and g operators}), then a unary connective $h \in \mathcal{F} \cap \mathcal{G}$ can be assigned one and the same structural operator $\tilde{h}$, which is interpreted as $h$ when occurring in precedent position and in succedent position:	
	\begin{center}
		\begin{tabular}{|r|c|c|}
			\hline
			\scriptsize{Structural symbols} & \mc{2}{c|}{$\rule[2.2ex]{0pt}{0ex}\tilde{h}$} \\
			%\scriptsize{Structural symbols} & \mc{2}{c|}{$H$} \\
			\hline
			\scriptsize{Operational symbols} & $h$ & $h$ \\
			\hline
		\end{tabular}
	\end{center}
\end{remark}
%The calculus $\mathrm{D.LE}$ manipulates $\mathcal{L}_\mathrm{LE}$-formulas defined below:
%\begin{align*}
%\mathcal{L}_\mathrm{LE}\ni \varphi \ & ::=\ p \mid \bot \mid \top \mid  \varphi \vee \varphi \mid \varphi \wedge \varphi \mid f (\varphi_1, \ldots, \varphi_{n_f}) \mid g (\varphi_1, \ldots, \varphi_{n_g}) 
%\end{align*}
%where $p$ is an atomic formula and $f \in \mathcal{F}$ and $g \in \mathcal{G}$.

For notational convenience, we let $\mathcal F^\partial:=\mathcal G$ and 
$\mathcal G^\partial:=\mathcal F$. Moreover, given the sets $\mathsf{Str}_\mathcal F$,
$\mathsf{Str}_\mathcal G$ defined below 
and any order type $\epsilon$ on $n$, we let $\mathsf{Str}_\mathcal{F}^{\epsilon} : = \prod_{i = 1}^{n}\mathsf{Str}_{\mathcal{F}^{\epsilon_i}}$ and $\mathsf{Str}_\mathcal{G}^{\epsilon} : = \prod_{i = 1}^{n}\mathsf{Str}_{\mathcal{G}^{\epsilon_i}}$. 

%, where for all $1 \leq i \leq n$,
%
%\begin{center}
%\begin{tabular}{ll}
%$\mathsf{Str}_\mathcal{F}^{\epsilon(i)} = \begin{cases} 
%\mathsf{Str}_\mathcal{F} &\mbox{ if } \epsilon(i) = 1\\
%\mathsf{Str}_\mathcal{G} &\mbox{ if } \epsilon(i) = \partial
%\end{cases}\quad$
%&
%$\mathsf{Str}_\mathcal{G}^{\epsilon(i)} = \begin{cases}
%\mathsf{Str}_\mathcal{G}& \mbox{ if } \epsilon(i) = 1,\\
%\mathsf{Str}_\mathcal{F} & \mbox{ if } \epsilon(i) = \partial.
%\end{cases}$
%\end{tabular}
%\end{center}
%The calculi $\mathrm{D.LE}$ manipulates $\mathcal{L}_\mathrm{LE}$-formulas:
%\begin{align*}
%\mathcal{L}_\mathrm{LE}\ni \varphi \ & ::=\ p \mid \bot \mid \top \mid  \varphi \vee \varphi \mid \varphi \wedge \varphi \mid f (\varphi_1, \ldots, \varphi_{n_f}) \mid g (\varphi_1, \ldots, \varphi_{n_g}) 
%\end{align*}
%where $p$ is an atomic formula and $f \in \mathcal{F}$ and $g \in \mathcal{G}$.

The calculus $\mathrm{D.LE}$ manipulates sequents $\Pi \vdash \Sigma$ where the structures $\Pi$ (for precedent) and $\Sigma$ (for succedent) are defined by the following simultaneous recursion: %\marginnote{AP: this introduces a bit of clash of notation with the way we use $\Gamma$ and $\Delta$ in section 3. I would suggest to write $\Gamma\vdash \Delta$ as $\Pi\vdash \Sigma$, where $\Pi$ stands for `precedent', and $\Sigma$ stands for `succedent', but then we need to find another capital Greek letter also for the neutral one (how about $\Upsilon$? Or shall we go back to\\ $X$ for precedent structures and\\ $Y$ for succedent ones? G: I implemented your first suggestion.}

\begin{center}
	\begin{tabular}{@{}r@{}l@{}}
	 	
		$\rule[-1.2ex]{0pt}{0ex}\mathsf{Str}_\mathcal{F} \ni \Pi$ \ & $ ::=\  \varphi \mid \AATOP \mid \FH\, (\overline{\Pi}^{(\varepsilon_f)})$  \\
		
		$\rule[-1.2ex]{0pt}{0ex}\mathsf{Str}_\mathcal{G} \ni \Sigma$ \ & $ ::=\  \varphi \mid \ABOT \mid \GC\, (\overline{\Sigma}^{(\varepsilon_g)}) $  \\
	\end{tabular}
\end{center}
with $\varphi\in \mathcal{L}_\mathrm{LE}$, and $\FH \in S_{\mathcal{F}}$, $\GC\in S_{\mathcal{G}}$, $\overline{\Pi}^{(\varepsilon_f)}\in \mathsf{Str}_\mathcal{F}^{\epsilon_f}$ and $\overline{\Sigma}^{(\varepsilon_g)} \in \mathsf{Str}_\mathcal{G}^{\epsilon_g}$. Notice that for any connective $h$ of arity $n \geq 1$ the notational convention $\hat{h}$ conveys also the information that $h$ is a left-adjoint/residual and the notational convention $\check{h}$ conveys the information that $h$ is a right-adjoint/residual. 
%with $\Delta^{1} = \Delta$ and $\Delta^{\partial} = \Gamma$. 

% for all $1 \leq i \leq n_f$,
%\[
%x^{\epsilon_f(i)} \in \begin{cases} 
%\mathsf{Str}_\mathcal{F}&\mbox{ if } \epsilon_f(i) = 1\\
%\mathsf{Str}_\mathcal{G} &\mbox{ if } \epsilon_f(i) = \partial
%\end{cases}
%\] 
%and for all $1 \leq i \leq n_g$,
%\[
%y^{\epsilon_g(i)} \in \begin{cases}
%\mathsf{Str}_\mathcal{G}& \mbox{ if } \epsilon_g(i) = 1,\\
%\mathsf{Str}_\mathcal{F} & \mbox{ if } \epsilon_g(i) = \partial.
%\end{cases}
%\]
%In what follows, we let $\bari{x} := (x_1, \ldots, x_{i-1}, x_{i+1},\ldots, x_n)$ and $\bari{x}_z := (x_1, \ldots, x_{i-1}, z, x_{i+1},\ldots, x_n)$. $\bari{y}$ and $\bari{y}_z$ are defined likewise. 

In what follows, we use $\Upsilon_1, \ldots, \Upsilon_n$ as structure metavariables in $\mathsf{Str}_\mathcal{F} \cup \mathsf{Str}_\mathcal{G}$. The introduction rules of the calculus below will guarantee that $\Upsilon \in \mathsf{Str}_\mathcal{F}$ whenever it occurs in precedent position, and $\Upsilon \in \mathsf{Str}_\mathcal{G}$ whenever it occurs in succedent position. 
The calculus $\mathrm{D.LE} = \mathrm{D.LE}_{\mathcal{L}}$ consists of the following rules:\footnote{For any LE-language $\mathcal{L}$, we will sometimes let $\mathrm{D.LE}^\ast :=  \mathrm{D.LE}_{\mathcal{L}^\ast}$, i.e.~we will let $\mathrm{D.LE}^\ast$ denote the calculus  obtained by instantiating the general definition of the basic calculus $\mathrm{D.LE}_{\mathcal{L}}$ to $\mathcal{L}: = \mathcal{L}^\ast$.}

\begin{itemize}
	\item Identity and cut rules:\footnote{In the display calculi literature, the identity rule is sometimes defined as $\varphi \fCenter \varphi$, where $\varphi$ is an arbitrary, possibly complex, formula. The difference is inessential, given that, in any display calculus, $p \fCenter p$ is an instance of $\varphi \fCenter \varphi$, and $\varphi \fCenter \varphi$ is derivable for any formula $\varphi$ whenever $p \fCenter p$ is the Identity rule.}
\end{itemize}
\begin{center}
	\begin{tabular}{rl}
		\AXC{\phantom{$\Gamma \fCenter \varphi$}}
		\LL{\fns Id}
		\UI$p \fCenter p$
		\DP
		& 
		\AX$\Pi \fCenter \varphi$
		\AX$\varphi \fCenter \Sigma$
		\RL{\fns Cut}
		\BI$\Pi \fCenter \Sigma$
		\DP
		\\
	\end{tabular}
\end{center}

%\item Display rules: 
%$$
%			\begin{array}{cc}
%			\AX$\FH (\obx) \fCenter y$
%			\doubleLine
%			\LL{$(\epsilon_{f,i} = 1)$}
%			\UI$x_i \fCenter \G_{f^\sharp_{i}}(\bari{x}_y)$
%			\DP
%			&
%			\quad
%			\AxiomC{$\Gamma \Rightarrow \GC(\ory)$}
%			\doubleLine
%			\RightLabel{$(\epsilon_{g,i} = 1)$}
%			\UnaryInfC{$\F_{g^\flat_{i}}(\bari{y}_x)\Rightarrow y_i$}
%			\DisplayProof
%			\end{array}
%			$$
%			$$
%			\begin{array}{cc}
%			\AxiomC{$\FH(\obx) \Rightarrow y$}
%			\doubleLine
%			\LeftLabel{$(\epsilon_{f,i} = \partial)$}
%			\UnaryInfC{$\F_{f^\sharp_{i}}(\bari{x}_y)\Rightarrow x_i$}
%			\DisplayProof
%			&
%			\quad
%			\AxiomC{$x \Rightarrow \GC(\ory)$}
%			\doubleLine
%			\RightLabel{($\epsilon_{g,i} = \partial)$}
%			\UnaryInfC{$y_i \Rightarrow \G_{g^\flat_{i}}(\bari{y}_x)$}
%			\DisplayProof
%			\end{array}
%			$$

\begin{itemize}		
	\item Display postulates for $f\in \mathcal{F}$ and $g\in \mathcal{G}$: for any $1\leq i,j\leq n_f$ and $1\leq h,k\leq n_g$,
\end{itemize}
\begin{itemize}
	\item[] If $\varepsilon_{f}(i) = 1$ and $\varepsilon_{g}(h) = 1$,
\end{itemize}
\begin{center}
	\begin{tabular}{@{}c@{}c@{}}
		\AXC{$\FH\, (\Upsilon_1, \ldots, \Pi_i, \ldots, \Upsilon_{n_f}) \fCenter \Sigma$}
		\doubleLine
		\LL{\fns $\FH \dashv \FCS_i$}
		\UIC{$\Pi_i \fCenter \FCS_i\, (\Upsilon_1, \ldots, \Sigma, \ldots, \Upsilon_{n_f})$}
		\DP
		& 
		\AXC{$\Pi \fCenter \GC\, (\Upsilon_1 \ldots, \Sigma_h, \ldots \Upsilon_{n_g})$}
		\doubleLine
		\RL{\fns $\GHF_h \dashv \GC$}
		\UIC{$\GHF_h\, (\Upsilon_1, \ldots, \Pi, \ldots, \Upsilon_{n_g}) \fCenter \Sigma_h$}
		\DP \\
	\end{tabular}
\end{center}
\begin{itemize}
	\item[] If $\varepsilon_{f}(j) = \partial$ and $\varepsilon_{g}(k) = \partial$,
\end{itemize}
\begin{center}
	\begin{tabular}{@{}c@{}c@{}}					
		\AX$\FH\, (\Upsilon_1, \ldots, \Sigma_j, \ldots, \Upsilon_{n_f}) \fCenter \Sigma$
		\doubleLine
		\LL{\fns $(\FH, \FHS_j)$}
		\UI$\FHS_j\, (\Upsilon_1, \ldots, \Sigma, \ldots, \Upsilon_{n_f}) \fCenter \Sigma_j$
		\DP
		& 
		\AX$\Pi \fCenter \GC\, (\Upsilon_1, \ldots, \Pi_k, \ldots, \Upsilon_{n_g})$
		\doubleLine
		\RL{\fns $(\GC, \GCF_k)$}
		\UI$ \Pi_k \fCenter \GCF_k\, (\Upsilon_1, \ldots, \Pi, \ldots, \Upsilon_{n_g})$
		\DP \\
	\end{tabular}
\end{center}

\begin{itemize}
	\item Structural rules for lattice connectives:
\end{itemize}
\begin{center}
	\begin{tabular}{rl}
		\AX$\AATOP \fCenter \Sigma$
		\LL{\fns $\aatop_W$}
		\UI$\Pi \fCenter \Sigma$
		\DP
		& 
		\AX$\Pi \fCenter \ABOT$
		\RL{\fns $\abot_W$}
		\UI$\Pi \fCenter \Sigma$
		\DP
	\end{tabular}
\end{center}

\begin{itemize}
	\item Logical introduction rules for lattice connectives:
\end{itemize}
\begin{center}
	\begin{tabular}{rl}
		\AX$\AATOP \fCenter \Sigma$
		\LL{\fns $\aatop_L$}
		\UI$\aatop \fCenter \Sigma$
		\DP
		\ 
		\AXC{$\phantom{\AATOP \fCenter \Sigma}$}
		\RL{\fns $\aatop_R$}
		\UI$\AATOP \fCenter \aatop$
		\DP
		& 
		\AXC{$\phantom{\Pi \fCenter \ABOT}$}
		\LL{\fns $\abot_L$}
		\UI$\abot \fCenter \ABOT$
		\DP
		\ 
		\AX$\Pi \fCenter \ABOT$
		\RL{\fns $\abot_R$}
		\UI$\Pi \fCenter \abot$
		\DP
		\\
		
		& \\
		
		\AX$\psi \fCenter \Sigma$
		\LL{\fns $\aand_{L2}$}
		\UI$\varphi \aand \psi \fCenter \Sigma$
		\DP
		\ 
		\AX$\varphi \fCenter \Sigma$
		\LL{\fns $\aand_{L1}$}
		\UI$\varphi \aand \psi \fCenter \Sigma$
		\DP
		& 
		\AX$\Pi \fCenter \varphi$
		\AX$\Pi \fCenter \psi$
		\RL{\fns $\aand_R$}
		\BI$\Pi \fCenter \varphi \aand \psi$
		\DP
		\\
		
		& \\
		
		\AX$\varphi \fCenter \Sigma$
		\AX$\psi \fCenter \Sigma$
		\LL{\fns $\aor_L$}
		\BI$\varphi \aor \psi \fCenter \Sigma$
		\DP
		& 
		\AX$\Pi \fCenter \varphi$
		\RL{\fns $\aor_{R1}$}
		\UI$\Pi \fCenter \varphi \aor \psi$
		\DP
		\ 
		\AX$\Pi \fCenter \psi$
		\RL{\fns $\aor_{R2}$}
		\UI$\Pi \fCenter \varphi \aor \psi$
		\DP
		\\
	\end{tabular}
	
\end{center}

\begin{itemize}
	\item Logical introduction rules for $f\in\mathcal{F}$ and $g\in\mathcal{G}$:
\end{itemize}
\begin{center}
	\begin{tabular}{c}
		\bottomAlignProof
		\AxiomC{$\Big(\Upsilon_i \fCenter \varphi_i \quad \varphi_j \fCenter \Upsilon_j \mid 1\leq i, j\leq n_f, \varepsilon_{f}(i) = 1\mbox{ and } \varepsilon_{f}(j) = \partial\Big)$}
		\RL{\fns$f_R$}
		\UI$\FH\, (\Upsilon_1,\ldots, \Upsilon_{n_f})\fCenter f(\varphi_1,\ldots, \varphi_{n_f})$
		\DP
	\end{tabular}
\end{center}
\begin{center}
	\begin{tabular}{c}
		\bottomAlignProof
		\AxiomC{$\Big( \varphi_i \fCenter \Upsilon_i \quad \Upsilon_j \fCenter \varphi_j \,\mid\, 1\leq i, j\leq n_g, \varepsilon_{g}(i) = 1\mbox{ and } \varepsilon_{g}(j) = \partial \Big)$}
		\LL{\fns$g_L$}
		\UI$g(\varphi_1,\ldots, \varphi_{n_g}) \fCenter \GC\, (\Upsilon_1,\ldots, \Upsilon_{n_g})$
		\DP
	\end{tabular}
\end{center}
\begin{center}
	\begin{tabular}{c c}
		\bottomAlignProof
		\AX$\FH\, (\varphi_1,\ldots, \varphi_{n_f}) \fCenter \Sigma$
		\LL{\fns$f_L$}
		\UI$f(\varphi_1,\ldots, \varphi_{n_f}) \fCenter \Sigma$
		\DP
		&
		\bottomAlignProof
		\AX$\Pi \fCenter \GC\, (\varphi_1,\ldots, \varphi_{n_g})$
		\RL{\fns$g_R$}
		\UI$\Pi \fCenter g(\varphi_1,\ldots, \varphi_{n_g})$
		\DP
	\end{tabular}
\end{center}

%where $i$-monotone ($j$-antitone) means that the operator is order preserving (order reversing) at the $i$-th ($j$-th) coordinate.
If $f$ and $g$ are $0$-ary (i.e.~they are constants), the rules $f_R$ and $g_L$ above reduce to the axioms ($0$-ary rule) $\FH \fCenter f$ and $g \fCenter \GC$.

\begin{rem}\label{rem: both f and g operators}
If we admit that the sets $\mathcal{F}$ and $\mathcal{G}$ have a non empty intersection (c.f.~Footnote \ref{footnote: f and g operators}), then the rules capturing a generic connective $h \in (\mathcal{F} \cap \mathcal{G})$ of arity $n = 1$ are as follows (notice that the notational convention $\tilde{h}$ conveys also the information that $h$ is both a left-adjoint and a right-adjoint):

\begin{itemize}		
	\item Display postulates for $h\in (\mathcal{F} \cap \mathcal{G})$ occurring in precedent and in succedent position:
\end{itemize}
\begin{itemize}
	\item[] If $\varepsilon_{h}(1) = 1$,
\end{itemize}
\begin{center}
	\begin{tabular}{@{}c@{}c@{}}
		\AX$\tilde{h}\, \Pi \fCenter \Sigma$
		\doubleLine
		\LL{\fns $\tilde{h} \dashv \check{h}^{\,\sharp}$}
		\UI$\Pi \fCenter \check{h}^{\,\sharp} \Sigma$
		\DP
		& 
		\AX$\Pi \fCenter \tilde{h}\, \Sigma$
		\doubleLine
		\RL{\fns $\hat{h}^{\,\flat} \dashv \tilde{h}$}
		\UI$\hat{h}^{\,\flat}\, \Pi \fCenter \Sigma$
		\DP \\
	\end{tabular}
\end{center}
\begin{itemize}
	\item[] If $\varepsilon_{h}(1) = \partial$,
\end{itemize}
\begin{center}
	\begin{tabular}{@{}c@{}c@{}}					
		\AX$\tilde{h}\, \Sigma_1 \fCenter \Sigma_2$
		\doubleLine
		\LL{\fns $(\hat{h}^{\,\sharp}, \tilde{h})$}
		\UI$\hat{h}^{\,\sharp}\, \Sigma_2 \fCenter \Sigma_1$
		\DP
		& 
		\AX$\Pi_1 \fCenter \tilde{h}\, \Pi_2$
		\doubleLine
		\RL{\fns $(\check{h}^{\,\flat}, \tilde{h})$}
		\UI$ \Pi_2 \fCenter \check{h}^{\,\flat}\, \Pi_1$
		\DP \\
	\end{tabular}
\end{center}

\begin{itemize}
	\item Structural rules for $h\in (\mathcal{F} \cap \mathcal{G})$:
\end{itemize}
%\begin{center}
%	\begin{tabular}{c}
%		\bottomAlignProof
%		\AxiomC{$\Big(\Upsilon_i \fCenter \Xi_i \quad \Xi_j \fCenter \Upsilon_j \mid 1\leq i, j\leq n_h, \varepsilon_{h}(i) = 1\mbox{ and } \varepsilon_{h}(j) = \partial\Big)$}
%		\RL{\fns$\tilde{h}$}
%		\UI$\tilde{h}\, (\Upsilon_1,\ldots, \Upsilon_{n_h}) \fCenter \tilde{h}\,(\Xi_1,\ldots, \Xi_{n_h})$
%		\DP
%	\end{tabular}
%\end{center}
%h(x_1,....h^?_j(x_1,..., y,... x_n),... x_n) |-- z / y|--z
%\begin{center}
%	\begin{tabular}{c}
		%\bottomAlignProof
		%\AX$\tilde{h} \, \hat{h}^{\,\flat}\, \Pi \fCenter \Sigma$
		%\LL{\fns$\tilde{h}, \hat{h}^{\,\flat}$}
		%\UI$\Pi \fCenter \Sigma$
		%\DP
		%& 
%		\bottomAlignProof
		%\AX$\tilde{h} \, \hat{h}^{\,\flat}\, \Pi \fCenter \Sigma$
%		\AX$\tilde{h} \, (\Upsilon_1, \ldots, \hat{h}_i^{\,\flat}\, (\Upsilon_1, \ldots, \Upsilon_i, \ldots, \Upsilon_n), \ldots, \Upsilon_n) \fCenter \Sigma $
%		\LL{\fns$\tilde{h}, \hat{h}^{\,\flat}$}
%		\UI$\Upsilon_i \fCenter \Sigma$
%		\DP
%		 \\
%	\end{tabular}
%\end{center}
\begin{itemize}
	\item[] If $\varepsilon_{h}(1) = 1$,
\end{itemize}
\begin{center}
	\begin{tabular}{cc}
		\AX$\Pi \fCenter \Sigma$
		\LL{\fns $\tilde{h}$}
		\UI$\tilde{h}\, \Pi \fCenter \tilde{h}\, \Sigma$
		\DP
		 & 
		\AX$\tilde{h}\, \hat{h}^{\,\flat}\, \Pi \fCenter \Sigma$
		\LL{\fns $(\tilde{h}, \hat{h}^{\,\flat})$}
		\UI$\Pi \fCenter \Sigma$
		\DP
		\\
	\end{tabular}
\end{center}
\begin{itemize}
	\item[] If $\varepsilon_{h}(1) = \partial$,
\end{itemize}
\begin{center}
	\begin{tabular}{cc}
		\AX$\Pi \fCenter \Sigma$
		\RL{\fns $\tilde{h}$}
		\UI$\tilde{h}\, \Sigma \fCenter \tilde{h}\, \Pi$
		\DP
		 & 
		\AX$\Pi \fCenter \tilde{h}\, \check{h}^{\,\sharp}\, \Sigma$
		\RL{\fns $(\tilde{h}, \check{h}^{\,\sharp})$}
		\UI$\Pi \fCenter \Sigma$
		\DP
		\\
	\end{tabular}
\end{center}

\begin{itemize}
	\item Logical introduction rules for $h\in (\mathcal{F} \cap \mathcal{G})$ occurring in precedent and in succedent position:
\end{itemize}
\begin{center}
	\begin{tabular}{cc}
		\bottomAlignProof
		\AX$\tilde{h}\, (\varphi_1,\ldots, \varphi_{n_h}) \fCenter \Sigma$
		\LL{\fns$h_L$}
		\UI$h(\varphi_1,\ldots, \varphi_{n_h}) \fCenter \Sigma$
		\DP
		&
		\bottomAlignProof
		\AX$\Pi \fCenter \tilde{h}\, (\varphi_1,\ldots, \varphi_{n_h})$
		\RL{\fns$h_R$}
		\UI$\Pi \fCenter h(\varphi_1,\ldots, \varphi_{n_h})$
		\DP
	\end{tabular}
\end{center}
\end{rem}

% \item Introduction rules for additional connectives:
%  \begin{displaymath}
%\mathrm{(f_L)}\ \frac{\FH (\overline{\varphi})\Rightarrow y}{f(\overline{\varphi})\Rightarrow y} \quad \mathrm{(f_R)}\ \frac{\Big(x^{\epsilon_{f,i}} \Rightarrow \varphi_i \quad \varphi_j \Rightarrow x^{\epsilon_{f,j}}  \mid 1\leq i, j\leq n_f, \varepsilon_{f,i} = 1\mbox{ and } \varepsilon_{f,j} = \partial \Big)}{\FH(\obx)\Rightarrow f(\overline{\varphi})}
%\end{displaymath}
% \begin{displaymath}
%\mathrm{(g_L)}\ \frac{x \Rightarrow \GC(\overline{\psi})}{x \Rightarrow g(\overline{\psi})} \quad \mathrm{(g_R)}\ \frac{\Big( \psi_i \Rightarrow y^{\epsilon_{g,i}}  \quad y^{\epsilon_{g,j}}  \Rightarrow \psi_j \,\mid\, 1\leq i, j\leq n_g, \varepsilon_{g,i} = 1\mbox{ and } \varepsilon_{g,j} = \partial \Big)}{g(\overline{\psi}) \Rightarrow \GC (\ory)}		
%\end{displaymath}	

Let $\mathrm{\cfDLE}$ %\marginnote{can we please have a cuter symbol for the cut-free calculus? this looks disgusting. In the following page we are changing again, we need to make these names uniform throughout the paper. G: Yes: it isa bit long: any suggestion? I introduced a macro \cfDLE. By the way, where are we changing the name of the calculus exactly?} 
denote the calculus obtained by removing Cut in $\mathrm{D.LE}$. In what follows, we indicate that the sequent $\varphi \vdash \psi$ is derivable in $\mathrm{D.LE}$ (resp.~in $\mathrm{\cfDLE}$) by $\vdash_{\mathrm{D.LE}} \varphi \vdash \psi$ (resp.~$\vdash_{\mathrm{\cfDLE}} \varphi \vdash \psi$).

%\marginnote{If needed, remember to define $D.LE^\ast$}

\begin{prop}[Soundness]\label{prop:soundness of D.LE}
	The calculus $\mathrm{D.LE}$ (hence also $\mathrm{\cfDLE}$) is sound w.r.t.~the class of  complete $\mathcal{L}$-algebras.
\end{prop}
\begin{proof}
	The soundness of the basic lattice rules is clear. The soundness of the remaining rules is due to the monotonicity (resp.~antitonicity) of the algebraic connectives interpreting each $f\in \mathcal{F}$ and $g\in \mathcal{G}$, and their adjunction/residuation properties, which hold since any complete $\mathcal{L}$-algebra is an $\mathcal{L}^*$-algebra.
\end{proof}
	%\marginnote{Here we might want to edit this part about removing cut, since it's not relevant. Here we might want to remove reference to the full display property, since we are not really defining the sufficient condition and it's not relevant. G: I agree: I commented this paragraph.}
\begin{comment}
	A display calculus enjoys the {\em full display property} (resp.\ the {\em relativized display property}) if for every (derivable) sequent $X \vdash Y$ and every substructure $Z$ of either $X$ or $Y$, the sequent $X \vdash Y$ can be equivalently transformed, using the rules of the system, into a sequent which is either of the form $Z \vdash W$ or of the form $W \vdash Z$, for some structure $W$. %In the first case, $Z$ occurs is \emph{displayed in precedent position}, and in the second case, $Z$ is \emph{displayed in succedent position}.
A routine check will show that the display calculus $\mathrm{D.LE}$ 
%and $\mathbf{DL}^*$ both 
enjoys the relativized display property, and moreover, if $\mathcal{F}$ and $\mathcal{G}$ are such that for every $f\in \mathcal{F}$ the dual of $f$ is in $\mathcal{G}$ and for every $g\in \mathcal{G}$ the dual of $g$ is in $\mathcal{F}$, then $\mathrm{D.LE}$
%and $\mathbf{DL}^*$ both 
enjoys the full display property. The proof of these facts is omitted.

\begin{prop}\label{prop: DL has relativized display property}
	The  calculus $ \mathrm{D.LE}$
%and $\mathbf{DL}^*$ 
enjoys the relativized display property, and under the assumption above on $\mathcal{F}$ and $\mathcal{G}$, it enjoys the full display property.
\end{prop}
\end{comment}
\begin{prop}\label{prop: DL has cut elim}
	
	The  calculus $\mathrm{D.LE}$
	% and $\mathbf{DL}^*$ 
	is a proper display calculus (cf.~\cite[Theorem 26]{GMPTZ}), and hence cut elimination holds for it as a consequence of a Belnap-style cut elimination meta-theorem (cf.~\cite[Section 2.2 and Appendix A]{GMPTZ} and \cite[Theorem 2]{GreLianMosPal2020}).	\end{prop}

\subsection{The setting of distributive LE-logics}	
	
In this section we discuss how the general setting presented above can account for  the assumption that the given LE-logic is distributive, i.e.\ that the {\em distributive laws} $(p\vee r)\wedge (p\vee q)\vdash p\vee (r\wedge q)$ and $ p\wedge (r\vee q) \vdash (p\wedge r)\vee (p\wedge q)$ are valid.	Such logics will be referred to as DLE-logics, since they are algebraically captured by varieties of {\em normal distributive lattice expansions} (DLEs), i.e.~LE-algebras as 
	%%%%%%%%%%%%%%%%%%%%%%%%%%%
 in Definition \ref{def:LE} such that $\mathbb{L}$ is assumed to be a \emph{bounded distributive lattice}. For any (D)LE-language, the basic %logic for distributive lattices 
 $\mathcal{L}_\mathrm{DLE}$-logic
 is defined as in Definition \ref{def:LE:logic:general} augmented with  the distributive laws above. %the axiom  $p\wedge (q\vee r)\vdash (p\wedge q)\vee (p\wedge r)$.

Since $\land$ and $\lor$ distribute over each other, besides being $\Delta$-adjoints, they can also be treated as elements of $\mathcal{F}$ and $\mathcal{G}$ respectively. In particular,  the binary connectives $\leftarrow$ and $\rightarrow$ occur  in the fully residuated language $\mathcal{L}_\mathrm{DLE}^*$, the intended interpretations of which are the right residuals of $\wedge$ in the first and second coordinate respectively, as well as the binary connectives $\pdla$ and $ \pdra$, the intended interpretations of which are the left residuals of $\vee$ in the first and second coordinate, respectively.
Following the general convention discussed in Section \ref{ssec:expanded language}, we stipulate that $\pdra, \pdla\in \mathcal{F}^*$ and $\rightarrow, \leftarrow \;  \in  \mathcal{G}^*$. The basic fully residuated $\mathcal{L}^\ast_\mathrm{DLE}$-{\em logic}, which will sometimes be referred to as the {\em basic bi-intuitionistic `tense'} DLE-logic, is given as per Definition \ref{def:tense lattice logic}. In particular, the residuation rules for the lattice connectives are specified as follows:\footnote{Notice that $\varphi\rightarrow \chi$ and $\chi\leftarrow \varphi$ are interderivable for any $\varphi$ and $\psi$, since $\wedge$ is commutative; similarly,  $\psi \pdra \varphi $ and $\varphi \pdla \psi $ are interderivable, since $\vee$ is commutative. Hence in what follows we consider explicitly only $\rightarrow$ and $\pdla$.}
%\begin{remark}
	%Note that ($\top, \bot$), ($\pand, \por$), ($\pdra, \pra$) and ($\pdla, \leftarrow$) are \emph{dual pairs}. Operators in a dual pair have the same arity and the same order type. The pairs ($\pand, \pra$), ($\wedge, \leftarrow$), ($\pdra, \por$) and ($\pdla,\vee$) are \emph{residual pairs} as follows: $\pand \dashv \ \pra$, $\pand \dashv \ \leftarrow$, $\pdra \dashv \por$, and $\pdla \dashv \por$.
	%\end{remark}
		$$
			\begin{array}{cccc}
			\AX$\varphi\wedge\psi \fCenter \chi$
			\doubleLine
			\UI$\psi \fCenter \varphi\rightarrow \chi$
			\DP
			&
			\AX$\psi\wedge\varphi \fCenter \chi$
			\doubleLine
			\UI$\psi \fCenter \chi\leftarrow \varphi$
			\DP
			&
			\AX$\varphi \fCenter \psi\vee\chi$
			\doubleLine
			\UI$\psi \pdra \varphi \fCenter \chi$
			\DP
			&
			\AX$\varphi \fCenter \chi\vee\psi$
			\doubleLine
			\UI$\varphi \pdla \psi \fCenter \chi$
			\DP
			\end{array}
			$$

When interpreting LE-languages on perfect distributive lattice expansions (perfect DLEs, cf.~Footnote \ref{def:can:ext}), the logical disjunction is interpreted by means of the coordinatewise completely $\wedge$-preserving join operation of the lattice, and the logical conjunction with the coordinatewise completely $\vee$-preserving meet operation of the lattice. Hence we are justified in listing $+\wedge$ and $-\vee$ among the SLRs, and $+\vee$ and $-\wedge$ among the SRRs, as is done in Table \ref{Join:and:Meet:Friendly:Table:DLE}. Consequently, the classes of ({\em analytic}) {\em inductive $\mathcal{L}_\mathrm{DLE}$-inequalities} are obtained
	%, which we will call the \emph{distributive Sahlqvist} and \emph{distributive inductive} inequalities respectively,
	by simply applying Definitions \ref{def:good:branch} and \ref{Inducive:Ineq:Def} with respect to Table \ref{Join:and:Meet:Friendly:Table:DLE} below. %\footnote{In the analogous table given in \cite{UnifiedCor}, the nodes $+\wedge$ and $-\vee$ are listed among the $\Delta$-adjoints. These definitions both yield the same syntactic classes, but cater for different algorithms. In particular, the algorithm illustrated in \cite{UnifiedCor} is based on \cite{ALBAPaper} and uses the splitting rules (the soundness of which is based on $\Delta$-adjunction) as part of the approximation phase. In contrast, in the present paper, approximation is taken care of by different rules which pivot exclusively on join- and meet-preservation or reversal.}
	\begin{table}[h]
		\begin{center}
			\begin{tabular}{| c | c |}
				\hline
				Skeleton &PIA\\
				\hline
				$\Delta$-adjoints & SRA \\
				\begin{tabular}{ c c c c c c}
					$\phantom{\wedge}$ &$+$ &$\vee$ &$\phantom{\lhd}$ & &\\
					$\phantom{\vee}$ &$-$ &$\wedge$ \\
					
				\end{tabular}
				&
				\begin{tabular}{c c c c }
					$+$ &$\wedge$ &$g$ & with $n_g = 1$ \\
					$-$ &$\vee$ &$f$ & with $n_f = 1$ \\
		
				\end{tabular}
				\\			\hline
				SLR &SRR\\
				\begin{tabular}{c c c c }
					$+$ & $\wedge$  &$f$ & with $n_f \geq 1$\\
					$-$ & $\vee$ &$g$ & with $n_g \geq 1$ \\
				\end{tabular}
				&\begin{tabular}{c c c c}
					$+$ & $\vee$ &$g$ & with $n_g \geq 2$\\
					$-$ & $\wedge$  &$f$ & with $n_f \geq 2$\\
				\end{tabular}
				\\
				\hline
			\end{tabular}
		\end{center}
		\vspace{-1em}
		\caption{Skeleton and PIA nodes for $\mathcal{L}_\mathrm{DLE}$.}\label{Join:and:Meet:Friendly:Table:DLE}
	\end{table}

Precisely because, as reported in Table \ref{Join:and:Meet:Friendly:Table:DLE}, the nodes $+\land$ and $-\lor$ are now also SLR nodes, and $+\lor$ and $-\land$ are also SRR nodes (see also Remark \ref{remark:distr}),  the classes of (analytic) inductive $\mathcal{L}_\mathrm{DLE}$-inequalities are strictly larger than the classes of (analytic) inductive $\mathcal{L}_\mathrm{LE}$-inequalities in the same signature, as shown in the next example.
\begin{example}\label{Ex:associativity}
The inequality $\Diamond \Box( p \lor q) \le \Box \Diamond p\lor \Box \Diamond q$ is not an  inductive $\mathcal{L}_\mathrm{LE}$-inequality for any order-type, but it is an $\epsilon$-Sahlqvist $\mathcal{L}_\mathrm{DLE}$-inequality e.g.~for $\epsilon(p,q)=(\partial,\partial)$. The classification of nodes in the signed generation trees of $\Diamond \Box( p \lor q) \le \Box \Diamond p\lor \Box \Diamond q$ as an $\mathcal{L}_\mathrm{DLE}$-inequality is on the left-hand side of the picture below, and the one as an $\mathcal{L}_\mathrm{LE}$-inequality is on the right (see Notation \ref{notation: representations of signed generation trees}). In the classification on the right, no branch is good, therefore $\Diamond \Box( p \lor q) \le \Box \Diamond p\lor \Box \Diamond q$ is not an inductive $\mathcal{L}_\mathrm{LE}$-inequality for any order-type. 

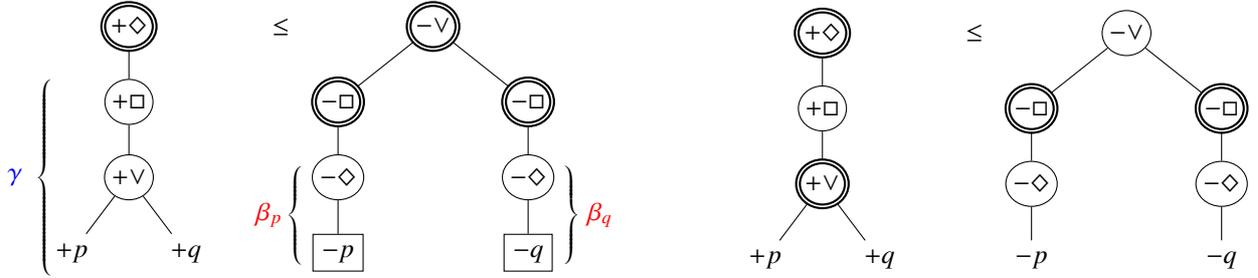
\begin{center}
\begin{tabular}{c}
\begin{tikzpicture}
		\node at(-4.5,0){
			\begin{tikzpicture}
			\tikzstyle{level 1}=[level distance=1cm, sibling distance=2.5cm]
			\tikzstyle{level 2}=[level distance=1cm, sibling distance=2.5cm]
			\tikzstyle{level 3}=[level distance=1 cm, sibling distance=1.5cm]
			\node[Ske] at (-2,0) {$\begin{aligned} +\Diamond \end{aligned}$}
			child{node[PIA]{$\begin{aligned}+\Box\end{aligned}$}
			child{node[PIA]{$\begin{aligned}+\lor\end{aligned}$}
			child{node{$+p$}}
			child{node{$+q$}}
			}
			}          
			;
			\node at (0,0) {$\le$}; 
			
			\node[Ske] at (2,0) {$\begin{aligned} -\lor \end{aligned}$}
			child{node[Ske]{$\begin{aligned} -\Box \end{aligned}$}
			child{node[PIA]{$\begin{aligned} -\Diamond \end{aligned}$}
			child{node[draw]{$-p$}}
			}
			}
			child{node[Ske]{$\begin{aligned} -\Box \end{aligned}$}
			child{node[PIA]{$\begin{aligned} -\Diamond \end{aligned}$}
			child{node[draw]{$-q$}}
			}
			}
			;
			
			\node[rotate = +90] at (3.8, -2.5) {$\underbrace{\hspace{1.3cm}}$};
			\node at (4.2,-2.5) {$\textcolor{red}{\beta_q}$};
			
			\node[rotate = -90] at (0.2, -2.5) {$\underbrace{\hspace{1.3cm}}$};
			\node at (-0.15,-2.5) {$\textcolor{red}{\beta_p}$};
			
			\node[rotate = -90] at (-3.1, -2) {$\underbrace{\hspace{2.6cm}}$};
			\node at (-3.5,-2) {$\bgamma$};
			
			%\draw[help lines] (-4,-5) grid (6,6);
			%\node[rotate = +90] at (3.3, -2) {$\underbrace{\hspace{2.6cm}}$};
			%\node at (3.6,-2) {$\beta$};
		\end{tikzpicture}
	};
	\node at(+4.5,0){
	\begin{tikzpicture}
		\tikzstyle{level 1}=[level distance=1cm, sibling distance=2.5cm]
		\tikzstyle{level 2}=[level distance=1cm, sibling distance=2.5cm]
		\tikzstyle{level 3}=[level distance=1 cm, sibling distance=1.5cm]
		\node[Ske] at (-2,0) {$\begin{aligned} +\Diamond \end{aligned}$}
		child{node[PIA]{$\begin{aligned}+\Box\end{aligned}$}
			child{node[Ske]{$\begin{aligned}+\lor\end{aligned}$}
				child{node{$+p$}}
				child{node{$+q$}}
			}
		}          
		;
		\node at (0,0) {$\le$}; 
		
		\node[PIA] at (2,0) {$\begin{aligned} -\lor \end{aligned}$}
		child{node[Ske]{$\begin{aligned} -\Box \end{aligned}$}
			child{node[PIA]{$\begin{aligned} -\Diamond \end{aligned}$}
				child{node{$-p$}}
			}
		}
		child{node[Ske]{$\begin{aligned} -\Box \end{aligned}$}
			child{node[PIA]{$\begin{aligned} -\Diamond \end{aligned}$}
				child{node{$-q$}}
			}
		}
		;
		
		%\node[rotate = +90] at (4.3, -2.5) {$\underbrace{\hspace{1.3cm}}$};
		%\node at (4.7,-2.5) {$\textcolor{red}{\beta_q}$};
		
		%\node[rotate = -90] at (0.7, -2.5) {$\underbrace{\hspace{1.3cm}}$};
		%\node at (0.35,-2.5) {$\textcolor{red}{\beta_p}$};
		
		%\node[rotate = -90] at (-3.6, -2) {$\underbrace{\hspace{2.6cm}}$};
		%\node at (-4,-2) {$\bgamma$};
		
		%\draw[help lines] (-4,-5) grid (6,6);
		%\node[rotate = +90] at (3.3, -2) {$\underbrace{\hspace{2.6cm}}$};
		%\node at (3.6,-2) {$\beta$};
	\end{tikzpicture}
	};
\end{tikzpicture}
\end{tabular}
\end{center}

The inequality $p \land  (q \lor r )\le q \lor ( p\land r)$ is an $\epsilon$-Sahlqvist $\mathcal{L}_\mathrm{DLE}$-inequality e.g.~for $\epsilon(p,q, r)=(1, 1, 1)$ but is not an inductive $\mathcal{L}_\mathrm{LE}$-inequality for any order-type.  The classification of nodes in the signed generation trees of $p \land  (q \lor r )\le q \lor ( p\land r)$ as an $\mathcal{L}_\mathrm{DLE}$-inequality is on the left-hand side of the picture below, and the one as an $\mathcal{L}_\mathrm{LE}$-inequality is on the right (see Notation \ref{notation: representations of signed generation trees}). The squared variable occurrences are the $\epsilon$-critical ones, the doubly circled nodes are Skeleton and the single-circle ones are PIA. In the classification on the right, no branch is good leading to occurrences of $r$, therefore $p \land  (q \lor r )\le q \lor ( p\land r)$ is not an inductive $\mathcal{L}_\mathrm{LE}$-inequality for any order-type.
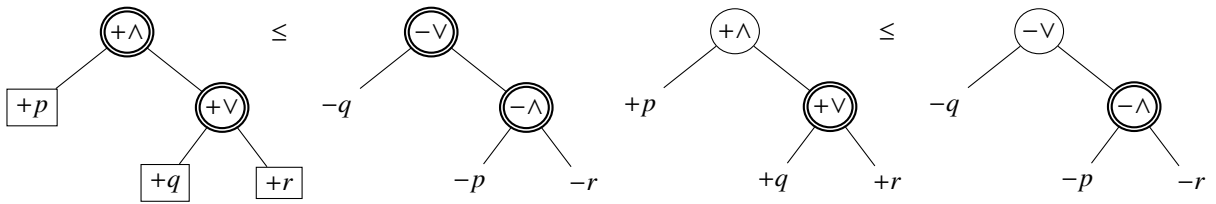
\begin{center}
\begin{tabular}{c}
	\begin{tikzpicture}
		\node at(-4,0){
			\begin{tikzpicture}
			\tikzstyle{level 1}=[level distance=1cm, sibling distance=2.5cm]
			\tikzstyle{level 2}=[level distance=1cm, sibling distance=1.5cm]
			\tikzstyle{level 3}=[level distance=1 cm, sibling distance=1.5cm]
			\node[Ske] at (-2,0) {$\begin{aligned} + \land \end{aligned}$}
child{node[draw]{$+p$}}
            child{node[Ske]{$\begin{aligned} +\lor \end{aligned}$}
            child{node[draw]{$+q$}}
               child{node[draw]{$+r$}}
               }
			;
			\node at (0,0) {$\le$}; 
			
			\node[Ske] at (2,0) {$\begin{aligned} -\lor \end{aligned}$}
child{node{$-q$}}
            child{node[Ske]{$\begin{aligned} -\land \end{aligned}$}
            child{node{$-p$}}
               child{node{$-r$}}
               }
			;
			\end{tikzpicture}
		};

	\node at(+4,0){
			\begin{tikzpicture}
			\tikzstyle{level 1}=[level distance=1cm, sibling distance=2.5cm]
			\tikzstyle{level 2}=[level distance=1cm, sibling distance=1.5cm]
			\tikzstyle{level 3}=[level distance=1 cm, sibling distance=1.5cm]
			\node[PIA] at (-2,0) {$\begin{aligned} + \land \end{aligned}$}
child{node{$+p$}}
            child{node[Ske]{$\begin{aligned} +\lor \end{aligned}$}
            child{node{$+q$}}
               child{node{$+r$}}
               }
			;
			\node at (0,0) {$\le$}; 
			
			\node[PIA] at (2,0) {$\begin{aligned} -\lor \end{aligned}$}
child{node{$-q$}}
            child{node[Ske]{$\begin{aligned} -\land \end{aligned}$}
            child{node{$-p$}}
               child{node{$-r$}}
               }
			;
			\end{tikzpicture}
		};

	\end{tikzpicture}
%\caption{Representations of example \ref{Ex:associativity} accordingly to notation \ref{notation: representations of signed generation trees}.}
%\vspace{-1em}
\end{tabular}
\end{center}
\end{example}

\begin{comment}		
\subsubsection{Inductive and analytic inductive  $\mathcal{L}_{\mathrm{DLE}}$-inequalities/sequents}
			\begin{table}
						\begin{center}
							\begin{tabular}{| c | c |}
								\hline
								Skeleton &PIA\\
								\hline
								$\Delta$-adjoints & SRA \\
								\begin{tabular}{ c c c c c c}
									$+$ &$\vee$ &$\wedge$ &$\phantom{\lhd}$ &\quad \ \ &\qquad\\
									$-$ &$\wedge$ &$\vee$&\quad \ \ &\qquad\\
								\end{tabular}
								&
								\begin{tabular}{c c c c }
									$+$ &$\wedge$ &$g$ & with $n_g = 1$ \\
									$-$ &$\vee$ &$f$ & with $n_f = 1$ \\
								\end{tabular}
								%
								\\ \hline
								SLR &SRR\\
								\begin{tabular}{c c c c }
									$+$ & $\wedge$ &$f$ & with $n_f \geq 1$\\
									$-$ & $\vee$ &$g$ & with $n_g \geq 1$ \\
								\end{tabular}
								%
								&
								\begin{tabular}{c c c c}
									$+$ &$\vee$ &$g$ & with $n_g \geq 2$\\
									$-$ & $\wedge$ &$f$ & with $n_f \geq 2$\\
								\end{tabular}
								\\
								\hline
							\end{tabular}
						\end{center}
						\caption{Skeleton and PIA nodes for $\mathrm{DLE}$.}\label{Join:and:Meet:Friendly:Table}
						\vspace{-1em}
					\end{table}
\end{comment}	
Also, definite  Skeleton and definite PIA $\mathcal{L}_{\mathrm{DLE}}$-formulas are defined verbatim in the same way as in the setting of $\mathcal{L}_{\mathrm{LE}}$-formulas. Namely, $\ast \xi$ (resp.~$\ast \varphi$) is definite Skeleton (resp.~definite PIA) iff  all nodes of $\ast \xi$ (resp.~$\ast \varphi$) are SLR (resp.~SRR).	However, the classification of nodes we need to consider is now the one of Table  \ref{Join:and:Meet:Friendly:Table:DLE}, where $+\wedge$ and $-\vee$ are also SLR-nodes, and $+\vee$ and $-\wedge$ are also SRR-nodes.			
Definition \ref{def: RA and LA} is specified for $\land$, $\lor$, $\rightarrow$ and $\pdla$ as follows:
			%\marginnote{AP:I think here it makes sense to use $\alpha$ (with subscripts) to denote the positive PIA formulas and $\beta$ (with subscripts) to denote the negative PIA formulas}
			%\marginnote{AP: please remove colors}
			\begin{center}
				\begin{tabular}{r c l}
					%$\mathsf{LA}^{\varepsilon}(\top)$ &= &$\bot$ \\
				%	$\mathsf{LA}(x)$ &= &$u$;\\
					%$\mathsf{LA}^{\varepsilon}(x_j)$ &= &$\bot^{\varepsilon_j}$ when $i \neq j$;\\
					%$\mathsf{LA}(\Box \varphi(x, \oz))$ &= &$\mathsf{LA}(\varphi)(\Diamondblack u, \overline{z})$;\\
					%$\mathsf{LA}(\varphi_1(x, \oz) \wedge \varphi_2(\oz))$ &= &$\mathsf{LA}(\varphi_1)(u, \overline{z})$;\\
					$\mathsf{la}(\xi(\oz) \rightarrow \psi(x, \oz))$ &= &$\mathsf{la}(\psi)(u \wedge \xi(\oz), \oz)$;\\
					$\mathsf{la}(\psi_1(\oz) \vee \psi_2(x, \oz))$ &= &$\mathsf{la}(\psi_2)(u \pdla \psi_1(\oz), \oz)$;\\
					$\mathsf{la}(\xi(x, \oz) \rightarrow \psi(\oz))$ &= &$\mathsf{ra}(\xi)(u \rightarrow \psi(\oz), \oz)$;\\
					%$\mathsf{LA}(g(\overrightarrow{\varphi_{-j}(\oz)},\varphi_j(x,\oz), \overrightarrow{\psi(\oz)}))$ &= %&$\mathsf{LA}(\varphi_j)(g^{\flat}_{j}(\overrightarrow{\varphi_{-j}(\oz)},u, \overrightarrow{\psi(\oz)} ), \oz)$;\\
					%$\mathsf{LA}(g(\overrightarrow{\varphi(\oz)}, \overrightarrow{\psi_{-j}(\oz)},\psi_j(x,\oz)))$ &= %&$\mathsf{RA}(\psi_j)(g^{\flat}_{j}(\overrightarrow{\varphi(\oz)}, \overrightarrow{\psi_{-j}(\oz)},u), \oz)$;\\
					&&\\
					%$\mathsf{RA}^{\varepsilon}(\bot)$ &= &$\top$ \\
					%$\mathsf{RA}(x)$ &= &$u$;\\
					%$\mathsf{RA}^{\varepsilon}(x_j)$ &= &$\top^{\varepsilon_j}$ when $i \neq j$;\\
					%$\mathsf{RA}(\Diamond \psi(x, \oz))$ &= &$\mathsf{RA}(\psi)(\blacksquare u, \overline{z})$;\\
					%$\mathsf{RA}(\psi_1(x, \oz) \vee \psi_2(\oz))$ &= &$\mathsf{RA}(\psi_1)(u, \overline{z})$;\\
					$\mathsf{ra}(\xi(x, \oz) \pdla \psi(\oz))$ &= &$\mathsf{ra}(\xi)(\psi(\oz) \vee u, \oz)$;\\
					$\mathsf{ra}(\xi_1(\oz) \wedge \xi_2(x, \oz))$ &= &$\mathsf{ra}(\xi_2)(\xi_1(\oz) \rightarrow u, \oz)$;\\
					$\mathsf{ra}(\xi(\oz) \pdla \psi(x, \oz))$ &= &$\mathsf{la}(\psi)(\xi(\oz) \pdla u, \oz)$;\\
					%$\mathsf{RA}(f(\overrightarrow{\psi_{-j}(\oz)},\psi_j(x,\oz), \overrightarrow{\varphi(\oz)}))$ &= &$\mathsf{RA}(\psi_j)(f^{\sharp}_{j}(\overrightarrow{\psi_{-j}(\oz)},u, \overrightarrow{\varphi(\oz)} ), \oz)$;\\
					%$\mathsf{RA}(f(\overrightarrow{\psi(\oz)}, \overrightarrow{\varphi_{-j}(\oz)},\varphi_j(x,\oz)))$ &= &$\mathsf{LA}(\varphi_j)(f^{\sharp}_{j}(\overrightarrow{\psi(\oz)}, \overrightarrow{\varphi_{-j}(\oz)},u), \oz)$.\\
				\end{tabular}
			\end{center} %\marginnote{AP: put all operational arrows in round brackets}
Finally, as to  the display calculus $\mathrm{D.DLE}$ for the basic $\mathcal{L}_{\mathrm{DLE}}$-logic,  its language is obtained by augmenting the language of $\mathrm{D.LE}$ with the following structural symbols for the lattice operators and their residuals:\footnote{In the presence of the exchange rules $E_L$ and $E_R$, the structural connectives $\ADRARR, \ALARR$ and the corresponding operational connectives $\adrarr, \alarr$ are redundant. For simplicity, we consider languages and calculi where the operational connectives $\adrarr$ and $\alarr$ and their introduction rules are not included.}
						
%\marginnote{Now we have a one-one correspondence, but it would be useful to mention that we overloaded the notation in some previous paper. Moreover, we should mention that $\hat{\cdot}$ is associated to left-adjoints/residuals and $\check{\cdot}$ to right adjoint/residuals. G: we mention that it is possible to overload the notation in Remark 11. I also added a sentence on page 9 (immediately after the definition of the structural language) about the notational conventions $\hat{\cdot}$ and $\check{\cdot}$.}
			\begin{center}
				\begin{tabular}{|r|c|c|c|c|c|c|c|c|}
					\hline
					\scriptsize{Structural symbols} & $\ATOP$ & $\ABOT$ & $\AAND$ & $\AOR$ & $\ADRARR$ & $\ARARR$ & $\ADLARR$ & $\ALARR$ \\
					\hline
					\scriptsize{Operational symbols} & $\top$ & $\bot$ & $\pand$ & $\por$ & $(\pdra)$ & $(\pra)$ & $(\pdla)$ & $(\leftarrow)$\\
					\hline
				\end{tabular}
			\end{center}			

		%A {\em structure} is an expression built from $\mathcal{L}_\mathrm{DLE}^*$-terms via structural operators. We use capital letters $X, Y, Z$ to denote structures. A sequent is of the form $X\vdash Y$ where $X,Y$ are structures. $X$ is called {\em antecedent}, and $Y$ the {\em succedent}.
		
\noindent Display postulates for lattice connectives and their residuals are specified as follows:
				\begin{center}
					\begin{tabular}{c}
						\AX$\Pi_1 \AAND \Pi_2 \fCenter \Sigma$
						\LL{\fns$\hat{\wedge} \dashv \check{\ararr}$}
						\doubleLine
						\UI$\Pi_2 \fCenter \Pi_1 \ARARR \Sigma$
						\DP
						\ \ 
						\AX$\Pi_1 \AAND \Pi_2 \fCenter \Sigma$
						\LL{\fns$\hat{\wedge} \dashv \check{\alarr}$}
						\doubleLine
						\UI$\Pi_1 \fCenter \Sigma \ALARR \Pi_2$
						\DP
						\ \ 
						\AX$\Pi \fCenter \Sigma_1 \AOR \Sigma_2$
						\RL{\fns$\hat{\adrarr} \dashv \check{\aor}$}
						\doubleLine
						\UI$\Sigma_1 \ADRARR \Pi \fCenter \Sigma_2$
						\DP
						 \ \ 
						\AX$\Pi \fCenter \Sigma_1 \AOR \Sigma_2$
						\RL{\fns$\hat{\adlarr} \dashv \check{\aor}$}
						\doubleLine
						\UI$\Pi \ADLARR \Sigma_2 \fCenter \Sigma_1$
						\DP
					\end{tabular}
				\end{center}

%%%
\noindent Moreover, $\mathrm{D.DLE}$ is augmented with the following structural rules encoding the characterizing properties of the  lattice connectives:
				\begin{center}
					\begin{tabular}{rlcrl}
						\AX$\Pi \fCenter \Sigma$
						\doubleLine
						\LL{\fns$\ATOP_{L}$}
						\UI$\ATOP \AAND \Pi \fCenter Y$
						\DP
						&
						\AX$\Pi \fCenter \Sigma$
						\doubleLine
						\RL{\fns$\ABOT_{R}$}
						\UI$\Pi \fCenter \Sigma \AOR \ABOT$
						\DP
						& &
						\AX$\Pi_1 \AAND \Pi_2 \fCenter \Sigma$
						\LL{\fns$E_L$}
						\UI$\Pi_2 \AAND \Pi_1 \fCenter \Sigma$
						\DP
						&
						\AX$\Pi \fCenter \Sigma_1 \AOR \Sigma_2$
						\RL{\fns$E_R$}
						\UI$\Pi \fCenter \Sigma_2 \AOR \Sigma_1$
						\DP
						\\
						&\\
						\AX$\Pi_2 \fCenter \Sigma$
						\LL{\fns$W_L$}
						\UI$\Pi_1 \AAND \Pi_2 \fCenter \Sigma$
						\DP
						&
						\AX$\Pi \fCenter \Sigma_1$
						\RL{\fns$W_R$}
						\UI$\Pi \fCenter \Sigma_1 \AOR \Sigma_2$
						\DP
						& &
						\AX$\Pi \AAND \Pi \fCenter \Sigma$
						\LL{\fns$C_L$}
						\UI$\Pi \fCenter \Sigma$
						\DP
						&
						\AX$\Pi \fCenter \Sigma \AOR \Sigma$
						\RL{\fns$C_R$}
						\UI$\Pi \fCenter \Sigma$
						\DP
						\\
						&\\
						\mc{2}{c}{
							\AX$\Pi_1 \AAND (\Pi_2 \AAND \Pi_3) \fCenter \Sigma$
							\doubleLine
							\LL{\fns$A_{L}$}
							\UI$(\Pi_1 \AAND \Pi_2) \AAND \Pi_3 \fCenter \Sigma$
							\DP}
						& &
						\mc{2}{c}{
							\AX$\Pi \fCenter (\Sigma_1 \AOR \Sigma_2) \AOR \Sigma_3$
							\doubleLine
							\RL{\fns$A_{R}$}
							\UI$\Pi \fCenter \Sigma_1 \AOR (\Sigma_2 \AOR \Sigma_3)$
							\DP}
					\end{tabular}
				\end{center}
and the introduction rules for the lattice connectives (and their residuals) follow the same pattern as the introduction rules of any $f\in \mathcal{F}$ and $g\in \mathcal{G}$:
		
				{
				\begin{center}
					\begin{tabular}{@{}rl | rl@{}}
						\AXC{\phantom{$\abot \fCenter \ABOT$}}
						\LL{\fns$\abot_L$}
						\UI$\abot \fCenter \ABOT$
						\DP
						&
						\AX$\Pi \fCenter \ABOT$
						\RL{\fns$\abot_R$}
						\UI$\Pi \fCenter \abot$
						\DP
						&
						\AX$\ATOP \fCenter \Sigma$
						\LL{\fns$\aatop_L$}
						\UI$\aatop \fCenter \Sigma$
						\DP
						&
						\AXC{\phantom{$\ATOP \fCenter \aatop$}}
						\RL{\fns$\aatop_R$}
						\UI$\ATOP \fCenter \top$
						\DP
						\\
						& & & \\
						\AX$\varphi \AAND \psi \fCenter \Sigma$
						\LL{\fns$\pand_L$}
						\UI$\varphi \pand \psi \fCenter \Sigma$
						\DP
						&
						\AX$\Pi_1 \fCenter \varphi$
						\AX$\Pi_2 \fCenter \psi$
						\RL{\fns$\pand_R$}
						\BI$\Pi_1 \AAND \Pi_2 \fCenter \varphi \pand \psi$
						\DP
						\ \ &\ \ 
						\AX$\varphi \fCenter \Sigma_1$
						\AX$\psi \fCenter \Sigma_2$
						\LL{\fns$\por_L$}
						\BI$\varphi \por \psi \fCenter \Sigma_1 \AOR \Sigma_2$
						\DP
						&
						\AX$\Pi \fCenter \varphi \AOR \psi$
						\RL{\fns$\por_R$}
						\UI$\Pi \fCenter \varphi \por \psi$
						\DP
						\\
						%& \\
						%\AX$\Pi \fCenter \varphi$
						%\AX$\psi \fCenter \Sigma$
						%\LL{\fns$\pra_L$}
						%\BI$\varphi \ararr \psi \fCenter \Pi \ARARR \Sigma$
						%\DP
						%&
						%\AX$\Pi \fCenter \varphi \ARARR \psi$
						%\RL{\fns$\ararr_R$}
						%\UI$\Pi \fCenter \varphi \ararr \psi$
						%\DP
						%&
						%\AX$\varphi \ADRARR \psi \fCenter \Sigma$
						%\LL{\fns$(\adrarr_L)$}
						%\UI$\varphi \adrarr \psi \fCenter \Sigma$
						%\DP
						%&
						%\AX$\varphi \fCenter \Sigma$
						%\AX$\Pi \fCenter \psi$
						%\RL{\fns$(\adrarr_R)$}
						%\BI$\Sigma \ADRARR \Pi \fCenter \varphi \adrarr \psi$
						%\DP
						%\\
						%& \\
						%\AX$\psi \fCenter \Sigma$
						%\AX$\Pi \fCenter \varphi$
						%\LL{\fns$(\alarr_L)$}
						%\BI$\psi \alarr \varphi \fCenter \Sigma \ALARR \Pi$
						%\DP
						%&
						%\AX$\Pi \fCenter \varphi \ALARR \psi$
						%\RL{\fns$(\alarr_R)$}
						%\UI$\Pi \fCenter \varphi \alarr \psi$
						%\DP
						%&
						%\AX$\psi \ADLARR \varphi \fCenter \Sigma$
						%\LL{\fns$\adlarr_L$}
						%\UI$\psi \adlarr \varphi \fCenter \Sigma$
						%\DP
						%&
						%\AX$\Pi \fCenter \psi$
						%\AX$\varphi \fCenter \Sigma$
						%\RL{\fns$\adlarr_R$}
						%\BI$\Pi \ADLARR \Sigma \fCenter \psi \adlarr \varphi$
						%\DP
						%\\
					\end{tabular}
				\end{center}
				 }
			
\begin{rem}	 \label{rem: derivable rules}
Rules $\aand_{L1}$, $\aand_{L2}$, $\aand_{L}$, $\aor_{R1}$, $\aor_{R2}$ and $\aor_{R}$ in  $\mathrm{D.LE^\ast}$ are derivable in $\mathrm{D.DLE}$ as follows:
\begin{center}
\begin{tabular}{ccc}
{
$\aand_{L2}$:
\begin{tabular}{c}
\AX$\psi \fCenter \Sigma$
\LL{\fns$W_L$}
\UI$\varphi \AAND \psi \fCenter \Sigma$
\LL{\fns$\aand_L$}
\UI$\varphi \aand \psi \fCenter \Sigma$
\DP 
 \\
\end{tabular}
}
&
{
$\aand_{L1}$:
\begin{tabular}{c}
\AX$\varphi \fCenter \Sigma$
\LL{\fns$W_L$}
\UI$ \psi \AAND \varphi \fCenter \Sigma$
\LL{\fns$E_L$}
\UI$ \varphi \AAND \psi \fCenter \Sigma$
\LL{\fns$\aand_L$}
\UI$\varphi \aand \psi \fCenter \Sigma$
\DP 
 \\
\end{tabular}
}
 & 
{
$\aand_{R}$:
\begin{tabular}{c}
\AX$\Pi \fCenter \varphi$
\AX$\Pi \fCenter \psi$
\RL{\fns$\aand_R$}
\BI$\Pi \AAND \Pi \fCenter \varphi \aand \psi$
\LL{\fns$C_L$}
\UI$\Pi \fCenter \varphi \aand \psi$
\DP 
 \\
\end{tabular}
}
\\

{
$\aor_{L}$:
\begin{tabular}{c}
\AX$\varphi \fCenter \Sigma$
\AX$ \psi \fCenter \Sigma$
\LL{\fns$\aor_L$}
\BI$\varphi \aor \psi \fCenter \Sigma \AOR\Delta$
\RL{\fns$C_R$}
\UI$\varphi \aor \psi \fCenter \Sigma$
\DP 
 \\
\end{tabular}
}
 & 
{
$\aor_{R1}$:
\begin{tabular}{c}
\AX$\Pi \fCenter \varphi$
\RL{\fns$W_R$}
\UI$\Pi \fCenter \varphi \AOR \psi$
\RL{\fns$\aor_R$}
\UI$\Pi \fCenter \varphi \aor \psi$
\DP 
 \\
\end{tabular}
}
 & 
{
$\aor_{R_2}$:
\begin{tabular}{c}
\AX$\Pi \fCenter \psi$
\RL{\fns$W_R$}
\UI$\Pi \fCenter \psi \AOR \varphi$
\RL{\fns$E_R$}
\UI$\Pi \fCenter \varphi \AOR\psi $
\RL{\fns$\aor_R$}
\UI$\Pi \fCenter \varphi \aor \psi $
\DP 
 \\
\end{tabular}
}
\end{tabular}
\end{center}
\end{rem}

\begin{remark}\label{remark:distr}
	In what follows, we will work in the non-distributive setting with the calculus $\mathrm{D.LE}$ and its extensions. However, all the results we obtain about derivations in $\mathrm{D.LE}$ straightforwardly transfer to $\mathrm{D.DLE}$ using the following procedure: all applications of $\aand_{L1}$, $\aand_{L2}$, $\aand_{R}$, $\aor_{R1}$, $\aor_{R2}$ and $\aor_{L}$ will be replaced by their derivations in $\mathrm{D.DLE}$  (cf.~Remark \ref{rem: derivable rules}). 
	
	%All occurrences of $\land$ and $\lor$  in an inductive $\mathcal{L}_{\mathrm{DLE}}$-inequality which are classified as $\Delta$-adjoints  will be treated %in same way they are treated in $\mathrm{D.LE}$. \marginnote{AP: this is not right, is it? they will be treated by replacing the introduction rule which used in $\mathrm{D.LE}$ with the macro of its derivation in D.DLE (cf.~Remark \ref{rem: derivable rules})} 
	All occurrences of $\land$ (resp.~$\lor$)  in an inductive $\mathcal{L}_{\mathrm{DLE}}$-inequality which are classified as SLR (resp.~SRR)  will be treated as  connectives in $\mathcal{F}$ (resp.~$\mathcal{G}$). 
	\end{remark} 

\subsection{Derivations in pre-normal form}
\label{ssec: canonical form}
%\textcolor{red}{This subsection needs to be edited. We give the def in 3 parts: for definite analytic inductive LE-axioms; then for general analytic inductive LE-axioms, and finally for the distributive setting. At the moment we only have the definite case}
In Section \ref{sec:syntactic completeness}, we will show that any analytic inductive LE-axiom $\varphi \vdash \psi$ can be effectively derived in the corresponding basic cut-free calculus $\mathrm{\cfDLE}$ enriched with the structural analytic rules $R_1,\ldots, R_n$ corresponding to $\varphi \vdash \psi$. In fact,  the cut-free derivation we produce has a particular shape, referred to as {\em pre-normal form}, which we define in the present section.
%In this section, we describe the general shape of the (cut-free) derivation of any analytic inductive axiom $x\vdash y$ from its corresponding analytic rules $R_1,\ldots, R_n$. 
Informally, in a derivation in pre-normal form, a division of labour is effected on the applications of rules:\footnote{The name `pre-normal' is intended to remind of a similar division of labour, among rules applied in derivations in  normal form of the well known natural deduction systems for classical and intuitionistic logic.} some rules are applied only before the application of $R_i$ and some rules are applied only after the application of $R_i$. 

Before moving on to the definitions, we highlight the following fact: when using ALBA to compute the analytic structural rule(s) corresponding to a given analytic inductive LE-axiom $x\vdash y$, if $+\land$ and $-\lor$ occur as SRA nodes in a non-critical maximal PIA subtree of $\varphi \vdash \psi$, then this subtree will generate  two or more premises of one of the corresponding rules (depending on the number of occurrences of $+\land$ and $-\lor$). If $-\land$ and $+\lor$ occur as $\Delta$-adjoints in the Skeleton of $\varphi \vdash \psi$, then the axiom is non-definite, and by exhaustively permuting those occurrences upwards, i.e.~towards the roots of the signed generation trees, and then applying the ALBA splitting rules, the given axiom can be equivalently transformed into a set of definite axioms, each of which will correspond to one analytic structural rule.

\begin{definition}
\label{def:canonical form nonDist}
A derivation $\pi$ in $\mathrm{\cfDLE}$  of the analytic inductive axiom $\varphi \vdash \psi$ (also indicated as $Ax$) is in \emph{pre-normal form} if the unique application of each rule in its corresponding set of analytic structural rules $R_1(Ax),\ldots,R_m(Ax)$ computed by ALBA splits $\pi$ into the following components: 

{\small
\begin{center}
\begin{tabular}{c}
\hspace{-0.35cm}
\begin{tikzpicture}		
\node at(0,0) {
\AXC{$p_{1.1} \fCenter p_{1.1} \ \cdots \ p_{1.k} \fCenter p_{1.k} \ \ \ \ \ \cdots \ \ \ \ \ p_{n.1} \fCenter p_{n.1} \ \ \cdots \ \ p_{n.\ell} \fCenter p_{n.\ell}$}
\noLine
\UIC{$\ddots\vdots\iddots \quad \quad \quad \quad \quad \ \ \ \ \ \quad \quad \quad \quad \quad \!\ddots\vdots\iddots$}
\noLine
\UIC{$\ \ \Pi^1_1 \fCenter \Sigma^1_1 \quad \quad \quad \quad \ \ \cdots \ \ \ \ \quad \quad \quad \quad \Pi^1_n \fCenter \Sigma^1_n \ $}
\RL{$R_1(Ax)$}
\UIC{$\Pi^1 \fCenter \Sigma^1$}

%
%\AXC{$\cdots$}
%

\AXC{$p_{1.1} \fCenter p_{1.1} \ \cdots \ p_{1.k'} \fCenter p_{1.k'} \ \ \ \ \ \cdots \ \ \ \ \ p_{n'.1} \fCenter p_{n'.1} \ \cdots \ p_{n'.\ell'} \fCenter p_{n'.\ell'}$}
\noLine
\UIC{$\ \ddots\vdots\iddots \ \quad \quad \quad \quad \quad \quad \quad \quad \quad \quad \quad \ \ \ddots\vdots\iddots\quad$}
\noLine
\UIC{$\Pi^m_1 \fCenter \Sigma^{m\phantom{1}}_1 \quad \quad \quad \ \ \ \ \ \ \cdots \ \ \quad \quad \quad \quad \ \ \Pi^m_{n'} \fCenter \Sigma^m_{n'}\ \ $}
\RL{$\scriptsize R_m(Ax)$}
\UIC{$\Pi^m \fCenter \Sigma^{m\phantom{1}}$}
\noLine
\BIC{$\ $}
\noLine
\UIC{$\ \ \ \ \ \ddots\vdots\iddots$}
\noLine
\UIC{$\ \ \ \ \ \varphi \fCenter \psi$}
\DP
};
\node at (0.5,-0.3) {$\cdots$};
\node[rotate = -90] at (-7.2, -0.8) {$\underbrace{\hspace{1.53cm}}$};
\node at (-8.1,-0.8) {Skeleton($\pi$)};
\node[rotate = -90] at (-8, 0.79) {$\underbrace{\hspace{1.53cm}}$};
\node at (-8.6, 0.79) {PIA($\pi$)};
\end{tikzpicture}
\end{tabular}
\end{center} }

\noindent where: 
\begin{itemize}
\item[(i)] Skeleton($\pi$) is the proof-subtree of $\pi$ containing the root of $\pi$ and applications of \emph{invertible rules} for the introduction of all connectives occurring in the Skeleton of $\varphi \vdash \psi$ (possibly modulo applications of display rules);
\item[(ii)] PIA($\pi$) is a collection of  proof-subtrees of $\pi$ containing the initial axioms of $\pi$ and all the applications of \emph{non-invertible  rules} for the introduction of connectives occurring in the maximal PIA-subtrees in the signed generation trees of $\varphi \vdash \psi$ (possibly modulo applications of display rules) and such that
\item[(iii)] the root of each proof-subtree in PIA($\pi$) coincides with a premise of the application of $R(ax)$ in $\pi$, where the atomic structural variables are suitably instantiated with operational maximal PIA-subtrees of $\varphi \vdash \psi$. 
\end{itemize}
\end{definition}

\begin{definition}
	\label{def:canonical form Dist}
	A derivation $\pi$ in $\mathrm{\underline{D.DLE}}$ of the analytic inductive axiom $\varphi \vdash \psi$ (also indicated as $Ax$) is in \emph{pre-normal form} if the unique application of each rule in its corresponding set of analytic structural rules $R_1(Ax), \ldots, R_m(Ax)$ computed by ALBA splits $\pi$ into the following components:

{\small
\begin{center}
\begin{tabular}{c}
\hspace{-0.35cm}
\begin{tikzpicture}		
\node at(0,0) {
\AXC{$p_{1.1} \fCenter p_{1.1} \ \cdots \ p_{1.k} \fCenter p_{1.k} \ \ \ \ \ \cdots \ \ \ \ \ p_{n.1} \fCenter p_{n.1} \ \ \cdots \ \ p_{n.\ell} \fCenter p_{n.\ell}$}
\noLine
\UIC{$\ddots\vdots\iddots \quad \quad \quad \quad \quad \ \ \ \ \ \quad \quad \quad \quad \quad \!\ddots\vdots\iddots$}
\noLine
\UIC{$\ \ \Pi^1_1 \fCenter \Sigma^1_1 \quad \quad \quad \quad \ \ \cdots \ \ \ \ \quad \quad \quad \quad \Pi^1_n \fCenter \Sigma^1_n \ $}
\RL{$R_1(Ax)$}
\UIC{$\Pi^1 \fCenter \Sigma^1$}

%
%\AXC{$\cdots$}
%

\AXC{$p_{1.1} \fCenter p_{1.1} \ \cdots \ p_{1.k'} \fCenter p_{1.k'} \ \ \ \ \ \cdots \ \ \ \ \ p_{n'.1} \fCenter p_{n'.1} \ \cdots \ p_{n'.\ell'} \fCenter p_{n'.\ell'}$}
\noLine
\UIC{$\ \ddots\vdots\iddots \ \quad \quad \quad \quad \quad \quad \quad \quad \quad \quad \quad \ \ \ddots\vdots\iddots\quad$}
\noLine
\UIC{$\Pi^m_1 \fCenter \Sigma^{m\phantom{1}}_1 \quad \quad \quad \ \ \ \ \ \ \cdots \ \ \quad \quad \quad \quad \ \ \Pi^m_{n'} \fCenter \Sigma^m_{n'}\ \ $}
\RL{$\scriptsize R_m(Ax)$}
\UIC{$\Pi^m \fCenter \Sigma^{m\phantom{1}}$}
\noLine
\BIC{$\ $}
\noLine
\UIC{$\ \ \ \ \ \ddots\vdots\iddots$}
\noLine
\UIC{$\ \ \ \ \ \varphi \fCenter \psi$}
\DP
};
\node at (0.5,-0.3) {$\cdots$};
\node[rotate = -90] at (-7.2, -0.8) {$\underbrace{\hspace{1.53cm}}$};
\node at (-8.1,-0.8) {Skeleton($\pi$)};
\node[rotate = -90] at (-8, 0.79) {$\underbrace{\hspace{1.53cm}}$};
\node at (-8.6, 0.79) {PIA($\pi$)};
\end{tikzpicture}
\end{tabular}
\end{center} }

	\noindent where: 
	\begin{itemize}
		\item[(i)] Skeleton($\pi$) is the proof-subtree of $\pi$ containing, possibly modulo applications of display rules, the root of $\pi$ and applications of
		\begin{itemize}
			\item[(a)] \emph{invertible rules} for the introduction of all %non-lattice connectives occurring in the Skeleton of $x\vdash y$;
	%\item[(b)] \emph{invertible  rules} for the introduction of the lattice 
			connectives occurring as SLR nodes in the Skeleton of $\varphi \vdash \psi$;
			\item[(b)] \emph{non-invertible  rules} and Contraction for the introduction of all %lattice 
			connectives occurring as $\Delta$-adjoint nodes in the Skeleton of $\varphi \vdash \psi$;
		\end{itemize} 
		\item[(ii)] PIA($\pi$) is a collection of  proof-subtrees of $\pi$ containing, possibly modulo applications of display rules, the initial axioms of $\pi$ and  applications of 
		\begin{itemize}
			\item[(a)] \emph{non-invertible  rules} for the introduction of all %non-lattice connectives occurring in the maximal PIA-subtrees in the signed generation trees of $x\vdash y$;
			%\item[(b)] \emph{non-invertible rules} for the introduction of lattice 
			connectives occurring as unary SRA nodes or as SRR nodes in the maximal PIA-subtrees in the signed generation trees of $\varphi \vdash \psi$;
			\item[(b)] \emph{invertible rules} and Weakening for the introduction of all lattice connectives occurring as SRA nodes in the maximal PIA-subtrees in the signed generation trees of $\varphi \vdash \psi$; 
		\end{itemize} 	
	 and such that
		\item[(iii)] the root of each proof-subtree in PIA($\pi$) coincides with a premise of the application of $R(ax)$ in $\pi$, where the atomic structural variables are suitably instantiated with operational maximal PIA-subtrees of $\varphi \vdash \psi$. 
	\end{itemize}
\end{definition} %\marginnote{AP: frankly, this piece referring to propositions and corollaries sounds very cryptic to me. I wonder if, rather than only pointing at the examples that we give later, it would not be better to illustrate the two definitions with an example already now, possibly the same example, so that we can also show that  translating a non distributive derivation into a distributive one preserves normal form. }
The key tools for obtaining the sub-derivations in PIA($\pi$) introducing the connectives occurring as unary SRA nodes or as SRR nodes are given in Proposition \ref{prop: generalized derivation of identities} and Corollary \ref{cor: generalized deriving la almost implies atom}. An inspection on the proofs of these results reveals that indeed only non-invertible logical rules and display rules are applied. 
The key tools involving the introduction of the lattice connectives occurring as  SRA nodes in PIA($\pi$)  (resp.~as $\Delta$-adjoint nodes in Skeleton($\pi$)) are given in Proposition \ref{prop: generalized derivation of amlost identities} (resp.~Proposition \ref{prop: the thing needed for the non-definite}). Again, inspecting the proofs of these results reveals that only introduction rules of {\em one type} are applied in each component. 
\begin{remark}
The binary introduction rules of  $\mathrm{\cfDLE}$  for lattice connectives  are invertible, while the corresponding rules of $\mathrm{\underline{D.DLE}}$ are not, and Contraction is needed to derive these rules of $\mathrm{\cfDLE}$ in $\mathrm{\underline{D.DLE}}$. Likewise, the unary introduction rules of  $\mathrm{\cfDLE}$ for lattice connectives  are not invertible, while the corresponding rules of $\mathrm{\underline{D.DLE}}$ are, and so Weakening is needed to derive these rules of  $\mathrm{\cfDLE}$ in $\mathrm{\underline{D.DLE}}$. This is why derivations in pre-normal form of  analytic inductive axioms in the general lattice setting of Definition \ref{def:canonical form nonDist} can be described purely in terms of invertible and non-invertible introduction rules, while in the distributive lattice setting of Definition \ref{def:canonical form Dist}, the occurrences of lattice connectives in `$\Delta$-adjoint/SRA-position' in  the signed generation trees of a given analytic inductive axiom need to be accounted for separately (cf.~clauses (b) of Definition \ref{def:canonical form Dist}).  
However, if $\varphi \vdash \psi$ is an analytic inductive axiom in the general lattice setting, applying the process described in Remark \ref{remark:distr} to a derivation of $\varphi \vdash \psi$ in $\mathrm{\cfDLE}$  in pre-normal form according to Definition \ref{def:canonical form nonDist}  results in a derivation of $\varphi \vdash \psi$ in $\mathrm{\underline{D.DLE}}$ which is in pre-normal form according to Definition \ref{def:canonical form Dist}.
\end{remark}

\begin{remark}
If $\varphi \vdash \psi$ is a {\em definite} analytic inductive axiom, then ALBA yields a {\em single} analytic structural rule corresponding to it.  %\marginnote{AP: here the sentence continued as follows: ``and this rule has as many premises as there are non-critical maximal PIA subtrees in the given axiom.'' But given the preliminary remark I do not think that this is tue because one PIA subtree can split into more. Do you agree?}
So, both in the general lattice and in the distributive settings, the Skeleton part of the derivation of $\varphi \vdash \psi$ in pre-normal form will only have one branch, yielding the following simpler shape of $\pi$:
	\begin{center}
	\begin{tabular}{c}
		\begin{tikzpicture}		
			\node at(0,0) {
				\AXC{$p_{1.1} \fCenter p_{1.1} \quad \cdots \quad p_{1.k} \fCenter p_{1.k}$}
				\noLine
				\UIC{$\ddots\vdots\iddots$}
				\noLine
				\UIC{$\Pi_1 \fCenter \Sigma_1$}
				\AXC{$\cdots$}
				\noLine
				\UIC{$\phantom{\vdots}$}
				\noLine
				\UIC{$\phantom{\Pi_i} \cdots \phantom{\Pi_j}$}
				\AXC{$p_{n.1} \fCenter p_{n.1} \quad \cdots \quad p_{n.\ell} \fCenter p_{n.\ell}$}
				\noLine
				\UIC{$\ddots\vdots\iddots$}
				\noLine
				\UIC{$\Pi_n \fCenter \Sigma_n$}
				\RL{$R(Ax)$}
				\TIC{$\Pi \fCenter \Sigma$}
				\noLine
				\UIC{$\vdots$}
				\noLine
				\UIC{$\varphi \fCenter \psi$}
				\DP
			};
			\node[rotate = -90] at (-3.9, -0.8) {$\underbrace{\hspace{1.4cm}}$};
			\node at (-5,-0.8) {Skeleton($\pi$)};
			\node[rotate = -90] at (-5.2, 0.75) {$\underbrace{\hspace{1.47cm}}$};
			\node at (-6, 0.75) {PIA($\pi$)};
		\end{tikzpicture}
	\end{tabular}
\end{center}
\end{remark}
All derivations in Examples \ref{example: deriving qsi axioms} and \ref{ex:derivations general} are derivations of definite analytic inductive axioms in pre-normal form.

\section{Properties of the basic display calculi $\mathrm{D.LE}$}	
\label{sec:properties}	
		
In this section, we will state and prove the key lemmas needed for the proof of the syntactic completeness. Throughout this section, we let $\mathcal{L}_{\mathrm{LE}}$ (resp.~$\mathcal{L}_{\mathrm{DLE}}$) be an arbitrary but fixed (D)LE-language, and $\mathrm{D.LE}$ (resp.~$\mathrm{D.DLE}$) denote the proper display calculi for the basic $\mathcal{L}_{\mathrm{LE}}$-logic (resp.~$\mathcal{L}_{\mathrm{DLE}}$-logic).

\begin{notation}		
\label{notation: structural counterparts of skeleton and pia} 
For any definite Skeleton (resp.~definite PIA) formula $\varphi$ (resp.~$\psi, \gamma, \delta, \xi, \ldots$), we let its corresponding capital Greek letter $\Phi$ (resp.~$\Psi, \Gamma, \Delta, \Xi, \ldots$) denote its structural counterpart, defined by induction as follows (cf.~Notation \ref{notation: placeholder variables}): %\marginnote{AP: do we really need to separate variables here?\\ I don't think so. G: strictly speaking we do not, I guess. But we can and this is consistent with the rest of the paper. I would also suggest to recall here that we adopt the convention that we put the monotone coordinates first and point to def 10 where we adopt the same convention... we are used to this, but it makes sense to spend 1 or 2 lines on this to help the reader... moreover this match with colours and the readability / meaningfulness of our conventions will be more transparent.}
\begin{enumerate} 
\item if $\varphi: =p\in \mathsf{AtProp}$, then $\Phi: = p$; 
\item if $\varphi: =f(\overline{\xi}, \overline{\psi})$, then $\Phi: = \hat{f}(\overline{\Xi}, \overline{\Psi})$;
\item if $\varphi: =g(\overline{\psi}, \overline{\xi})$, then $\Phi: = \check{g}(\overline{\Psi}, \overline{\Xi})$.
\end{enumerate}
Notice that items 2 and 3 above cover also the case of zero-ary connectives (and of $\wedge$ and $\vee$ in the setting of $\mathrm{D.DLE}$).
\end{notation}
Also,  notice that the introduction rules of $\mathrm{D.LE}$ (resp.~$\mathrm{D.DLE}$) are such that structural counterparts of connectives in $\mathcal{F}$ (resp.~in $\mathcal{F}\cup\{\wedge\}$) can only occur in precedent position, and  structural counterparts of connectives in $\mathcal{G}$ (resp.~in $\mathcal{G}\cup\{\vee\}$) can only occur in succedent position, which is why Notation \ref{notation: structural counterparts of skeleton and pia} only applies to definite Skeleton and definite PIA formulas.

\begin{notation}\label{notation: isotone vs antitone}
In what follows, we let $\overline{\sigma}$, $\overline{S}$ and $\overline{\sigma\vdash S}$ (resp.~$\overline{\tau}$, $\overline{U}$ and $\overline{U\vdash \tau}$) denote finite vectors of formulas, of structures in $\mathsf{Str}_\mathcal{G}$ (resp.~$\mathsf{Str}_\mathcal{F}$) and of $\mathrm{D.(D)LE}$-sequents. 
%
%We also extend the colour notation to the structural terms of the display calculi  $\mathrm{D.LE}$ and $\mathrm{D.DLE}$, with the proviso that the colour (blue or red) we assign to the occurrence of a {\em structural} term is {\em the same} as the colour assigned to the corresponding occurrence of its operational counterpart. This implies that for occurrences of structural terms in sequents, the colour blue (resp.~red) indicates that the given occurrence is in {\em succedent} (resp.~{\em precedent}) position, exactly the opposite of what the colours indicate relative to operational terms (cf.~Notation \ref{notation: analytic inductive}).
%In the remainder of this section, we write $\gamma(!\obx, !\ory)$ and $\varphi(!\obx, !\ory) $ (resp.~$\delta(!\ory, !\obx)$ and $\psi(!\ory, !\obx)$) to signify that $\varepsilon(x)\prec \ast \gamma$ and $\varepsilon(x)\prec \ast \varphi$ for every $x$ in $\obx$ (resp.~$\varepsilon(y)\prec \ast \delta$ and $\varepsilon(y)\prec \ast \psi$ for every $y$ in $\ory$), and $\varepsilon^\partial(y)\prec \ast \gamma$ and $\varepsilon^\partial(y)\prec \ast \varphi$ for every $y$ in $\ory$ (resp.~$\varepsilon^\partial(x)\prec \ast \delta$ and $\varepsilon^\partial(x)\prec \ast \psi$ for every $x$ in $\obx$).
\end{notation}

\begin{prop}\label{prop: generalized derivation of identities}
For every definite positive PIA (i.e.~definite negative Skeleton) formula $\gamma(!\overline{x}, !\overline{y})$ and every definite negative PIA (i.e.~definite positive Skeleton) formula $\delta(!\overline{y}, !\overline{x})$,
			%The  calculi $\mathbf{DL}$ derive the following sequents (cf.~Notation \ref{notation: structural counterparts of skeleton and pia}):
\begin{enumerate} %\marginnote{AP: please check the notational convention before this propsition. G: it is fine because the main connective of $\gamma$ (or $\varphi$) is supposed to be a g-operator (right-adjoint/residual) later on... see for instance the next Proposition). and in the case of g-operator our convention is to put monotone variables first, so blue (and vice versa for $\delta$ or $\psi$). Should we add an intuitive explanation of Notation 5 in the text?}
\item if $\overline{\sigma\vdash S}$ and $\overline{U\vdash \tau}$ are derivable in $\mathrm{D.LE}$ (resp.~$\mathrm{D.DLE}$), then so is $\gamma[\overline{\bsigma}/!\obx, \overline{\rtau}/!\ory]\vdash \Gamma[\overline{S}/!\ox, \overline{U}/!\oy]$;
\item if $\overline{\sigma\vdash S}$ and $\overline{U\vdash \tau}$ are derivable in $\mathrm{D.LE}$ (resp.~$\mathrm{D.DLE}$), then so is $\Delta[\overline{U}/!\oy, \overline{S}/!\ox]\vdash \delta[\overline{\rtau}/!\ory, \overline{\bsigma}/!\obx]$.
%\item $\Delta\vdash \delta$ for every negative PIA (i.e.~positive Skeleton) formula $\delta$,
\end{enumerate}	
with derivations such that, if any rules are applied other than right-introduction rules for  negative SRR-connectives and negative unary SRA-connectives (cf.~Tables \ref{Join:and:Meet:Friendly:Table} and \ref{Join:and:Meet:Friendly:Table:DLE}, and Definition \ref{def: signed gen tree}), 
and left-introduction rules for positive SRR-connectives and positive unary SRA-connectives, then they are applied only in the derivations of $\overline{\sigma\vdash S}$ and $\overline{U\vdash \tau}$.
			\end{prop}
	\begin{proof}
	By simultaneous induction on $\gamma$ and $\delta$. If $\gamma : =x$, then $\Gamma : =x$. Hence, $\gamma[\overline{\bsigma}/!\obx, \overline{\rtau}/!\ory] \vdash \Gamma[\overline{S}/!\ox, \overline{U}/!\oy]$ reduces to $\sigma \vdash S$, which is derivable by  assumption. The case of $\delta: =y$ is shown similarly.
		As to the inductive steps, let $\gamma(!\overline{x}, !\overline{y}) := g(\overline{\psi}(!\overline{x}, !\overline{y}), \overline{\xi}(!\overline{x}, !\overline{y}))$ with $\overline{\psi}$ definite positive PIA-formulas and  $\overline{\xi}$ definite negative PIA-formulas.
%
% $\gamma = g(\overline{\psi}, \overline{\xi})[!\obx, !\ory]$. By definition, $\overline{\psi}$ are definite positive PIAs and  $\overline{\xi}$ are definite negative PIAs. We may assume that $\gamma = g(\overline{\psi}[\overline{!x_1}, \overline{!y_1}] , \overline{\xi}[\overline{!y_2},\overline{!x_2}])$, where $\overline{!x_1}\cup \overline{!x_2} =!\obx$ and $\overline{!y_1}\cup \overline{!y_2} =!\ory$. 
Then 
$\gamma[\overline{\sigma}/!\overline{x}, \overline{\tau}/!\overline{y}] = g(\overline{\psi}[\overline{\sigma}/ !\overline{x},\overline{\tau}/ !\overline{y}], \overline{\xi}[\overline{\tau}/!\overline{y},\overline{\sigma}/!\overline{x}])$ and $ \Gamma[\overline{S}/!\overline{x}, \overline{U}/!\overline{y}]= \check{g}(\overline{\Psi}[\overline{S}/ !\overline{x},\overline{U}/ !\overline{y}] , \overline{\Xi}[\overline{U}/!\overline{y},\overline{S}/!\overline{x}])$. %where $\overline{\sigma_1}\cup \overline{\sigma_2}=\overline{\sigma}$, $\overline{\tau_1}\cup\overline{\tau_2}= \overline{\tau}$, $\overline{S_1} \cup \overline{S_2} =\overline{S}$ and $\overline{U_1} \cup \overline{U_2}=\overline{U}$. 

By induction hypothesis, all sequents in the following vectors are derivable in  $\mathrm{D.LE}$ (resp.~$\mathrm{D.DLE}$):
\[\overline{\psi[\overline{\bsigma}/ !\obx,\overline{\rtau}/ !\ory] \vdash \Psi [\overline{S}/ !\ox,\overline{U}/ !\oy]} \quad \text{ and }
\quad \overline{\Xi[\overline{U}/!\oy, \overline{S}/!\ox] \vdash \xi [\overline{\rtau}/!\ory,\overline{\bsigma}/!\obx]}.\] Then we can derive the required sequent $\gamma[\overline{\bsigma}/!\obx, \overline{\rtau}/!\ory]\vdash \Gamma[\overline{S}/!\ox, \overline{U}/!\oy]$ by prolonging all these derivations with an application of $g_L$ as follows:
\[
		\AxiomC{$\overline{\psi[\overline{\bsigma}/ !\obx,\overline{\rtau}/ !\ory]\vdash \Psi[\overline{S}/ !\ox,\overline{U}/ !\oy]}$}
		\AxiomC{$\overline{\Xi[\overline{U}/!\oy, \overline{S}/!\ox]\vdash\xi[ \overline{\rtau}/!\ory,\overline{\bsigma}/!\obx]}$}
		\LL{\fns$g_L$}
		\BinaryInfC{$g(\overline{\bpsi}[\overline{\bsigma}/ !\obx,\overline{\rtau}/ !\ory], \overline{\rxi}[ \overline{\rtau}/!\ory,\overline{\bsigma}/!\obx]) \vdash \check{g}(\overline{\Psi}[\overline{S}/ !\ox,\overline{U}/ !\oy], \overline{\Xi}[\overline{U}/!\oy, \overline{S}/!\ox])$}
		\DisplayProof
\]

Let $\delta (!\overline{x}, !\overline{y}) : = f(\overline{\xi}(!\overline{x}, !\overline{y}), \overline{\psi}(!\overline{x}, !\overline{y}))$ with $\overline{\xi}$  definite negative PIA-formulas (i.e.~positive Skeleton-formulas) and $\overline{\psi}$  definite positive PIA-formulas (i.e.~negative Skeleton-formulas).
% We may assume that $\delta = f ( \overline{\xi}[\overline{!y_2},\overline{!x_2}],\overline{\psi}[\overline{!x_1}, \overline{!y_1}])$, where $\overline{!x_1}\cup \overline{!x_2} =!\obx$ and $\overline{!y_1}\cup \overline{!y_2} =!\ory$. 
Then 
$\delta[\overline{\tau}/!\overline{y}, \overline{\sigma}/!\overline{x}] =  f ( \overline{\xi}[\overline{\tau}/ !\overline{y},\overline{\sigma}/!\overline{x}],\overline{\psi}[\overline{\sigma}/!\overline{x}, \overline{\tau}/!\overline{y}])$ and $ \Delta[\overline{U}/!\overline{y}, \overline{S}/!\overline{x}]=  \hat{f} ( \overline{\Xi}[\overline{U}/ !\overline{y},\overline{S}/!\overline{x}],\overline{\Psi}[\overline{S}/!\overline{x}, \overline{U}/!\overline{y}])$. %where $\overline{\sigma}\cup \overline{\sigma}=\overline{\sigma}$, $\overline{\tau_1}\cup\overline{\tau_2}= \overline{\tau}$, $\overline{S_1} \cup \overline{S_2} =\overline{S}$ and $\overline{U_1} \cup \overline{U_2}=\overline{U}$. 
By induction hypothesis, all sequents in the following vectors are derivable in  $\mathrm{D.LE}$ (resp.~$\mathrm{D.DLE}$): \[\overline{\psi[\overline{\bsigma}/ !\obx,\overline{\rtau}/ !\ory]\vdash \Psi[\overline{S}/ !\ox, \overline{U}/ !\oy]}\quad \text{ and } \quad
\overline{\Xi[\overline{U}/!\oy, \overline{S}/!\ox])\vdash \xi [\overline{\rtau}/!\ory,\overline{\bsigma}/!\obx]}.\]  Then we can derive the required sequent $ \Delta[\overline{U}/!\oy, \overline{S}/!\ox]\vdash \delta [\overline{\rtau}/!\ory, \overline{\bsigma}/!\obx]$ by prolonging all these derivations with an application of $f_R$ as follows:
\[
\AxiomC{$\overline{\Xi[\overline{U}/!\oy, \overline{S}/!\ox]\vdash\xi [\overline{\rtau}/!\ory,\overline{\bsigma}/!\obx]}$}
\AxiomC{$\overline{\psi[\overline{\bsigma}/!\obx,\overline{\rtau}/!\ory]\vdash \Psi[\overline{S}/ !\ox,\overline{U}/!\oy]}$}
\RL{\fns$f_R$}
		\BinaryInfC{$\hat{f}(\overline{\Xi}[\overline{U}/!\oy, \overline{S}/!\ox],\overline{\Psi}[\overline{S}/!\ox,\overline{U}/!\oy]) \vdash f(\overline{\rxi}[ \overline{\rtau}/!\ory,\overline{\bsigma}/!\obx],\overline{\bpsi}[\overline{\bsigma}/ !\obx,\overline{\rtau}/ !\ory] )$}
		\DisplayProof
\]
The proof, specific to the setting of $\mathrm{D.DLE}$, of the case in which $\gamma: = \gamma_1\vee \gamma_2$ (resp.~$\delta: = \delta_1\wedge \delta_2$) goes like the case of arbitrary $g\in \mathcal{G}$ (resp.~$f\in \mathcal{F}$) discussed above, using the $\mathrm{D.DLE}$-rule $\vee_L$ (resp.~$\wedge_R$).
\end{proof}
	By instantiating $\overline{\sigma\vdash S}$ and $\overline{U\vdash \tau}$ in the proposition above to identity axioms, we immediately get the following
\begin{cor}\label{prop: derivation of identities}
			Any  calculus $\mathrm{\cfDLE}$ (resp.~$\mathrm{\cfDDLE}$) derives the following sequents (cf.~Notation \ref{notation: structural counterparts of skeleton and pia}):
\begin{enumerate}
\item $\gamma\vdash \Gamma$ for every definite positive PIA (i.e.~definite negative Skeleton) formula $\gamma$;
\item $\Delta\vdash \delta$ for every definite negative PIA (i.e.~definite positive Skeleton) formula $\delta$,
\end{enumerate}	
with derivations which only consist of identity axioms, and applications of right-introduction rules for  negative SRR-connectives and negative unary SRA-connectives (cf.~Tables \ref{Join:and:Meet:Friendly:Table} and \ref{Join:and:Meet:Friendly:Table:DLE}, and Definition \ref{def: signed gen tree}),  and left-introduction rules for positive SRR-connectives and positive unary SRA-connectives. 
			\end{cor}

\begin{example}
The formula $\wdia (p \mand q) \ararr (q \mor p) $ is definite positive PIA in any (D)LE-language such that $\Diamond, \mand\in \mathcal{F}$ and $\mor, \rightarrow\; \in \mathcal{G}$ with $n_\Diamond = 1$ and $\epsilon_\Diamond(1) = 1$, and $n_\mand = n_\mor = n_\rightarrow = 2$ and $\varepsilon_\circ (i) = 1$ for every $\circ\in \{\mand, \mor, \rightarrow\}$ and every $1\leq i\leq 2$ except $\varepsilon_{\rightarrow}(1) = \partial$. Then, instantiating the argument above, we can derive  the sequent
$
\wdia (p \mand q) \ararr (q \mor p) \vdash
\WDIA(p\MAND q) \ARARR(q \MOR p)
$ in $\mathrm{\cfDLE}$ (resp.~$\mathrm{\cfDDLE}$) as follows: 
\[
\AX$p\fCenter p$
\AX$q\fCenter q$
\RL{\fns $\mand_R$}
\BI$p \MAND q \fCenter p \mand q$
\RL{\fns $\wdia_R$}
\UI$\WDIA(p \MAND q)\fCenter \wdia (p\mand q)$
\AX$p\fCenter p$
\AX$q\fCenter q$
\LL{\fns $\mor_L$}
\BI$p \mor q \fCenter p\MOR q$
\LL{\fns $\ararr_L$}
\BI$ \wdia (p\mand q) \ararr (p \mor q ) \fCenter \WDIA(p \MAND q) \ARARR  p\MOR q$
\DP
\] 
The formula $\wdia p \mand q $ in the same language is  definite negative PIA. Then, instantiating the argument above, we can derive  the sequent $\WDIA p \MAND q \vdash \wdia p\mand q$ in $\mathrm{\cfDLE}$ (resp.~$\mathrm{\cfDDLE}$) as follows: 

\[
\AX$p\fCenter p$
\RL{\fns$\wdia_R$}
\UI$\WDIA p \fCenter \wdia p $
\AX$q \fCenter q$
\RL{\fns$\mand_R$}
\BI$ \WDIA p \MAND q \fCenter \wdia p\mand q$
\DP
\]
\end{example}

\begin{prop}\label{prop: generalized derivation of amlost identities}
	%The  calculi $\mathbf{DL}$ derive the following sequents (cf.~Notation \ref{notation: structural counterparts of skeleton and pia}):
	Let $\gamma = \gamma(!\overline{x}, !\overline{y})$ and $\delta = \delta(!\overline{y}, !\overline{x})$ be a positive and a negative PIA formula, respectively, and let $\bigwedge_{i\in I}\gamma_i$ and $\bigvee_{j\in J}\delta_j$ be their equivalent rewritings as per Lemma \ref{lemma: reduction to definite}, so that each $\gamma_i$ (resp.~each $\delta_j$) is definite positive (resp.~negative) PIA. 
	\begin{enumerate}
		\item  If $\overline{\sigma\vdash S}$ and $\overline{U\vdash \tau}$ are derivable in $\mathrm{\cfDLE}$ (resp.~$\mathrm{\cfDDLE}$), then so is $\gamma[\overline{\bsigma}/!\obx, \overline{\rtau}/!\ory] \vdash \Gamma_i[\overline{S}/!\ox, \overline{U}/!\oy]$ for each $i\in I$;
		\item  if $\overline{\sigma\vdash S}$ and $\overline{U \vdash \tau}$ are derivable in $\mathrm{\cfDLE}$ (resp.~$\mathrm{\cfDDLE}$), then so is $\Delta_{j\,}[\overline{U}/!\oy, \overline{S}/!\ox] \vdash \delta[\overline{\rtau}/!\ory, \overline{\bsigma}/!\obx]$ for each $j\in J$, %\marginnote{AP: I think we should stick to the $x$ for monotone coordinates and $y$ for antitone coordinates, so item 2 needs to be edited accordingly. G: well this information (being an atomic variable occurring in an monotone versus antitone coordinate) is conveyed by the main connective of the formula under consideration and the fact that it occurs before vs after the comma in our notation, i.e.~$\gamma(!\obx, !\ory)$ is declared positive PIA, so the main connective of $\gamma$ is a right adjoint, so $x$ is monotone in this case (we put monotone coordinates first), but $\delta = (!\ory, !\obx)$ is negative PIA, so the main connective of $\delta$ is a left adjoint, so $y$ is monotone in this case (we put monotone coordinates first). In the end it is a matter of conventions... we need to choose the convention that is more transparent / easy to handle / conveys the `right' information. I guess that thanks to the actual convention the substitutions are always red/red or blue/blue (without the need to control the position in the notation... and I do not even know it this is actually possible just looking at the formula)... but we should clarify this point.}
		%\item $\Delta\vdash \delta$ for every negative PIA (i.e.~positive Skeleton) formula $\delta$,
	\end{enumerate}	
	with derivations such that, if any rules are applied other than right-introduction rules for  negative PIA-connectives (cf.~Tables \ref{Join:and:Meet:Friendly:Table} and \ref{Join:and:Meet:Friendly:Table:DLE}, and Definition \ref{def: signed gen tree}), 
and left-introduction rules for positive PIA-connectives 
(and weakening and exchange rules in the case of $\mathrm{\cfDDLE}$), then they are applied only in the derivations of $\overline{\sigma\vdash S}$ and $\overline{U\vdash \tau}$.
\end{prop}

\begin{proof}
Let $n_\gamma(+\land)$ (resp.~$n_\delta(+\land)$) be the number of occurrences of $+\land$  in $+\gamma$ (resp.~$-\delta$), and let $n_\gamma(-\lor)$ (resp.~$n_\delta(-\lor)$) be the number of occurrences of $-\lor$ in $+\gamma$ (resp.~$-\delta$). The proof is by simultaneous induction on $n_\gamma = n_\gamma(+\land)+ n_\gamma(-\lor)$ and $n_\delta =  n_\delta(+\land)+ n_\delta(-\lor)$. 

If  $n_\gamma = n_\delta = 0$, then  $\gamma$ (resp.~$\delta$) is definite positive (resp.~negative) PIA. Then the claims follow from Proposition \ref{prop: generalized derivation of identities}.

If $n_\gamma \geq 1$, then let us consider one %of the topmost 
occurrence of $+\land$ or $-\lor$ in $+\gamma$, which we will refer to as `the focal occurrence'. 
Let us assume that the focal  occurrence of $+\land$ or $-\lor$ in $+\gamma$ %is not the root of $+\gamma$, and that it 
is an occurrence of $-\vee$ (the case in which it is an occurrence of $+\land$ is argued similarly).

Let $-\xi'$ and $-\xi''$ be the two subtrees under the focal occurrence of  $-\lor $. Then $\xi '\lor \xi''$ is a subformula of $\gamma$ such that $\xi'$ and $\xi ''$ are negative PIA formulas, and $n_{\xi'}$ and $n_{\xi''}$ are strictly smaller than $n_{\gamma}$. Let $u$ be a fresh variable which does not occur in $\gamma$, and let $\gamma'$ be the formula obtained by substituting the occurrence of $\xi'\lor \xi''$ in $\gamma$ with $u$.  Then $\gamma'$ is a positive PIA formula such that $n_{\gamma'} $ is strictly smaller than $n_\gamma$, and $\gamma = \gamma' [(\xi '\lor \xi '')/!u]$. 
Let $\bigwedge_{i\in I}\gamma_i$, $\bigwedge_{j\in J}\gamma'_{j}$, 
$\bigvee_{h\in H}\xi'_{h}$ and $\bigvee_{k\in K}\xi''_{k}$ be the equivalent rewritings of $\gamma$, $\gamma'$, $\xi'$ and $\xi''$, respectively, resulting from distributing exhaustively $+\wedge$ and $-\vee$ over each connective in $\gamma$, $\gamma'$, $\xi'$ and $\xi''$, respectively. Then, 
%for every $i\in I$, either $\gamma_i = \gamma'_{j}[\xi'_{h}/!u ]$ for some $j \in J$ and $h \in H$, or $\gamma_i = \gamma'_{j}[\xi''_{k}/!u ]$ for some $j \in J$ and $k \in K$. 
\[\{\gamma_i\mid i\in I\} = \{\gamma'_{j}[\xi'_{h}/!u ] \mid j \in J \text{ and } h \in H\}\cup \{\gamma'_{j}[\xi''_{k}/!u ] \mid j \in J \text{ and } k \in K\}.\]
By induction hypothesis, the following sequents are derivable in $\mathrm{\cfDLE}$ (resp.~$\mathrm{\cfDDLE}$) for every $h \in H$ and $k \in K$:
\[
 \Xi_{h\,}[\overline{U}/!\oy, \overline{S}/!\ox] \vdash \xi'[\overline{\rtau}/!\ory, \overline{\bsigma}/!\obx] \quad \text{~and~}\quad \Xi_{k\,}[\overline{U}/!\oy, \overline{S}/!\ox] \vdash \xi''[\overline{\rtau}/!\ory, \overline{\bsigma}/!\obx] .
\]
Then, by prolonging the derivations of the two sequents above with suitable applications of $(\lor_{R1})$ and $(\lor_{R2})$, we obtain   derivations in $\mathrm{\cfDLE}$ (resp.~$\mathrm{\cfDDLE}$)\footnote{In the calculus $\mathrm{\cfDDLE}$, we obtain the derivations of the sequents in \eqref{eq:der sequents} by replacing the  applications of the derivable rules $(\lor_{R1})$ and $(\lor_{R2})$ with the derivations of those applications as shown in Remark \ref{rem: derivable rules}, thereby implementing the general strategy outlined in Remark \ref{remark:distr}. That is:
\begin{center}
\begin{tabular}{cc}
{\fns
\begin{tabular}{c}
\AX$\Xi_{h\,}[\overline{U}/!\oy, \overline{S}/!\ox]\fCenter \xi'[\overline{\rtau}/!\ory, \overline{\bsigma}/!\obx]$
\RL{\fns$W_{R}$}
\UI$ \Xi_{h\,}[\overline{U}/!\oy, \overline{S}/!\ox]\fCenter \xi' [\overline{\rtau}/!\ory, \overline{\bsigma}/!\obx]\AOR \xi''[\overline{\rtau}/!\ory, \overline{\bsigma}/!\obx]$
\RL{\fns$\lor_{R}$}
\UI$ \Xi_{h\,}[\overline{U}/!\oy, \overline{S}/!\ox]\fCenter \xi' [\overline{\rtau}/!\ory, \overline{\bsigma}/!\obx]\lor \xi''[\overline{\rtau}/!\ory, \overline{\bsigma}/!\obx]$
\DP 
 \\
\end{tabular}
}
&
{\fns
\begin{tabular}{c}
\AX$\Xi_{k\,}[\overline{U}/!\oy, \overline{S}/!\ox]\fCenter  \xi''[\overline{\rtau}/!\ory, \overline{\bsigma}/!\obx]$
\RL{\fns$W_{R}$}
\UI$\Xi_{k\,}[\overline{U}/!\oy, \overline{S}/!\ox]\fCenter  \xi''[\overline{\rtau}/!\ory, \overline{\bsigma}/!\obx] \AOR \xi'[\overline{\rtau}/!\ory, \overline{\bsigma}/!\obx]$
\RL{\fns$E_{R}$}
\UI$\Xi_{k\,}[\overline{U}/!\oy, \overline{S}/!\ox]\fCenter \xi'[\overline{\rtau}/!\ory, \overline{\bsigma}/!\obx] \AOR  \xi''[\overline{\rtau}/!\ory, \overline{\bsigma}/!\obx] $
\RL{\fns$\lor_{R}$}
\UI$\Xi_{k\,}[\overline{U}/!\oy, \overline{S}/!\ox]\fCenter \xi'[\overline{\rtau}/!\ory, \overline{\bsigma}/!\obx] \lor \xi''[\overline{\rtau}/!\ory, \overline{\bsigma}/!\obx] $
\DP 
 \\
\end{tabular}
}
\end{tabular}
\end{center}

} of the following sequents  for every $h \in H$ and $k \in K$:
\begin{equation} \label{eq:der sequents}
 \Xi_{h\,}[\overline{U}/!\oy, \overline{S}/!\ox] \vdash (\xi'\lor  \xi'')[\overline{\rtau}/!\ory, \overline{\bsigma}/!\obx] \quad \text{~and~}\quad \Xi_{k\,}[\overline{U}/!\oy, \overline{S}/!\ox]\vdash (\xi' \lor  \xi'')[\overline{\rtau}/!\ory, \overline{\bsigma}/!\obx] .
\end{equation}

\begin{center}
\begin{tabular}{cc}
{\fns
\begin{tabular}{c}
\AX$\Xi_{h\,}[\overline{U}/!\oy, \overline{S}/!\ox]\fCenter \xi'[\overline{\rtau}/!\ory, \overline{\bsigma}/!\obx]$
\RL{\fns$\lor_{R1}$}
\UI$ \Xi_{h}[\overline{U}/!\oy, \overline{S}/!\ox]\fCenter \xi' [\overline{\rtau}/!\ory, \overline{\bsigma}/!\obx]\lor  \xi''[\overline{\rtau}/!\ory, \overline{\bsigma}/!\obx]$
\DP 
 \\
\end{tabular}
}
&
{\fns
\begin{tabular}{c}
\AX$\Xi_{k}[\overline{U}/!\oy, \overline{S}/!\ox] \fCenter  \xi''[\overline{\rtau}/!\ory, \overline{\bsigma}/!\obx]$
\RL{\fns$\lor_{R2}$}
\UI$\Xi_{k}[\overline{U}/!\oy, \overline{S}/!\ox]\fCenter \xi'[\overline{\rtau}/!\ory, \overline{\bsigma}/!\obx] \lor \xi''[\overline{\rtau}/!\ory, \overline{\bsigma}/!\obx] $
\DP 
 \\
\end{tabular}
}
\end{tabular}
\end{center}

By induction hypothesis on $\gamma'$, the following sequents are also derivable in $\mathrm{\cfDLE}$ (resp.~$\mathrm{\cfDDLE}$)   for every $j \in J$, $h \in H$ and $k \in K$:
\begin{align*}
&\gamma'[\overline{\bsigma}/!\obx, \overline{\rtau}/!\ory, \textcolor{red}{(\xi' \lor \xi'')}[\overline{\rtau}/!\ory, \overline{\bsigma}/!\obx]/!\textcolor{red}{u}] \vdash \Gamma'_{j\,}[\overline{S}/!\ox, \overline{U}/!\oy, \Xi_{h\,}[\overline{U}/!\oy, \overline{S}/!\ox]/!u] \quad \text{~and~}\\
&\gamma'[\overline{\bsigma}/!\obx, \overline{\rtau}/!\ory, \textcolor{red}{(\xi'\lor \xi'')}[\overline{\rtau}/!\ory, \overline{\bsigma}/!\obx]/!\textcolor{red}{u}] \vdash \Gamma'_{j\,}[\overline{S}/!\ox, \overline{U}/!\oy, \Xi_{k\,}[\overline{U}/!\oy, \overline{S}/!\ox]/!u].
\end{align*}
which is enough to prove the statement, since $\gamma = \gamma' [(\xi '\lor \xi '')/!u]$, and for every $i\in I$, either $\gamma_i = \gamma'_{j}[\xi'_{h}/!u ]$ for some $j \in J$ and $h \in H$, or $\gamma_i = \gamma'_{j}[\xi'_{k}/!u ]$ for some $j \in J$ and $k \in K$.
The induction step for $n_\delta \geq 1$ is similar to the induction step above.
\end{proof}

By instantiating $\overline{\sigma\vdash S}$ and $\overline{U\vdash \tau}$ in the proposition above to identity axioms, we immediately get the following
\begin{cor}\label{prop: derivation of identities with weakening}
For any positive (resp.~negative) PIA formula $\gamma$ (resp.~$\delta$), let $\bigwedge_{i\in I}\gamma_i$ (resp.~$\bigvee_{j\in J}\delta_j$) be its equivalent rewriting as per Lemma \ref{lemma: reduction to definite}, so that each $\gamma_i$ (resp.~$\delta_j$) is definite positive (resp.~negative) PIA. Then the following sequents are derivable in $\mathrm{\cfDLE}$ (resp.~$\mathrm{\cfDDLE}$, cf.~Notation \ref{notation: structural counterparts of skeleton and pia}):
			 \begin{enumerate}
\item $\gamma\vdash \Gamma_i$ for every $i\in I$; %for every positive PIA (i.e.~negative Skeleton) formula $\gamma$;
\item $\Delta_j\vdash \delta$ for every $j\in J$, %for every negative PIA (i.e.~positive Skeleton) formula $\delta$,
\end{enumerate}	
with derivations which only consist of identity axioms, and applications of right-introduction rules for  negative PIA-connectives (cf.~Tables \ref{Join:and:Meet:Friendly:Table} and \ref{Join:and:Meet:Friendly:Table:DLE}, and Definition \ref{def: signed gen tree}), % SRR-connectives and negative unary SRA-connectives,  
and left-introduction rules for positive PIA-connectives %SRR-connectives and positive unary SRA-connectives. 
%with derivations only consisting of identity axioms and applications of left-introduction of $\wedge$, right-introduction of $\vee$,  right-introduction $f$-connectives and left-introduction of $g$-connectives 
(and weakening and exchange rules in the case of $\mathrm{\cfDDLE}$).
			\end{cor}

\begin{example}\label{ex:usedforadjoints}
The formula $\wbox (p \aand q)$ is a positive PIA in any (D)LE-language such that  $\wbox \in \mathcal{G}$, and is equivalent to $\wbox p\aand\wbox q$. Since $p\vdash  \BBOX \wdia p$ and $q \vdash  \BBOX \wdia q$ are derivable sequents 
\begin{center}
	\begin{tabular}{cc}
		\AX$p \fCenter p $
		\RL{\fns$\wdia_R$}
		\UI$\WDIA p  \fCenter \wdia p $
		\UI$p \fCenter \BBOX \wdia p $
		\DP
		& 
		\AX$q \fCenter q$
		\RL{\fns$\wdia_R$}
		\UI$\WDIA q \fCenter \wdia q$
		\UI$q \fCenter \BBOX \wdia q$
		\DP
	\end{tabular}
\end{center}
instantiating the argument in Proposition \ref{prop: generalized derivation of amlost identities}, we can derive the sequents
$\wbox (p \aand q) \fCenter \WBOX \BBOX \wdia p$ and $\wbox (p \aand q) \fCenter \WBOX \BBOX \wdia q$ in $\mathrm{D.LE}$ and in $\mathrm{D.DLE}$ as follows: 

\begin{center}
\begin{tabular}{ccccc}
\cfDLE-derivation of $\wbox (p \aand q) \fCenter \WBOX \BBOX \wdia p$:
 & & & &
\cfDDLE-derivation of $\wbox (p \aand q) \fCenter \WBOX \BBOX \wdia p$: 
\rule[-1.8mm]{0mm}{0mm} \\
\cline{0-1}
\cline{4-5}

\AXC{$\phantom{p \fCenter p}$}
\noLine
\UI$p \fCenter p $
\RL{\fns$\wdia_R$}
\UI$\WDIA p  \fCenter \wdia p $
\UI$p \fCenter \BBOX \wdia p $
%\LL{\fns$W_L$}
%\UI$p \AAND q \fCenter \BBOX \wdia p$
\LL{\fns$\aand_L$}
\UI$p \aand q \fCenter \BBOX \wdia p$
\LL{\fns$\wbox_L$}
\UI$\wbox (p \aand q) \fCenter \WBOX \BBOX \wdia p$
\DP
&&\ \ \ \ \ \ \ &&
\AX$p \fCenter p \rule[3mm]{0mm}{0mm}$
\RL{\fns$\wdia_R$}
\UI$\WDIA p  \fCenter \wdia p $
\UI$p \fCenter \BBOX \wdia p $
\LL{\fns$W_L$}
\UI$p \AAND q \fCenter \BBOX \wdia p$
\LL{\fns$\aand_L$}
\UI$p \aand q \fCenter \BBOX \wdia p$
\LL{\fns$\wbox_L$}
\UI$\wbox (p \aand q) \fCenter \WBOX \BBOX \wdia p$
\DP
 \\
&&\ \ \ \ \ \ \ &&\\
\cfDLE-derivation of $\wbox (p \aand q) \fCenter \WBOX \BBOX \wdia q$:
 & & & &
\cfDDLE-derivation of $\wbox (p \aand q) \fCenter \WBOX \BBOX \wdia q$: 
\rule[-1.8mm]{0mm}{0mm} \\
\cline{0-1}
\cline{4-5}

\AXC{$\phantom{q \fCenter q}$}
\noLine
\UIC{$\phantom{q \fCenter q}$}
\noLine
\UI$q \fCenter q$
\RL{\fns$\wdia_R$}
\UI$\WDIA q \fCenter \wdia q$
\UI$q \fCenter \BBOX \wdia q$
%\LL{\fns$W_L$}
%\UI$q \AAND p \fCenter \BBOX \wdia q$
%\LL{\fns$E_L$}
\LL{\fns$\aand_L$}
\UI$p \aand q \fCenter \BBOX \wdia q$
\LL{\fns$\wbox_L$}
\UI$\wbox (p \aand q) \fCenter \WBOX \BBOX \wdia q$
\DP
&&\ \ \ \ \ \ \ &&
\AX$q \fCenter q \rule[3mm]{0mm}{0mm}$
\RL{\fns$\wdia_R$}
\UI$\WDIA q \fCenter \wdia q$
\UI$q \fCenter \BBOX \wdia q$
\LL{\fns$W_L$}
\UI$q \AAND p \fCenter \BBOX \wdia q$
\LL{\fns$E_L$}
\UI$p \AAND q \fCenter \BBOX \wdia q$
\LL{\fns$\aand_L$}
\UI$p \aand q \fCenter \BBOX \wdia q$
\LL{\fns$\wbox_L$}
\UI$\wbox (p \aand q) \fCenter \WBOX \BBOX \wdia q$
\DP
 \\
\end{tabular}
\end{center}
\end{example}	
  In the remainder of the present section, if $\varphi (!\overline{x}, !\overline{y})$   (resp.~$\psi(!\overline{y}, !\overline{x})$) is a definite positive (resp.~negative) PIA formula, we will need to fix one variable in $\overline{x}$ or in $\overline{y}$ and make it the pivotal variable for the computation of the corresponding $\mathsf{la}(\varphi)(u, \overline{z})$ (resp.~$\mathsf{ra}(\psi)(u, \overline{z})$), where the vector $\overline{z}$ of parametric variables exactly includes all the placeholder variables in $\overline{x}$ and in $\overline{y}$ different from the pivotal one. So we write e.g.~$\varphi_x$ (resp.~$\varphi_y$) to indicate that we are choosing the pivotal variable among the variables in $\overline{x}$ (resp.~$\overline{y}$). In order to simplify the notation, we leave it to be understood that the set  of parametric variables does not contain the pivotal one, although we do not make this fact explicit in the notation.  In the remainder of the paper, we will let e.g.~$\mathsf{LA}(\varphi)(u, \overline{z})$ denote the structural counterpart of $\mathsf{la}(\varphi)(u, \overline{z})$ (cf.~Definition \ref{def: RA and LA}).		
\begin{cor}
	\label{cor: generalized deriving la almost implies atom}
	Let $\psi(!\overline{x}, !\overline{y})$ and $\xi(!\overline{y}, !\overline{x})$ be a positive and a negative PIA formula respectively, and let $\bigwedge_{i\in I}\psi_i$ and $\bigvee_{j\in J}\xi_j$ be their equivalent rewritings as per Lemma \ref{lemma: reduction to definite}, so that each $\varphi_i$ (resp.~$\xi_j$) is a definite positive (resp.~negative) PIA formula.
	Then: 
	%For all positive PIA-formula $\alpha_p$ and $\alpha_q$ and  negative PIA-formula $\beta_p$ and $\beta_q$ such that $+p\prec +\alpha_p$, $-q\prec +\alpha_q$, $+p\prec -\beta_p$, $-q\prec -\beta_q$:
	\begin{enumerate}
		\item if $\overline{\sigma\vdash S}$ and $\overline{U\vdash \tau}$ are derivable in $\mathrm{\cfDLE}$ (resp.~$\mathrm{\cfDDLE}$), then so is $\mathsf{LA}(\psi_i)[\bpsi_{\textcolor{blue}{x\,}}[\overline{\bsigma}/!\obx, \overline{\rtau}/!\ory]/!\textcolor{blue}{u}, \overline{S}/!\ox, \overline{U}/!\oy] \vdash S_{\!x\,}$, where $\psi_i$ is the definite positive PIA formula in which the pivotal variable $x$ occurs;
		\item if $\overline{\sigma\vdash S}$ and $\overline{U\vdash \tau}$ are derivable in $\mathrm{\cfDLE}$ (resp.~$\mathrm{\cfDDLE}$), then so is $U_y\vdash \mathsf{LA}(\psi_i)[\bpsi_{\textcolor{blue}{y\,}}[\overline{\bsigma}/!\obx, \overline{\rtau}/!\ory]/!\textcolor{blue}{u}, \overline{S}/!\ox, \overline{U}/!\oy]$, where $\psi_i$ is the definite positive PIA formula in which the pivotal variable $y$ occurs;
		\item if $\overline{\sigma\vdash S}$ and $\overline{U\vdash \tau}$ are derivable in $\mathrm{\cfDLE}$ (resp.~$\mathrm{\cfDDLE}$), then so is $\mathsf{RA}(\xi_j)[\rxi_{\textcolor{red}{x\,}}[\overline{\rtau}/!\ory, \overline{\bsigma}/!\obx]/!\textcolor{red}{u}, \overline{U}/!\oy, \overline{S}/!\ox] \vdash S_{\!x\,}$, where $\xi_j$ is the definite negative PIA formula in which the pivotal variable $x$ occurs;
		\item if $\overline{\sigma\vdash S}$ and $\overline{U\vdash \tau}$ are derivable in $\mathrm{\cfDLE}$ (resp.~$\mathrm{\cfDDLE}$), then so is $U_y\vdash \mathsf{RA}(\xi_j)[\rxi_{\textcolor{red}{y\,}}[\overline{\rtau}/!\ory, \overline{\bsigma}/!\obx]/!\textcolor{red}{u}, \overline{U}/!\oy, \overline{S}/!\ox]$, where $\xi_j$ is the definite negative PIA formula in which the pivotal variable $y$ occurs,
		
		%\marginnote{check if the polarities are correct in items 3 and 4 of the statement}
		
		%	\item 	The sequent $p\vdash\mathrm{INV}(p)$ is derivable in the basic calculus
		%		\item[$\epsilon(p)=\partial$] The sequent $\mathrm{INV}(p)\vdash p$ is derivable in the basic calculus 
	\end{enumerate} 
	with derivations such that, if any rules are applied other than display rules, right-introduction rules for  negative PIA-connectives (cf.~Tables \ref{Join:and:Meet:Friendly:Table} and \ref{Join:and:Meet:Friendly:Table:DLE}, and Definition \ref{def: signed gen tree}), 
and left-introduction rules for positive PIA-connectives 
(and weakening and exchange rules in the case of $\mathrm{\cfDDLE}$), then they are applied only in the derivations of $\overline{\sigma\vdash S}$ and $\overline{U\vdash \tau}$.	
	
\end{cor}
\begin{proof}
1.   Let $\Psi_x$ denote the structural counterpart of $\psi_x$ (cf.~Notation \ref{notation: structural counterparts of skeleton and pia}). The assumptions imply, by Proposition \ref{prop: generalized derivation of amlost identities}, that the sequent $\psi_{x\,}[\overline{\bsigma}/!\obx, \overline{\rtau}/!\ory]\vdash \Psi_{i\,}[\overline{S}/!\ox, \overline{U}/!\oy]$ is derivable in $\mathrm{\cfDLE}$ (resp.~$\mathrm{\cfDDLE}$) with a derivation such that, if any rules are applied other than right-introduction rules for  negative PIA-connectives 
and left-introduction rules for positive PIA-connectives 
(and weakening and exchange rules in the case of $\mathrm{\cfDDLE}$), then they are applied only in the derivations of $\overline{\sigma\vdash S}$ and $\overline{U\vdash \tau}$. Then, we can prolong this derivation by applying display rules to each node  of the branch of $\Psi_i$ leading to  the pivotal variable $x$, so as to obtain a derivation of the required sequent $\mathsf{LA}(\psi_i)[\bpsi_{\textcolor{blue}{x\,}}[\overline{\bsigma}/!\obx, \overline{\rtau}/!\ory]/\textcolor{blue}{u}, \overline{S}/!\ox, \overline{U}/!\oy]\vdash S_{\!x\,}$. The remaining items are proved similarly.
%	By induction on $\Omega$ and the complexity of $\alpha$ and $\beta$. See. \cite[Lemma 49]{GMPTZ}
\end{proof}

By instantiating $\overline{\sigma\vdash S}$ and $\overline{U\vdash \tau}$ in the corollary above to identity axioms, we immediately get the following

\begin{cor}
\label{cor: deriving la implies atom with weakening}
The following sequents are derivable in  $\mathrm{\cfDLE}$ (resp.~$\mathrm{\cfDDLE}$) for any  positive PIA (i.e.~negative Skeleton) formula $\psi(!\overline{x}, !\overline{y})$ and any  negative PIA (i.e.~positive Skeleton) formula $\xi(!\overline{y}, !\overline{x})$, such that $\bigwedge_{i\in I}\psi_i$ and $\bigvee_{j\in J}\xi_j$ are their equivalent rewritings as per Lemma \ref{lemma: reduction to definite}, so that each $\psi_i$ (resp.~$\xi_j$) is a definite positive (resp.~negative) PIA formula.
\begin{enumerate}
\item $\mathsf{LA}(\psi_i)[\bpsi_{\textcolor{blue}{x}}/\textcolor{blue}{u}]\vdash x$, where $\psi_i$ is the definite positive PIA formula in which the pivotal variable $x$ occurs;
 %for every definite positive PIA-formula $\varphi_p$ such that $+p\prec +\varphi$;
\item $y\vdash \mathsf{LA}(\psi_i)[\bpsi_{\textcolor{blue}{y}}/\textcolor{blue}{u}]$, where $\psi_i$ is the definite positive PIA formula in which the pivotal variable $y$ occurs; % for every definite positive PIA-formula $\psi_q$ such that $-q\prec +\varphi$;
		\item $\mathsf{RA}(\xi_j)[\rxi_{\textcolor{red}{x}}/\textcolor{red}{u}]\vdash x$, where $\xi_j$ is the definite negative PIA formula in which the pivotal variable $x$ occurs; % for every definite negative PIA-formula $\psi_p$ such that $+p\prec -\psi$;
		\item $y\vdash \mathsf{RA}(\xi_j)[\rxi_{\textcolor{red}{y}}/\textcolor{red}{u}]$, where $\xi_j$ is the definite negative PIA formula in which the pivotal variable $y$ occurs, % for every definite negative PIA-formula $\psi_q$ such that $-q\prec -\psi$.	
		\end{enumerate}
		with derivations which only consist of identity axioms, and applications of display rules, right-introduction rules for  negative PIA-connectives (cf.~Tables \ref{Join:and:Meet:Friendly:Table} and \ref{Join:and:Meet:Friendly:Table:DLE}, and Definition \ref{def: signed gen tree}), 
and left-introduction rules for positive PIA-connectives 
(and weakening and exchange rules in the case of $\mathrm{\cfDDLE}$).
\end{cor}	

\begin{example}
	The formula $\wbox((\wbox (p \aand q))\oplus r)$ is a positive PIA in any (D)LE-language such that  $\wbox,\oplus \in \mathcal{G}$, and is equivalent to $\wbox((\wbox p)\oplus r)\aand\wbox((\wbox q)\oplus r)$. Let $x=p$, then $\mathsf{LA}(\wbox((\wbox p)\oplus r))[\wbox((\wbox (p \aand q))\oplus r)/u]=\BDIA(\BDIA \wbox((\wbox (p \aand q))\oplus r)\ADLARR r)$, where $\ADLARR$ is the left residual of $\oplus$ on the first coordinate. As shown in Example \ref{ex:usedforadjoints} $p\vdash  \BBOX \wdia p$ and $r \vdash  \BBOX \wdia r$ are derivable sequents. Instantiating the argument in Corollary \ref{cor: generalized deriving la almost implies atom}, we can derive the sequent
	$\BDIA(\BDIA \wbox((\wbox (p \aand q))\oplus r)\ADLARR \BBOX \wdia r) \fCenter \BBOX \wdia p$ as follows: 
	
	\begin{center}
		\begin{tabular}{c}
			\cfDLE-derivation of $\BDIA(\BDIA \wbox((\wbox (p \aand q))\oplus r)\ADLARR \BBOX \wdia r) \fCenter \BBOX \wdia p$:
	\\ 	
	\hline
			\AXC{$\phantom{p \fCenter p}$}
			\noLine
			\UI$p \fCenter p $
			\RL{\fns$\wdia_R$}
			\UI$\WDIA p  \fCenter \wdia p $
			\UI$p \fCenter \BBOX \wdia p $
			%\LL{\fns$W_L$}
			%\UI$p \AAND q \fCenter \BBOX \wdia p$
			\LL{\fns$\aand_L$}
			\UI$p \aand q \fCenter \BBOX \wdia p$
			\LL{\fns$\wbox_L$}
			\UI$\wbox (p \aand q) \fCenter \WBOX \BBOX \wdia p$
			\AXC{$\phantom{r \fCenter r}$}
			\noLine
			\UI$r \fCenter r $
			\RL{\fns$\wdia_R$}
			\UI$\WDIA r  \fCenter \wdia r $
			\UI$r \fCenter \BBOX \wdia r $
			\LL{\fns$\oplus_L$}
			\BI$\wbox (p \aand q)\oplus r\fCenter  \WBOX \BBOX \wdia p \MOR  \BBOX \wdia r$
			\LL{\fns$\wbox_L$}
			\UI$\wbox(\wbox (p \aand q)\oplus r)\fCenter \WBOX( \WBOX \BBOX \wdia p \MOR  \BBOX \wdia r)$
			\UI$\BDIA\wbox(\wbox (p \aand q)\oplus r)\fCenter  \WBOX \BBOX \wdia p \MOR  \BBOX \wdia r$
			\UI$\BDIA\wbox(\wbox (p \aand q)\oplus r)\ADLARR \BBOX \wdia r\fCenter \WBOX \BBOX \wdia p$
			\UI$\BDIA(\BDIA\wbox(\wbox (p \aand q)\oplus r)\ADLARR \BBOX \wdia r)\fCenter \BBOX \wdia p$
			\DP
		\\ 
		\\
		\\
			\cfDDLE-derivation of $\BDIA(\BDIA \wbox((\wbox (p \aand q))\oplus r)\ADLARR \BBOX \wdia r) \fCenter \BBOX \wdia p$: 
			\\\hline
			\AXC{$\phantom{p \fCenter p}$}
		\noLine
			\UI$p \fCenter p $
			\RL{\fns$\wdia_R$}
			\UI$\WDIA p  \fCenter \wdia p $
			\UI$p \fCenter \BBOX \wdia p $
			\LL{\fns$W_L$}
			\UI$p \AAND q \fCenter \BBOX \wdia p$
			\LL{\fns$\aand_L$}
			\UI$p \aand q \fCenter \BBOX \wdia p$
			\LL{\fns$\wbox_L$}
			\UI$\wbox (p \aand q) \fCenter \WBOX \BBOX \wdia p$
			\AXC{$\phantom{r \fCenter r}$}
			\noLine
			\UI$r \fCenter r $
			\RL{\fns$\wdia_R$}
			\UI$\WDIA r  \fCenter \wdia r $
			\UI$r \fCenter \BBOX \wdia r $
			\LL{\fns$\oplus_L$}
			\BI$\wbox (p \aand q)\oplus r\fCenter  \WBOX \BBOX \wdia p \MOR  \BBOX \wdia r$
			\LL{\fns$\wbox_L$}
			\UI$\wbox(\wbox (p \aand q)\oplus r)\fCenter \WBOX( \WBOX \BBOX \wdia p \MOR  \BBOX \wdia r)$
			\UI$\BDIA\wbox(\wbox (p \aand q)\oplus r)\fCenter  \WBOX \BBOX \wdia p \MOR  \BBOX \wdia r$
			\UI$\BDIA\wbox(\wbox (p \aand q)\oplus r)\ADLARR \BBOX \wdia r\fCenter \WBOX \BBOX \wdia p$
			\UI$\BDIA(\BDIA\wbox(\wbox (p \aand q)\oplus r)\ADLARR \BBOX \wdia r)\fCenter \BBOX \wdia p$
			\DP	
		\end{tabular}
	\end{center}
\end{example}	

\begin{prop}\label{prop: the thing needed for the non-definite}
	%The  calculi $\mathbf{DL}$ derive the following sequents (cf.~Notation \ref{notation: structural counterparts of skeleton and pia}):
	Let $\varphi = \varphi (!\overline{x}, !\overline{y})$ and $\psi = \psi(!\overline{y}, !\overline{x})$ be a positive and a negative Skeleton formula, respectively, and let $\bigvee_{j\in J}\varphi_j$ and $\bigwedge_{i\in I}\psi_i$  be their equivalent rewritings as per Lemma \ref{lemma: reduction to definite}, so that each $\varphi_j$  (resp.~each $\psi_i$) is definite positive (resp.~negative) Skeleton. Then: 
	\begin{enumerate}
	\item if $ \Phi_{j\,}[\overline{\bsigma}/!\obx, \overline{\rtau}/!\ory]\vdash \Sigma$ is derivable in $\mathrm{\cfDLE}$ (resp.~$\mathrm{\cfDDLE}$) for every $j\in J$, then so is $\varphi[\overline{\bsigma}/!\obx, \overline{\rtau}/!\ory] \vdash \Sigma$;
		\item if $\Pi\vdash \Psi_{i\,}[\overline{\rtau}/!\ory, \overline{\bsigma}/!\obx]$ is derivable in $\mathrm{\cfDLE}$ (resp.~$\mathrm{\cfDDLE}$) for every $i\in I$, then so is $\Pi\vdash \psi[\overline{\rtau}/!\ory, \overline{\bsigma}/!\obx]$,
		%\item $\Delta\vdash \delta$ for every negative PIA (i.e.~positive Skeleton) formula $\delta$,
	\end{enumerate}	
with derivations such that, if any rules are applied other than display rules,  left-introduction rules for positive Skeleton-connectives (cf.~Tables \ref{Join:and:Meet:Friendly:Table} and \ref{Join:and:Meet:Friendly:Table:DLE}, and Definition \ref{def: signed gen tree}), right-introduction rules for negative Skeleton-connectives, 
(and contraction  in the case of $\mathrm{\cfDDLE}$), then they are applied only in the derivations of  $ \Phi_{j\,}[\overline{\bsigma}/!\obx, \overline{\rtau}/!\ory]\vdash \Sigma$ and $\Pi\vdash \Psi_{i\,}[\overline{\rtau}/!\ory, \overline{\bsigma}/!\obx]$.	
\end{prop}

\begin{proof}
	Let $n_\varphi(+\lor)$  (resp.~$n_\psi(+\lor)$) be the number of occurrences of $+\lor$ in $+\varphi$  (resp.~$-\psi$), and let  $n_\varphi(-\land)$ (resp.~$n_\psi(-\land)$) be the number of occurrences of $-\land$  in $+\varphi$ (resp.~$-\psi$). The proof is by simultaneous induction on $n_\varphi = n_\varphi(+\lor) +  n_\varphi(-\land)$ and $n_\psi = n_\psi(+\lor) +n_\psi(-\land)$. 
	
	If  $n_\psi = n_\varphi = 0$, then $\varphi$ (resp.~$\psi$) is definite positive (resp.~negative) Skeleton. Then from a derivation of $\Phi[\overline{\bsigma}/!\obx, \overline{\rtau}/!\ory]\vdash\Sigma$ (resp.~$\Pi\vdash \Psi[\overline{\rtau}/!\ory, \overline{\bsigma}/!\obx]$) we obtain a derivation of $\varphi[\overline{\bsigma}/!\obx, \overline{\rtau}/!\ory]\vdash \Sigma$ (resp.~$\Pi \vdash \psi[\overline{\rtau}/!\ory, \overline{\bsigma}/!\obx]$) by applications of left-introduction rules for positive SLR-connectives, and right-introduction rules for negative SLR-connectives, interleaved with applications of display rules.

	If $n_\psi \geq 1$, then let us consider one %of the topmost 
	occurrence of $-\land$ or $+\lor$ in $-\psi$, which we will refer to as `the focal occurrence'. 
%
	%The case that $+\land $ is not the root of $+\gamma$ is similar to the case that $-\lor$ is not the root of $+\gamma$. To show how claim 1 and 2 in this proposition interacts, we give a detailed proof of  the case that $-\lor$ is not the root of $+\gamma$. Note that by definition, a topmost occurrence of $-\lor$ in $+\gamma$ is not the root of $+\gamma$.
	Let us assume that the focal  occurrence of $-\land$ or $+\lor$ in $-\psi$ %is not the root of $+\gamma$, and that it 
	is an occurrence of $+\vee$ (the case in which it is an occurrence of $-\land$ is argued similarly).
	Let $+\xi'$ and $+\xi''$ be the two subtrees under the focal occurrence of  $+\lor $. Then $\xi '\lor \xi''$ is a subformula of $\psi$ such that $\xi'$ and $\xi ''$ are positive Skeleton formulas, and $n_{\xi'}$ and $n_{\xi''}$ are strictly smaller than $n_{\psi}$. Let $u$ be a fresh variable which does not occur in $\psi$, and let $\psi'$ be the formula obtained by substituting the occurrence of $\xi'\lor \xi''$ in $\psi$ with $u$.  Then $\psi'$ is a negative Skeleton formula such that $n_{\psi'} $ is strictly smaller than $n_\psi$, and $\psi = \psi' [(\xi '\lor \xi '')/ !u ]$. 
	
	Let $\bigwedge_{i\in I}\psi_i$, $\bigwedge_{j\in J}\psi'_{j}$, 
	$\bigvee_{h\in H}\xi'_{h}$ and $\bigvee_{k\in K}\xi''_{k}$ be the equivalent rewritings of $\psi$, $\psi'$, $\xi'$ and $\xi''$, respectively, resulting from applying Lemma \ref{lemma: reduction to definite} to $\psi$, $\psi'$, $\xi'$ and $\xi''$, respectively. Then, 
	\[\{\psi_i\mid i\in I\} = \{\psi'_{j}[\xi'_{h}/!u ] \mid j \in J \text{ and } h \in H\}\cup \{\psi'_{j}[\xi''_{k}/!u ] \mid j \in J \text{ and } k \in K\}.\]
	%for every $i\in I$, 
	%\begin{itemize}
%		\item either $\gamma_i = \gamma'_{j}[\xi'_{h}/!u ]$ for some $j \in J$ and $h \in H$, in which case, the assumption can be equivalently reformulated as the statement that $\Pi\vdash\Gamma'_{j}[\Xi'_{h}/!u]$ is derivable;
%		\item or $\gamma_i = \gamma'_{j}[\xi''_{k}/!u ]$ for some $j \in J$ and $k \in K$,  in which case, the assumption can be equivalently reformulated as the statement that $\Pi\vdash\Gamma'_{j}[\Xi''_{k}/!u]$ is derivable.
%	\end{itemize}
	% In either case, 
	Hence, the assumptions can be equivalently reformulated as the following sequents being   derivable in $\mathrm{\cfDLE}$ (resp.~$\mathrm{\cfDDLE}$) for every $j\in J$, $h\in H$, and $k\in K$:
	\[\Pi\vdash\Psi'_{j}[\Xi'_h/!u]\quad\quad\Pi\vdash\Psi'_{j}[\Xi''_k/!u].\]
By prolonging those derivations with consecutive applications of display rules, we obtain derivations in $\mathrm{\cfDLE}$ (resp.~$\mathrm{\cfDDLE}$) of the following sequents, for every $j\in J$, $h\in H$, and $k\in K$:
	% \begin{itemize}
	 %	\item  $
	 \[\Xi'_{h}\vdash\mathsf{LA}(\psi'_{j})[\Pi/!v] \quad\quad %$ is derivable and 
	 	%\item  $
		\Xi''_{k}\vdash\mathsf{LA}(\psi'_{j})[\Pi/!v].\] %$ is derivable.
	% \end{itemize}
	  Hence, by induction hypothesis on $\xi'$ and $\xi''$, the following sequents are derivable in $\mathrm{\cfDLE}$ (resp.~$\mathrm{\cfDDLE}$) for every $j \in J$:
	% \begin{itemize}
	 	%\item 
		\[\xi'\vdash \mathsf{LA}(\psi'_{j})[\Pi/!v]\quad\quad
	 	%\item $
		\xi''\vdash \mathsf{LA}(\psi'_{j})[\Pi/!v].\] %$.
	% \end{itemize}
	Then, by prolonging the derivations of the two sequents above with suitable applications of $(\lor_{L})$, we obtain   derivations in $\mathrm{\cfDLE}$ (resp.~$\mathrm{\cfDDLE}$)\footnote{In the calculus $\mathrm{\cfDDLE}$, we obtain the derivations of the sequents in \eqref{eq:der sequents} by replacing the  applications of the derivable rules $(\lor_{L})$ with the derivations of those applications as shown in Remark \ref{rem: derivable rules}, thereby implementing the general strategy outlined in Remark \ref{remark:distr}, which in this specific case involves the application of structural contraction rule.} of the following sequents  for every $j \in J$:

	\begin{center}
                \begin{tabular}{c}
					\AX$\xi'[\overline{\bsigma}/!\obx, \overline{\rtau}/!\ory]\fCenter \mathsf{LA}(\psi'_{j})[\Pi/!v,\overline{\rtau}/!\ory, \overline{\bsigma}/!\obx]$
					\AX$\xi''[\overline{\bsigma}/!\obx, \overline{\rtau}/!\ory]\fCenter \mathsf{LA}(\psi'_{j})[\Pi/!v,\overline{\rtau}/!\ory, \overline{\bsigma}/!\obx]$
					\RL{\fns$\lor_{L}$}
					\BI$(\xi'\lor\xi'')[\overline{\bsigma}/!\obx, \overline{\rtau}/!\ory]\fCenter \mathsf{LA}(\psi'_{j})[\Pi/!v,\overline{\rtau}/!\ory, \overline{\bsigma}/!\obx]$
					\DP 
					\\
				\end{tabular}
	\end{center}
	By prolonging the derivations above with consecutive applications of display rules we obtain derivations in $\mathrm{\cfDLE}$ (resp.~$\mathrm{\cfDDLE}$) of the following sequents  for every $j \in J$:
$$\Pi\vdash\Psi'_{j}[\textcolor{blue}{(\xi_1\lor\xi_2)}[\overline{\bsigma}/!\obx, \overline{\rtau}/!\ory]/!\textcolor{blue}{v},\overline{\rtau}/!\ory, \overline{\bsigma}/!\obx].$$ 
	By induction hypothesis on $\psi'$, and recalling that $\psi = \psi' [(\xi '\lor \xi '')/ !u ]$, we can conclude that $\Pi \vdash \psi[\overline{\rtau}/!\ory, \overline{\bsigma}/!\obx]$ is derivable, as required.
\end{proof}

\begin{example}The formula $\wdia (p \lor \wdia p)$ is a negative PIA in any (D)LE-language such that  $\wdia \in \mathcal{F}$, and is equivalent to $\wdia p\lor\wdia\wdia p$. Assuming that $\WDIA p\vdash \WBOX \wdia p$ and $\WDIA\WDIA p\vdash \WBOX \wdia p$ are derivable sequents, instantiating the argument in Proposition \ref{prop: the thing needed for the non-definite}, we can derive the sequent
	$\wdia(p \lor \wdia p)\vdash \WBOX \wdia p$ in $\mathrm{\cfDLE}$ as follows:
\begin{center}
\AX$\WDIA p \fCenter\WBOX \wdia p$
\LL{\fns $\WDIA \dashv \BBOX$}
\UI$p \fCenter \BBOX\WBOX \wdia p$
\AX$\WDIA\WDIA p \fCenter\WBOX \wdia p$
\LL{\fns $\WDIA \dashv \BBOX$}
\UI$\WDIA p \fCenter \BBOX\WBOX \wdia p$
\LL{\fns{$\wdia_L$}}
\UI$\wdia p \fCenter\BBOX\WBOX \wdia p$
\LL{\fns{$\lor_L$}}
\BI$p \lor \wdia p \fCenter\BBOX\WBOX \wdia p$
\LL{\fns $\WDIA \dashv \BBOX$}
\UI$\WDIA(p \lor \wdia p)\fCenter \WBOX \wdia p$
\LL{\fns{$\wdia_L$}}
\UI$\wdia(p \lor \wdia p)\fCenter \WBOX \wdia p$
\DP
\end{center}
\end{example}

\section{Syntactic completeness}
\label{sec:syntactic completeness}

In the present section, we fix an arbitrary LE-language $\mathcal{L}_{\mathrm{LE}}$, for which we prove our main result (cf.~Theorem \ref{prop:synt-compl}), via an effective procedure which generates cut-free derivations in pre-normal form (cf.~Definitions \ref{def:canonical form nonDist} and \ref{def:canonical form Dist}) of any analytic inductive $\mathcal{L}_{\mathrm{LE}}$-sequent   in    $\mathrm{\cfDLE}$ (resp.~$\mathrm{\cfDDLE}$) augmented with the analytic structural rule(s) corresponding to the given sequent. In Section \ref{ssec: syntactic completeness special}, we
 will first illustrate some of the main ideas of the proof in the context of a proper subclass of analytic inductive sequents, which we refer to as \emph{quasi-special inductive}. Then in Section \ref{ssec: syntactic completeness general}, we state and prove this result for arbitrary analytic inductive sequents.

\begin{notation} In this section, we will often deal with vectors of formulas $\overline{\gamma}$ and $\overline{\delta}$ such that  each $\gamma$ in  $\overline{\gamma}$ (resp.~$\delta$ in $\overline{\delta}$) is a positive (resp.~negative) PIA formula, and hence, by Lemma \ref{lemma: reduction to definite}, is equivalent to $\bigwedge_{\lambda}\gamma^\lambda$ (resp.~$\bigvee_\mu\delta^\mu$). To avoid overloading notation, we will slightly abuse it and write $\overline{\gamma^\lambda}$ (resp.~$\overline{\delta^\mu}$), understanding that,  for each element of these vectors, each $\lambda$ and $\mu$ range over different sets.
\end{notation}

\subsection{Syntactic completeness for quasi-special inductive sequents}	
\label{ssec: syntactic completeness special}

\begin{definition}\label{def:quasianalyticinductive}
 For every analytic $(\Omega,\varepsilon)$-inductive inequality $s\leq t$, if every $\varepsilon$-critical branch of the signed generation trees $+s$ and $-t$ consists solely of Skeleton nodes, then $s\leq t$ is a \emph{quasi-special inductive inequality}. Such an inequality is {\em definite} if none of its Skeleton nodes is $+\vee $ or $-\wedge$.
\end{definition}
In terms of the convention introduced in Notation \ref{notation: analytic inductive}, quasi special inductive sequents can be represented as $(\varphi\fCenter \psi)[\obp, \orq,\overline{\bgamma}, \overline{\rdelta}]$, i.e.~as those $(\varphi\fCenter \psi)[\oba, \orb,\overline{\bgamma}, \overline{\rdelta}]$ such that each $\alpha$ in $\overline{\alpha}$ and $\beta$ in $\overline{\beta}$ is an atomic proposition. 

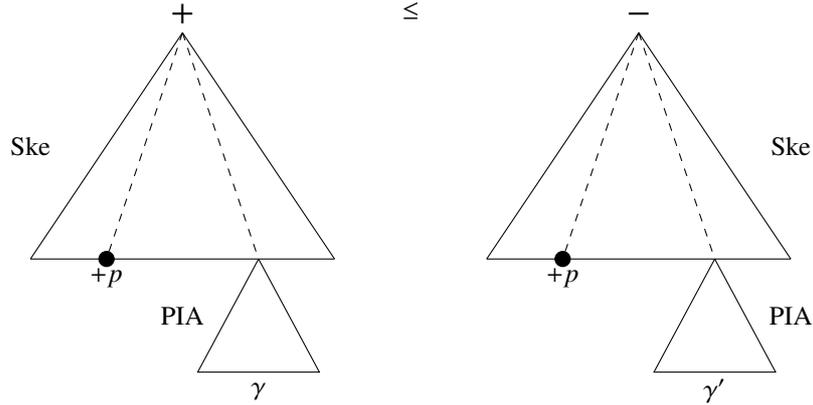
\begin{center}
	\begin{tikzpicture}
	\draw (-5,-1.5) -- (-3,1.5) node[above]{\Large$+$} ;
	\draw (-5,-1.5) -- (-1,-1.5) ;
	\draw (-3,1.5) -- (-1,-1.5);
	\draw (-5,0) node{Ske} ;
	\draw[dashed] (-3,1.5) -- (-4,-1.5);
	\draw[dashed] (-3,1.5) -- (-2,-1.5);
	\draw[fill] (-4,-1.5) circle[radius=.1] node[below]{$+p$};
	\draw
	(-2,-1.5) -- (-2.8,-3) -- (-1.2,-3) -- (-2,-1.5);
	%\fill[pattern=north east lines](-2,-1.5) -- (-2.8,-3) -- (-1.2,-3);
	\draw (-2,-3.25)node{$\gamma$};
	\draw (-3,-2.25) node{PIA} ;
	\draw (0,1.8) node{$\leq$};
	\draw (5,-1.5) -- (3,1.5) node[above]{\Large$-$} ;
	\draw (5,-1.5) -- (1,-1.5) ;
	\draw (3,1.5) -- (1,-1.5);
	\draw (5,0) node{Ske} ;
	\draw[dashed] (3,1.5) -- (4,-1.5);
	\draw[dashed] (3,1.5) -- (2,-1.5);
	\draw[fill] (2,-1.5) circle[radius=.1] node[below]{$+p$};
	\draw
	(4,-1.5) -- (4.8,-3) -- (3.2,-3) -- (4, -1.5);
	%\fill[pattern=north east lines](4,-1.5) -- (4.8,-3) -- (3.2,-3) -- (4, -1.5);
	\draw (4,-3.25)node{$\gamma'$};
	\draw (5,-2.25) node{PIA} ;
	\end{tikzpicture}
\end{center}

%%*****%%%%

\begin{example}\label{example: quasi special inductive} 
Let $\mathcal{L}: = \mathcal{L}(\mathcal{F}, \mathcal{G})$, where $\mathcal{F}: = \{\aand,\mand,\wdia\}$ and $\mathcal{G}: = \{\aor,\mor,\wbox\}$. The $\mathcal{L}_\mathrm{LE}$-inequality 	$\Diamond p \le \Box \Diamond p $, known in the modal logic literature as axiom 5,
  is a definite quasi-special inductive inequality, e.g.\ for $<_\Omega \,= \emptyset$ and $\epsilon(p) =1$, as can be seen from the signed generation tree below (see Notation \ref{notation: representations of signed generation trees}): % 
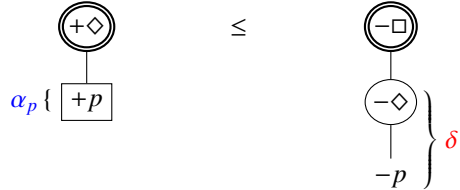
		\begin{center}
				\begin{tikzpicture}
				\tikzstyle{level 1}=[level distance=1cm, sibling distance=2.5cm]
				\tikzstyle{level 2}=[level distance=1cm, sibling distance=1.5cm]
				\tikzstyle{level 3}=[level distance=1 cm, sibling distance=1.5cm]
				\node[Ske] at (-2,0) {$\begin{aligned} +\wdia \end{aligned}$}
				child{node[draw]{$+p$}};
              \node at (0,0) {$\le$}; 
				\node[Ske] at (2,0) {$\begin{aligned} -\wbox \end{aligned}$}
				child{node[PIA]{$\begin{aligned} -\wdia \end{aligned}$}
				child{node{$-p$}}};
\node[rotate = +90] at (2.5, -1.5) {$\underbrace{\hspace{1.3cm}}$};
            % \draw[help lines] (-4,-4) grid (4,4);
             \node at (2.8,-1.5) {$\rdelta$};
               %\node [rotate =-90] at (-2.5,-1) {$\underbrace{}$};
          \node at (-2.7,-1) {$\textcolor{blue}{\alpha_p}\,\{\,$};
				\end{tikzpicture}
					%\caption{Signed generation tree for $\Box (p \rightarrow  q) \to \Box (\Box p\to\Box q)$}
			%\label{fig:fisher-servi}
		\end{center}
	The $\mathcal{L}_\mathrm{LE}$-inequality  $p\mand (p\mand \Box q) \le q \mor (q \mor \Diamond p)$ is a definite quasi-special inductive inequalities, e.g.\ for $p <_\Omega q $ and $\epsilon(p,q) =(1,\partial)$, as can be seen from the signed generation tree below (cf.~Notation \ref{notation: representations of signed generation trees}):	

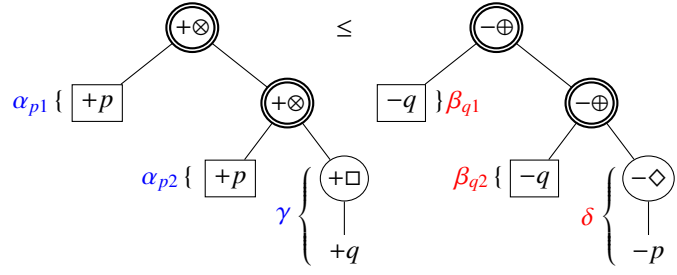
\begin{center}
	\begin{tikzpicture}
		\node{
		
			\begin{tikzpicture}
				\tikzstyle{level 1}=[level distance=1cm, sibling distance=2.5cm]
				\tikzstyle{level 2}=[level distance=1cm, sibling distance=1.5cm]
				\tikzstyle{level 3}=[level distance=1 cm, sibling distance=1.5cm]
				\node[Ske] at (-2,0) {$\begin{aligned} +\mand \end{aligned}$}
				child{node[draw]{$+p$}}          
				child{node[Ske]{$\begin{aligned} +\mand \end{aligned}$}
					child{node[draw]{$+p$}}
					child{node[PIA]{$\begin{aligned} +\wbox \end{aligned}$}
						child{node{$+q$}}
					}
				};
				\node at (0,0) {$\le$}; 
				
				\node[Ske] at (2,0) {$\begin{aligned} -\mor \end{aligned}$}
				child{node[draw]{$-q$}} 
				child{node[Ske] {$\begin{aligned} -\mor \end{aligned}$}      
					child{node[draw]{$-q$}}
					child{node[PIA]{$\begin{aligned} -\wdia \end{aligned}$}
						child{node{$-p$}}
					}
				};
				
				\node[rotate = -90] at (-0.5, -2.5) {$\underbrace{\hspace{1.3cm}}$};
				\node at (-0.8,-2.5) {$\bgamma$};
				% \draw[help lines] (-4,-4) grid (4,4);
				\node[rotate = -90] at (3.5, -2.5) {$\underbrace{\hspace{1.3cm}}$};
				\node at (3.2,-2.5) {$\rdelta$};
				\node at (-4,-1) {$\textcolor{blue}{\alpha_{p1}}\,\{\,$};
				\node at (-2.3,-2) {$\textcolor{blue}{\alpha_{p2}}\,\{\,$};
				\node at (1.5,-1) {$\}\, \textcolor{red}{\beta_{q1}}$};
				\node at (1.8,-2) {$\textcolor{red}{\beta_{q2}}\,\{\,$};
				
			\end{tikzpicture}
		};
		%\node at (0,-5.5) {\footnotesize{Representations of example \ref{example: quasi special inductive} accordingly to notation \ref{notation: representations of signed generation trees}.}};
	\end{tikzpicture}
	%\caption{Signed generation tree for $\Box (p \rightarrow  q) \to \Box (\Box p\to\Box q)$}
	%\label{fig:fisher-servi}
\end{center}		

The $\mathcal{L}_\mathrm{LE}$-inequality 	$\Diamond (p\lor \Diamond p) \le \Box \Diamond p $ is a \emph{non-definite} quasi-special inductive inequality, e.g.\ for $<_\Omega \,= \emptyset$ and $\epsilon(p) =1$, as can be seen from the signed generation tree below (see Notation \ref{notation: representations of signed generation trees}): % 
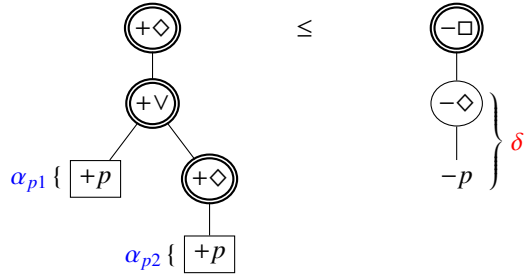
\begin{center}
	\begin{tikzpicture}
		\tikzstyle{level 1}=[level distance=1cm, sibling distance=2.5cm]
		\tikzstyle{level 2}=[level distance=1cm, sibling distance=1.5cm]
		\tikzstyle{level 3}=[level distance=1 cm, sibling distance=1.5cm]
		\node[Ske] at (-2,0) {$\begin{aligned} +\wdia \end{aligned}$}
		child{node[Ske]{$\begin{aligned} +\lor \end{aligned}$}
			child{node[draw]{$+p$}}
			child{node[Ske]{$\begin{aligned} +\wdia \end{aligned}$}
				child{node[draw]{$+p$}}}};
		\node at (0,0) {$\le$}; 
		\node[Ske] at (2,0) {$\begin{aligned} -\wbox \end{aligned}$}
		child{node[PIA]{$\begin{aligned} -\wdia \end{aligned}$}
			child{node{$-p$}}};
		\node[rotate = +90] at (2.5, -1.5) {$\underbrace{\hspace{1.3cm}}$};
		% \draw[help lines] (-4,-4) grid (4,4);
		\node at (2.8,-1.5) {$\rdelta$};
		%\node [rotate =-90] at (-2.5,-1) {$\underbrace{}$};
		\node at (-3.5,-2) {$\textcolor{blue}{\alpha_{p1}}\,\{\,$};
		\node at (-2,-3.03) {$\textcolor{blue}{\alpha_{p2}}\,\{\,$};
	\end{tikzpicture}
	%\caption{Signed generation tree for $\Box (p \rightarrow  q) \to \Box (\Box p\to\Box q)$}
	%\label{fig:fisher-servi}
\end{center}

Finally, in the distributive case, the $\mathcal{L}_\mathrm{DLE}$-inequality $p\land \Box q \le q \lor \Diamond p$, is a definite quasi-special inductive inequalities, e.g.\ for $p <_\Omega q $ and $\epsilon(p,q) =(1,\partial)$, as can be seen from the signed generation tree below (cf.~Notation \ref{notation: representations of signed generation trees}):	

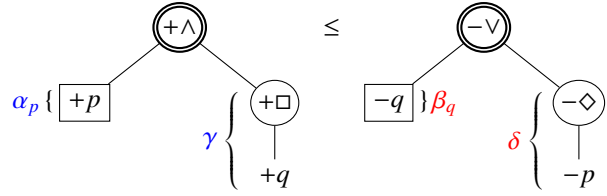
\begin{center}
	\begin{tikzpicture}
		\node{
			\begin{tikzpicture}
				\tikzstyle{level 1}=[level distance=1cm, sibling distance=2.5cm]
				\tikzstyle{level 2}=[level distance=1cm, sibling distance=1.5cm]
				\tikzstyle{level 3}=[level distance=1 cm, sibling distance=1.5cm]
				\node[Ske] at (-2,0) {$\begin{aligned} +\aand \end{aligned}$}
				child{node[draw]{$+p$}}
				child{node[PIA]{$\begin{aligned} +\wbox \end{aligned}$}
					child{node{$+q$}}
				}
				;
				\node at (0,0) {$\le$}; 
				
				\node[Ske] at (2,0) {$\begin{aligned} -\aor \end{aligned}$}
				child{node[draw]{$-q$}}
				child{node[PIA]{$\begin{aligned} -\wdia \end{aligned}$}
					child{node{$-p$}}
				}
				;
				\node[rotate = -90] at (2.7, -1.5) {$\underbrace{\hspace{1.3cm}}$};
				\node at (2.4 ,-1.5) {$\rdelta$};
				%  \draw[help lines] (-6,-6) grid (6,6);
				\node[rotate = -90] at (-1.3, -1.5) {$\underbrace{\hspace{1.3cm}}$};
				\node at (-1.6,-1.5) {$\bgamma$};
				\node at (-3.9,-1) {$\textcolor{blue}{\alpha_p}\,\{\,$};
				\node at (1.4,-1) {$\}\, \textcolor{red}{\beta_q}$};
				% \draw[help lines] (-4,-4) grid (4,4);
			\end{tikzpicture}
		};
	\end{tikzpicture}
	%\caption{Signed generation tree for $\Box (p \rightarrow  q) \to \Box (\Box p\to\Box q)$}
	%\label{fig:fisher-servi}
\end{center}
\end{example}
%Definite quasi-special inductive inequalities and quasi-special rules entertain the same privileged relation with each other as the one entertained by definite primitive inequalities and special rules. Indeed, translating into an inequality the rule obtained from a definite quasi-special inductive inequality leads to the original inequality (cf.\ Remark \ref{remark:quasi-special}). Notice that these are exactly the inequalities that have this property, since the inequality that is obtained by Proposition \ref{prop:rulestoineq} is always definite quasi-special inductive. Since every analytic inductive inequality is equivalent to a set of analytic rules (in fact quasi-special rules) and every analytic rule is equivalent to a definite quasi-special inductive inequality, is it clear that every analytic inductive inequality is equivalent to a set of definite quasi-special inductive inequalities.

%%%***
	
\begin{lemma}
\label{lemma: quasi special inductive corresponding rule}
If $(\varphi\fCenter\psi)[\obp, \orq,\overline{\bgamma}, \overline{\rdelta}]$ is a quasi-special inductive inequality (cf.~Notation \ref{notation: analytic inductive}), then each of its corresponding rules has the following shape: 
\begin{center}
\AXC{$\overline{(Z\vdash \Gamma^\lambda)_\lambda}[\overline{X}/\overline{p}, \overline{Y}/\overline{q}]$}
\AXC{$\overline{(\Delta^\mu\vdash W)_{\mu}}[\overline{X}/\overline{p}, \overline{Y}/\overline{q}]$}
\BIC{$(\Phi_j\vdash \Psi_i)[\overline{X}, \overline{Y}, \overline{Z}, \overline{W}]$}
\DP
\end{center}
where for each $\gamma$ in $\overline{\gamma}$ (resp.~each $\delta$ in $\overline{\delta}$), each $\Gamma^\lambda$ (resp.~$\Delta^\mu$) is the structural counterpart of some conjunct (resp.~disjunct) $\gamma^\lambda$ (resp.~$\delta^\mu$) of the equivalent rewriting of $\gamma$ (resp.~$\delta$) as $\bigwedge_{\lambda}\gamma^\lambda$ (resp.~$\bigvee_\mu\delta^\mu$),  as per Lemma \ref{lemma: reduction to definite},  %Since  $\gamma$ (resp.~$\delta$) is positive (resp.~negative) PIA, 
with each $\gamma^\lambda$ (resp.~$\delta^\mu$) being a definite positive (resp.~negative) PIA formula, and each $\Phi_j$ (resp.~$\Psi_i$) is the structural counterpart of some disjunct (resp.~conjunct) $\varphi_j$ (resp.~$\psi_i$) of the equivalent rewriting of $\varphi$ (resp.~$\psi$) as $\bigvee_{j\in J}\varphi_j$ (resp.~$\bigwedge_{i\in I}\psi_i$),  as per Lemma \ref{lemma: reduction to definite},  %Since  $\gamma$ (resp.~$\delta$) is positive (resp.~negative) PIA, 
with each $\varphi_j$ (resp.~$\psi_i$) being a definite positive (resp.~negative) Skeleton formula.
\end{lemma}
\begin{proof}
Let us apply the algorithm ALBA to $(\varphi\fCenter\psi)[\obp, \orq,\overline{\bgamma}, \overline{\rdelta}]$ to compute its corresponding analytic rules. Modulo pre-processing, we can assume w.l.o.g.~that $(\varphi\fCenter\psi)[\obp, \orq,\overline{\bgamma}, \overline{\rdelta}]$  is  definite,\footnote{If $(\varphi\fCenter\psi)[\obp, \orq,\overline{\bgamma}, \overline{\rdelta}]$ is not definite, then by Lemma \ref{lemma: from inductive to definite sequents} it can equivalently be transformed  into the conjunction of {\em definite} quasi-special inductive sequents which we can treat separately as shown in the proof.}  and hence we can proceed with first approximation:
\begin{equation}
\label{eq: first approx}
\forall \overline{xyzw}[(\overline{x\vdash p}\ \&\  \overline{q\vdash y}\ \&\ \overline{z \vdash \gamma}\ \&\ \overline{\delta\vdash w})\Rightarrow (\varphi\vdash\psi)[\obx, \ory,\overline{\bz}, \overline{\rw}]]
\end{equation} Since $\gamma$ (resp.~$\delta$) can be equivalently rewritten as $\bigwedge_{\lambda}\gamma^\lambda$ (resp.~$\bigvee_\mu\delta^\mu$), as per Lemma \ref{lemma: reduction to definite}, the quasi-inequality above can be equivalently rewritten as follows: \begin{equation}
	\label{eq: first approx2}
	\forall \overline{xyzw}\,[(\overline{x\vdash p}\ \&\  \overline{q\vdash y}\ \&\ \overline{(z\vdash \gamma^\lambda)_{\lambda}}\ \&\  \overline{(\delta^\mu\vdash w)_{\mu}})\Rightarrow (\varphi\vdash\psi)[\obx, \ory,\overline{\bz}, \overline{\rw}]].
	\end{equation}
If every $p$ in $\overline{p}$ and $q$ in $\overline{q}$ has  one critical occurrence, then we are in Ackermann-shape and hence we can eliminate the variables $\overline{p}$ and $\overline{q}$ as follows (since by assumption $\overline{\gamma}$ and $\overline{\delta}$ agree with  $\epsilon^\partial$):
\begin{equation}
\label{eq: application Ackermann2}
\forall \overline{xyzw}\,[( \overline{(z\vdash \gamma^\lambda)_{\lambda}} [\obx/\obp, \ory/\orq]\ \&\  \overline{(\delta^\mu\vdash w)_{\mu}}[\obx/\obp, \ory/\orq])\Rightarrow (\varphi\vdash\psi)[\obx, \ory,\overline{\bz}, \overline{\rw}]]
\end{equation}
which yields a rule of the desired shape. If there are multiple critical occurrences of  some $p$ in $\overline{p}$ or $q$ in $\overline{q}$, then the Ackermann-shape looks as in \eqref{eq: first approx2}, but with $\bigvee_{k = 1}^{n_i} x_k\vdash p_i$ and $q_j\vdash \bigwedge_{h = 1}^{m_j} y_h$, where $n_i$ (resp.~$m_j$) is the number of critical occurrences of $p_i$ (resp.~$q_j$). Hence, by applying the Ackermann rule we obtain a quasi-inequality similar to \eqref{eq: application Ackermann2}, except that the sequents in the antecedent have the following shape: 
\begin{equation}
\label{eq: after Ackermann2}
\overline{(z\vdash \gamma^\lambda)_\lambda}\left[\overline{\bigvee_{k = 1}^{n_i} \bx_{\bk}/\bp_{\bi}}, \overline{\bigwedge_{h = 1}^{m_j} \ry_{\rh}/\rrq_{\rj}}\right ]\quad \quad \overline{(\delta^\mu\vdash w)_\mu}\left [\overline{\bigvee_{k = 1}^{n_i} \bx_{\bk}/\bp_{\bi}}, \overline{\bigwedge_{h = 1}^{m_j} \ry_{\rh}/\rrq_{\rj}}\right].\end{equation}

Since by assumption $\varepsilon(p) = 1$ for every $p$ in $\overline{p}$ and $\varepsilon(q) = \partial$ for every $q$ in $\overline{q}$, recalling that $+\gamma^\lambda$ and $-\delta^\mu$ agree with  $\epsilon^\partial$ for each $\gamma^\lambda$ in $\overline{\gamma^\lambda}$ and $\delta^\mu$ in $\overline{\delta^\mu}$, and moreover every $\gamma^\lambda$  in $\overline{\gamma^\lambda}$ (resp.~$\delta^\mu$  in $\overline{\delta^\mu}$) is positive (resp.~negative) PIA, the following semantic equivalences hold for each $\gamma^\lambda$ in $\overline{\gamma^\lambda}$ and $\delta^\mu$ in $\overline{\delta^\mu}$:
\begin{equation}
\label{eq: after Ackermann3}
\gamma^\lambda\left[\overline{\bigvee_{k = 1}^{n_i} \bx_{\bk}/\bp_{\bi}}, \overline{\bigwedge_{h = 1}^{m_j} \ry_{\rh}/\rrq_{\rj}}\right ] = \bigwedge_{h = 1}^{m_j}\bigwedge_{k = 1}^{n_i}\gamma^\lambda\left[\overline{\bx_{\bk}/\bp_{\bi}}, \overline{\ry_{\rh}/\rrq_{\rj}}\right ]\quad\quad
 \delta^\mu\left [\overline{\bigvee_{k = 1}^{n_i} \bx_{\bk}/\bp_{\bi}}, \overline{\bigwedge_{h = 1}^{m_j} \ry_{\rh}/\rrq_{\rj}}\right] = \bigvee_{h = 1}^{m_j}\bigvee_{k = 1}^{n_i}\delta^\mu\left [\overline{\bx_{\bk}/\bp_{\bi}}, \overline{\ry_{\rh}/\rrq_{\rj}}\right ].
 \end{equation}
Hence, for every $\gamma^\lambda$  in $\overline{\gamma^\lambda}$ and $\delta^\mu$  in $\overline{\delta^\mu}$, the corresponding sequents  in \eqref{eq: after Ackermann2} can be equivalently replaced by (at most) $\Sigma_{n, m}(n_i m_j)$ sequents of the form

\begin{equation}
\label{eq:after splitting}
z\vdash \gamma^\lambda\left[\overline{\bx_{\bk}/\bp_{\bi}}, \overline{\ry_{\rh}/\rrq_{\rj}}\right ]\quad \quad \delta^\mu\left [\overline{\bx_{\bk}/\bp_{\bi}}, \overline{\ry_{\rh}/\rrq_{\rj}}\right]\vdash w,
\end{equation}
 yielding again a rule of the desired shape.
\end{proof}

As discussed, the Lemma above applies for both the non-distributive and distributive setting, following Remark \ref{remark:distr}.

\begin{example}
\label{example:computing rules}
%\marginnote{for every sequent in the example \ref{example: quasi special inductive}, illustrate the procedure described in the lemma above}
Let us illustrate the procedure described in the lemma above by applying it to the sequents discussed in Example \ref{example: quasi special inductive}.
{{ 
\begin{center}
\begin{tabular}{clr}
& ALBA-run computing the structural rule for $\wdia p \vdash \wbox \wdia p$: \\
\hline
    \ &$\wdia p\vdash \wbox \wdia p$ & \ \\
iff \ & $\forall p \forall x \forall w[x\vdash p\ \&\ \wdia p \vdash w \Rightarrow \wdia x \vdash \wbox w]$ & \ Instance of (\ref{eq: first approx}) \\
iff \ & $\forall  x \forall w [ \wdia x \vdash w \Rightarrow \wdia x \vdash \wbox w]$ & \ Instance of (\ref{eq: application Ackermann2}) \\
\end{tabular}
\end{center}
}}
Hence, the analytic rule corresponding to $\wdia p \vdash \wbox \wdia p$ is 
\[
\AX$\WDIA X \fCenter W$
\RL{\fns{$R_1$}}
\UI$\WDIA X \fCenter \WBOX W$
\DP
\]

{{ 
\begin{center}
\begin{tabular}{clr}
& ALBA-run computing the structural rule for $p \land \wbox q \vdash q \lor \wdia p$: \\
\hline
    \ &$p \land \wbox q \vdash q \lor \wdia p $ & \ \\
iff \ & $\forall p\forall q \forall x \forall y \forall z \forall w[(x\vdash p\ \&\ q \vdash y \ \&\ z \vdash \wbox q \ \&\ \wdia p \vdash w) \Rightarrow x \land z \vdash y \lor w]$ & \ Instance of (\ref{eq: first approx}) \\
iff \ & $\forall x \forall y \forall z \forall w[ (z \vdash \wbox y \ \&\ \wdia x \vdash w) \Rightarrow x \land z \vdash y \lor w]$& \ Instance of (\ref{eq: application Ackermann2}) \\
\end{tabular}
\end{center}
}}
Hence, the analytic rule corresponding to $p \land \wbox q \vdash q \lor \wdia p$ is 
\[
\AX$Z \fCenter \WBOX Y$
\AX$\WDIA X \fCenter W$
\RL{\fns$R_2$}
\BI$X \AAND Z \fCenter Y \AOR W$
\DP
\]

{{ 
\begin{center}
\begin{tabular}{clr}
& ALBA-run computing the structural rule for $p\mand (p\mand \wbox q) \vdash q \mor (q \mor \wdia p)$: & \\
\hline
    \ &$p\mand (p\mand \wbox q) \vdash q \mor (q \mor \wdia p)$& \ \\
iff \ & $\forall p\forall q \forall x_1 \forall x_2 \forall y_1 \forall y_2 \forall z \forall w$
\\
& $[(x_1\vdash p\ \&\ x_2\vdash p\ \&\ q \vdash y_1\ \&\ q \vdash y_2 \ \&\ z \vdash \wbox q \ \&\ \wdia p \vdash w) \Rightarrow x_1\mand (x_2 \mand z )\vdash y_1 \mor (y_2\mor w)]$ &  \ (\ref{eq: first approx})\\
iff \ & $\forall p\forall q \forall x_1 \forall x_2 \forall y_1 \forall y_2 \forall z \forall w$
\\
& $[(x_1\vee x_2\vdash p\ \&\ q \vdash y_1\wedge y_2 \ \&\ z \vdash \wbox q \ \&\ \wdia p \vdash w) \Rightarrow x_1\mand (x_2 \mand z )\vdash y_1 \mor (y_2\mor w)]$ &  \ (\ref{eq: after Ackermann2})\\
iff \ & $\forall p\forall q \forall x_1 \forall x_2 \forall y_1 \forall y_2 \forall z \forall w$
\\
& $[(  z \vdash \wbox (y_1\wedge y_2 ) \ \&\ \wdia (x_1\vee x_2) \vdash w) \Rightarrow x_1\mand (x_2 \mand z )\vdash y_1 \mor (y_2\mor w)]$    &(\ref{eq: after Ackermann3})\\
iff \ & $ \forall x_1 \forall x_2 \forall y_1 \forall y_2 \forall z \forall w$
\\
& $[(  z \vdash \wbox y_1\ \&\  z \vdash \wbox y_2  \ \&\ \wdia x_1 \vdash w \ \& \ \wdia x_2 \vdash w) \Rightarrow x_1\mand (x_2 \mand z )\vdash y_1 \mor (y_2\mor w)]$     &(\ref{eq:after splitting})\\
\end{tabular}
\end{center}
}}
Hence, the analytic rule corresponding to $p\mand (p\mand \wbox q) \vdash q \mor (q \mor \wdia p)$ is 
\[
\AXC{$Z\vdash \WBOX Y_1 \ \quad \ 
Z\vdash \WBOX Y_2 \ \quad \ 
\WDIA X_1 \vdash W \ \quad \ 
\WDIA X_2 \vdash W$}
\RL{\fns{$R_3$}}
\UIC{$X_1\MAND (X_2 \MAND Z) \fCenter Y_1 \MOR (Y_2\mor W)$}
\DP
\]
For the last inequality $\Diamond (p\lor \Diamond p) \le \Box \Diamond p $, we first need to preprocess the inequality and obtain two definite inequalities $\Diamond p\le \Box \Diamond p$  and $\Diamond\Diamond p\le \Box \Diamond p$. We now need to compute the rule for the second one (the first was already computed above):
{{ 
		\begin{center}
			\begin{tabular}{clr}
				& ALBA-run computing the structural rule for $\wdia\wdia p \vdash \wbox \wdia p$: \\
				\hline
				\ &$\wdia\wdia p\vdash \wbox \wdia p$ & \ \\
				iff \ & $\forall p \forall x \forall w[x\vdash p\ \&\ \wdia p \vdash w \Rightarrow \wdia\wdia x \vdash \wbox w]$ & \ Instance of (\ref{eq: first approx}) \\
				iff \ & $\forall  x \forall w [ \wdia x \vdash w \Rightarrow \wdia\wdia x \vdash \wbox w]$ & \ Instance of (\ref{eq: application Ackermann2}) \\
			\end{tabular}
		\end{center}
}}
Hence, the analytic rule corresponding to $\wdia\wdia p \vdash \wbox \wdia p$ is 
\[
\AX$\WDIA X \fCenter W$
\RL{\fns{$R_4$}}
\UI$\WDIA\WDIA X \fCenter \WBOX W$
\DP
\]
\end{example}
\begin{thm} If $(\varphi\fCenter\psi)[\obp, \orq,\overline{\bgamma}, \overline{\rdelta}]$ is a quasi-special inductive inequality,  then  a cut-free derivation in pre-normal form exists of $(\varphi\fCenter\psi)[\obp, \orq,\overline{\bgamma}, \overline{\rdelta}]$  in $\mathrm{\cfDLE} + \mathcal{R}$ (resp.~$\mathrm{\cfDDLE} + \mathcal{R}$), where $\mathcal{R}$ denotes the finite set of   analytic structural rules  corresponding to $(\varphi\fCenter\psi)[\obp, \orq,\overline{\bgamma}, \overline{\rdelta}]$ as in Lemma \ref{lemma: quasi special inductive corresponding rule}.
\end{thm}
\begin{proof} %\marginnote{AP: instead of talking about $f$ and $g$-connectives, I talk about SLR and SRR. Do you agree?}
 Recall that each $\gamma$ in $\overline{\gamma}$ (resp.~$\delta$ in $\overline{\delta}$) is a positive (resp.~negative) PIA formula. Hence, let $\bigwedge_\lambda\gamma^\lambda$ (resp.~$\bigvee_\mu\delta^\mu$) denote the equivalent rewriting of $\gamma$ (resp.~$\delta$) as conjunction (resp.~disjunction) of definite positive (resp.~negative) PIA formulas, as per Lemma \ref{lemma: reduction to definite}. Let us  assume that $(\varphi\fCenter\psi)[\obp, \orq,\overline{\bgamma}, \overline{\rdelta}]$ is definite,\footnote{\label{footnote:non definite}If $(\varphi\fCenter\psi)[\obp, \orq,\overline{\bgamma}, \overline{\rdelta}]$ is not definite, then, by Lemma \ref{lemma: from inductive to definite sequents}, it can  equivalently be transformed  into the conjunction of {\em definite} quasi-special inductive sequents $(\varphi_i\fCenter\psi_j)[\obp, \orq,\overline{\bgamma}, \overline{\rdelta}]$, where $\varphi$ is equivalent to $\bigvee_i\varphi_i$ and $\psi$ is equivalent to $\bigwedge_j\psi_j$,  which we can treat separately as shown in the proof. Then, a derivation of the original sequent can be obtained by applying the procedure indicated in the proof of Proposition \ref{prop: the thing needed for the non-definite} twice: by applying the procedure once, from derivations of $(\Phi_i\fCenter \Psi_j)[\obp, \orq, \overline{\bgamma}, \overline{\rdelta}]$ for every $i$ and $j$ we obtain derivations of $(\varphi\vdash\Psi_j)[\obp, \orq, \overline{\bgamma}, \overline{\rdelta}]$ for every $j$. Then, by applying the procedure again on these sequents, we obtain the required derivation of $(\varphi\fCenter \psi)[\obp, \orq, \overline{\bgamma}, \overline{\rdelta}]$.} and hence $ \mathcal{R}$ has only one element $\mathrm{R}$, which has the following shape (cf.~Lemma \ref{lemma: quasi special inductive corresponding rule}): 
\begin{center}
\AXC{$\overline{(Z\vdash \Gamma^\lambda)_\lambda}[\overline{X}/\overline{p}, \overline{Y}/\overline{q}]$}
\AXC{$\overline{(\Delta^\mu\vdash W)_{\mu}}[\overline{X}/\overline{p}, \overline{Y}/\overline{q}]$}
\BIC{$(\Phi\vdash \Psi)[\overline{X}, \overline{Y}, \overline{Z}, \overline{W}]$}
\DP
\end{center}
Then, 
modulo application of display rules, we can apply left-introduction (resp.~right-introduction) rules to positive (resp.~negative) SLR-connectives bottom-up, so as to transform all Skeleton connectives into structural connectives:
\begin{equation}
\label{apply invertible}
\AX$(\Phi\fCenter \Psi)[\obp, \orq, \overline{\bgamma}, \overline{\rdelta}]$
\noLine
\UIC{$\vdots$}
\noLine
\UI$(\varphi\fCenter \psi)[\obp, \orq, \overline{\bgamma}, \overline{\rdelta}]$
\DP
\end{equation} 
Notice that $(\Phi\fCenter \Psi)[\obp, \orq, \overline{\bgamma}, \overline{\rdelta}]$ is an instance of the conclusion of $\mathrm{R}$ with $\overline{p}/\overline{X}$, $\overline{q}/\overline{Y}$, $\overline{\gamma}/\overline{Z}$ and $\overline{\delta}/\overline{W}$. Hence, we can apply $\mathrm{R}$ bottom-up and obtain:
\begin{equation}
\label{structural rule}
\AXC{$\overline{(\gamma\vdash \Gamma^\lambda)_\lambda}[\obp/\obp, \orq/\orq]$}
\AXC{$\overline{(\Delta^\mu \vdash \delta)_\mu}[\obp/\obp, \orq/\orq]$}
\BIC{$(\Phi\vdash \Psi)[\obp,\orq, \overline{\bgamma}, \overline{\rdelta}]$}
\DP
\end{equation} 
 By Corollary \ref{prop: derivation of identities with weakening},  the sequents $(\gamma\vdash \Gamma^\lambda)[\obp/\obp, \orq/\orq]$  and $(\Delta^\mu\vdash \delta)[\obp/\obp, \orq/\orq]$ are cut-free derivable in $\mathrm{\cfDLE}$ (resp.~$\mathrm{\cfDDLE}$), with derivations which only contain identity axioms, and applications of right-introduction rules for  negative PIA-connectives, % SRR-connectives and negative unary SRA-connectives,  
and left-introduction rules for positive PIA-connectives %SRR-connectives and positive unary SRA-connectives. 
%with derivations only consisting of identity axioms and applications of left-introduction of $\wedge$, right-introduction of $\vee$,  right-introduction $f$-connectives and left-introduction of $g$-connectives 
(and weakening and exchange rules in the case of $\mathrm{\cfDDLE}$). Moreover, as discussed above (cf.~also Proposition \ref{prop: the thing needed for the non-definite}), the rules applied after applying the rules in $\mathcal{R}$ are only display rules, left-introduction rules for positive Skeleton-connectives and right-introduction rules for negative Skeleton-connectives (plus contraction in the case of $\mathrm{\cfDDLE}$). This completes the proof that the cut-free derivation in $\mathrm{\cfDLE} + \mathcal{R}$ (resp.~$\mathrm{\cfDDLE} + \mathcal{R}$) of $(\varphi\fCenter\psi)[\obp, \orq,\overline{\bgamma}, \overline{\rdelta}]$ is in pre-normal form. 
\end{proof}

\begin{example}
\label{example: deriving qsi axioms}
Let us illustrate the procedure described in the proposition above by deriving the sequents in Example \ref{example: quasi special inductive}.

\begin{center}
\begin{tabular}{ccccc}
\cfDLE-derivation of  $\wdia p \vdash \wbox \wdia p$:
 & & & &
\cfDDLE-derivation of $p \aand \wbox q \vdash q \aor \wdia p$: 
\rule[-1.8mm]{0mm}{0mm} \\
\cline{0-1}
\cline{4-5}

\begin{tikzpicture}		
\node at(0,0) {
\AX$p\fCenter p$
\RL{\fns{$\wdia_R$}}
\UI$\WDIA p \fCenter \wdia p$
\RL{\footnotesize$R_1$}
\UI$\WDIA p \fCenter\WBOX \wdia p$
\LL{\fns{$\wdia_L$}}
\UI$\wdia p \fCenter\WBOX \wdia p$
\RL{\fns{$\wbox_R$}}
\UI$\wdia p \fCenter \wbox \wdia p$
\DP
};
 %\draw[help lines] (-4,-4) grid (4,4);

\node[rotate = -90] at (-1.3, -0.5) {$\underbrace{\hspace{1.3cm}}$};
\node at (-1.8,-0.5) {(\ref{apply invertible})};
\node[rotate = 90] at (1.5, 0.26) {$\underbrace{\hspace{0.6cm}}$};
\node at (2,0.26) {(\ref{structural rule})};

\end{tikzpicture}

&&\ \ \ \ \ \ \ &&

\begin{tikzpicture}		
\node at(0,0) {

\AX$q\fCenter q$
\LL{\fns{$\wbox_L$}}
\UI$\wbox q \fCenter\WBOX q$
\AX$p\fCenter p$
\RL{\fns{$\wdia_R$}}
\UI$\WDIA p \fCenter \wdia p$
\RL{\footnotesize$R_2$}
\BI$p \AAND \wbox q \fCenter q \AOR \wdia p$
\LL{\fns{$\aand_L$}}
\UI$p \aand \wbox q \fCenter q \AOR \wdia p$
\RL{\fns{$\aor_R$}}
\UI$p \aand \wbox q \fCenter q \aor \wdia p$
\DP
  %\draw[help lines] (-4,-4) grid (4,4);
};
\node[rotate = -90] at (-1.75, -0.5) {$\underbrace{\hspace{1.3cm}}$};
   \node at (-2.25,-0.5) {(\ref{apply invertible})};
\node[rotate = 90] at (2.45, 0.26) {$\underbrace{\hspace{0.6cm}}$};
\node at (2.95,0.26) {(\ref{structural rule})};

\end{tikzpicture}

 \\
\end{tabular}
\end{center}

\begin{center}
\begin{tabular}{c}
\cfDLE-derivation of $p \mand (p \mand\wbox q )\fCenter q \mor( q \mor \wdia p)$: 
\rule[-1.8mm]{0mm}{0mm} \\
\hline
\begin{tikzpicture}		
\node at(0,0) {

\AX $q\fCenter q$
\LL{\fns{$\wbox_L$}}
\UI$\wbox q \fCenter \WBOX q$
\AX$q\fCenter q$
\LL{\fns{$\wbox_L$}}
\UI$\wbox q \fCenter\WBOX q$
\AX$p\fCenter p$
\RL{\fns{$\wdia_R$}}
\UI$\WDIA p \fCenter \wdia p$
\AX$p\fCenter p$
\RL{\fns{$\wdia_R$}}
\UI$\WDIA p \fCenter \wdia p$
\RL{\footnotesize$R_3$}
\QuaternaryInfC{$p \MAND (p \MAND \wbox q )\fCenter q \MOR ( q \MOR \wdia p)$}
\UI$p \MAND \wbox q  \fCenter p \MRARR (q \MOR( q \MOR \wdia p))$
\LL{\fns$\mand_L$}
\UI$p \mand \wbox q  \fCenter p \MRARR (q \MOR( q \MOR \wdia p))$
\UI$p \MAND(p \mand\wbox q )\fCenter q \MOR( q \MOR \wdia p)$
\LL{\fns$\mand_L$}
\UI$p \mand (p \mand\wbox q )\fCenter q \MOR( q \MOR \wdia p)$
\UI$q\MDLARR (p \mand (p \mand\wbox q ))  \fCenter  q \MOR \wdia p$
\RL{\fns$\mor_R$}
\UI$q\MDLARR (p \mand (p \mand\wbox q )) \fCenter  q \mor \wdia p$
\UI$p \mand (p \mand\wbox q )\fCenter q \MOR( q \mor \wdia p)$
\RL{\fns$\mor_R$}
\UI$p \mand (p \mand\wbox q )\fCenter q \mor( q \mor \wdia p)$
\DP
};
\node[rotate = -90] at (-2.8, -0.5) {$\underbrace{\hspace{4.5cm}}$};
\node at (-3.3,-0.5) {(\ref{apply invertible})};
\node[rotate = 90] at (5.3, 1.8) {$\underbrace{\hspace{0.6cm}}$};
\node at (5.8,1.8) {(\ref{structural rule})};
\end{tikzpicture}
\end{tabular}
\end{center}

\begin{center}
	\begin{tabular}{c}
		\cfDLE-derivation of $\wdia (p \lor \wdia p)\fCenter \wbox \wdia p$: 
		\rule[-1.8mm]{0mm}{0mm} \\
		\hline
		\begin{tikzpicture}		
			\node at(0,0) {
				
			\AX$p\fCenter p$
			\RL{\fns{$\wdia_R$}}
			\UI$\WDIA p \fCenter \wdia p$
			\RL{\footnotesize$R_1$}
			\UI$\WDIA p \fCenter\WBOX \wdia p$
			\LL{\fns $\WDIA \dashv \BBOX$}
			\UI$p \fCenter \BBOX\WBOX \wdia p$
			\AX$p\fCenter p$
			\RL{\fns{$\wdia_R$}}
			\UI$\WDIA p \fCenter \wdia p$
			\RL{\footnotesize$R_4$}
			\UI$\WDIA\WDIA p \fCenter\WBOX \wdia p$
			\LL{\fns $\WDIA \dashv \BBOX$}
			\UI$\WDIA p \fCenter \BBOX\WBOX \wdia p$
			\LL{\fns{$\wdia_L$}}
			\UI$\wdia p \fCenter\BBOX\WBOX \wdia p$
			\LL{\fns{$\lor_L$}}
			\BI$p \lor \wdia p \fCenter\BBOX\WBOX \wdia p$
			\LL{\fns $\WDIA \dashv \BBOX$}
			\UI$\WDIA(p \lor \wdia p)\fCenter \WBOX \wdia p$
			\LL{\fns{$\wdia_L$}}
			\UI$\wdia(p \lor \wdia p)\fCenter \WBOX \wdia p$
			\RL{\fns{$\wbox_R$}}
			\UI$\wdia(p \lor \wdia p)\fCenter \wbox \wdia p$
			\DP
			};
			\node[rotate = -90] at (-3.3, -0.8) {$\underbrace{\hspace{2.9cm}}$};
			\node at (-4.4,-0.8) {Footnote (\ref{footnote:non definite})};
			\node[rotate = 90] at (5, -0.5) {$\underbrace{\hspace{3.58cm}}$};
			\node at (6.1, -0.5) {Footnote (\ref{footnote:non definite})};
			\node[rotate = 90] at (3.9, 1.3) {$\underbrace{\hspace{0.6cm}}$};
			\node at (4.4,1.3) {(\ref{structural rule})};
			\node[rotate = -90] at (-2.2, 0.7) {$\underbrace{\hspace{0.6cm}}$};
			\node at (-2.7,0.7) {(\ref{structural rule})};
		\end{tikzpicture}
	\end{tabular}
\end{center}
\end{example}

\subsection{Syntactic completeness for analytic inductive sequents}
\label{ssec: syntactic completeness general}

%%%%%%%%%%%%%%%%%%%%%%%%%%%%%%%

%%%%%%%%%%%%%%%%%%%%%%%%%%%%%%%%%

%%%%%%%%%%%%%%%%%%%%%%%%%%%%%%%%%%%%

\begin{lemma}
\label{lemma:  inductive corresponding rule anind}
If $(\varphi\fCenter\psi)[\oba, \orb,\overline{\bgamma}, \overline{\rdelta}]$ is an analytic $(\Omega, \varepsilon)$-inductive sequent, then  each of its corresponding rules has the following shape:
\begin{equation}\label{eq:structural rule (definite) analytic-inductive}
\AXC{$\overline{(Z\vdash \Gamma^\lambda)_\lambda}[\overline{\mathsf{MV}(p)}/\overline{p}, \overline{\mathsf{MV}(q)}/\overline{q}]$}
\AXC{$\overline{(\Delta^\mu\vdash W)_\mu}[\overline{\mathsf{MV}(p)}/\overline{p}, \overline{\mathsf{MV}(q)}/\overline{q}]$}
\BIC{$(\Phi_j\vdash \Psi_i)[\overline{X}, \overline{Y}, \overline{Z}, \overline{W}]$}
\DP
\end{equation}
where $\Gamma^\lambda$, $\Delta^\mu$, $\Phi_j$ and $\Psi_i$ are as in Lemma \ref{lemma: quasi special inductive corresponding rule}, and  $\mathsf{MV}(p)$ and $\mathsf{MV}(q)$ denote the structural counterparts of the components of the minimal and maximal valuations $\mathsf{mv}(p)\in \mathsf{Mv}(p)$ and $\mathsf{mv}(q)\in \mathsf{Mv}(q)$ defined in the proof below. %\marginnote{G: shall we move the definitions of mv and MV before the Lemma?\\ AP: no, for this one time, let's not...}
\end{lemma}
\begin{proof}
Let us apply the algorithm ALBA to  $(\varphi\fCenter\psi)[\oba, \orb,\overline{\bgamma}, \overline{\rdelta}]$ to compute its corresponding analytic rules. Modulo pre-processing, we can assume w.l.o.g.~that $(\varphi\fCenter\psi)[\oba, \orb,\overline{\bgamma}, \overline{\rdelta}]$  is  definite,\footnote{If $(\varphi\fCenter\psi)[\oba, \orb,\overline{\bgamma}, \overline{\rdelta}]$ is not definite, then by Lemma \ref{lemma: from inductive to definite sequents} it can equivalently be transformed  into the conjunction of {\em definite}  analytic inductive sequents  which we can treat separately as shown in the proof.}  and hence we can proceed with first approximation,
which yields the following quasi-inequality: 
\begin{equation}
\label{eq: first approx anind}
\forall \ox\oy\overline{z}\overline{w}[(\overline{x\vdash \alpha}\ \&\    \overline{\beta\vdash y}\ \&\ \overline{z\vdash \gamma}\ \&\ \overline{\delta\vdash w})\Rightarrow (\varphi\vdash\psi)[\obx, \ory,\overline{\bz}, \overline{\rw}]],
\end{equation}
Modulo distribution and splitting (cf.~Lemma \ref{lemma: reduction to definite}), the quasi-inequality above can be equivalently rewritten as follows :
\begin{equation}
\label{eq: first approx anind definite}
\forall \ox\oy\overline{z}\overline{w}[(\overline{x\vdash \alpha_p}\ \&\  \overline{x\vdash \alpha_q}\ \&\  \overline{\beta_p\vdash y}\ \&\ \overline{\beta_q\vdash y}\ \&\ \overline{z\vdash \gamma}\ \&\ \overline{\delta\vdash w})\Rightarrow (\varphi\vdash\psi)[\obx, \ory,\overline{\bz}, \overline{\rw}]],
\end{equation}
where each $\alpha_p$ and $\alpha_q$ (resp.~$\beta_p$ and $\beta_q$)  is {\em definite} positive (resp.~negative) PIA and contains a unique $\varepsilon$-critical propositional variable occurrence, which we indicate in its subscript.
By applying adjunction and residuation ALBA-rules on all definite PIA-formulas $\alpha_p$, $\alpha_q$, $\beta_p$ and $\beta_q$ using each $\varepsilon$-critical propositional variable occurrence as the pivotal variable, the antecedent of the quasi-inequality above can be equivalently written as follows:  
\begin{equation}
\label{eq: adj anind}
\begin{split}
\overline{\mathsf{la}(\alpha_p)[\bx/\bu, \obp,\orq]\vdash p}\ &\&\ \overline{\mathsf{ra}(\beta_p)[\ry/\ru, \obp,\orq]\vdash p}\ \&\ \\ \overline{q\vdash \mathsf{la}(\alpha_q)[\bx/\bu, \obp,\orq]}\ &\&\ \overline{q\vdash \mathsf{ra}(\beta_q)[\ry/\ru, \obp,\orq]}\ \&\ \\\overline{z\vdash \gamma}\ &\&\ \overline{\delta\vdash w}
\end{split}
\end{equation}
Since each $\gamma$ (resp.~$\delta$) is a positive (resp.~negative) PIA formula, by Lemma \ref{lemma: reduction to definite}  it is equivalent to $\bigwedge_{\lambda}\gamma^\lambda$ (resp.~$\bigvee_\mu\delta^\mu$), where each $\gamma^\lambda$ (resp.~$\delta^\mu$) is \emph{definite} positive (resp.~negative) PIA. Therefore \eqref{eq: adj anind} can be equivalently rewritten as follows:
\begin{equation}
\label{eq: adj anind2}
\begin{split}
\overline{\mathsf{la}(\alpha_p)[\bx/\bu, \obp,\orq]\vdash p}\ &\&\ \overline{\mathsf{ra}(\beta_p)[\ry/\ru, \obp,\orq]\vdash p}\ \&\ \\ \overline{q\vdash \mathsf{la}(\alpha_q)[\bx/\bu, \obp,\orq]}\ &\&\ \overline{q\vdash \mathsf{ra}(\beta_q)[\ry/\ru, \obp,\orq]}\ \&\ \\ \overline{(z\vdash \gamma^\lambda)_{\lambda}}\ &\&\ \overline{(\delta^\mu\vdash w)_{\mu}}
\end{split}
\end{equation}
Notice that the `parametric' (i.e.~non-critical) variables in $\overline{p}$ and $\overline{q}$ actually occurring in each formula $\mathsf{la}(\alpha_p)[x/u, \overline{p}, \overline{q}]$, $\mathsf{ra}(\beta_p)[y/u, \overline{p}, \overline{q}]$, $\mathsf{la}(\alpha_q)[x/u, \overline{p}, \overline{q}]$, and $\mathsf{ra}(\beta_q)[y/u, \overline{p}, \overline{q}]$ are those that are strictly $<_\Omega$-smaller than the (critical and pivotal) variable indicated in the subscript of the given PIA-formula. After applying adjunction and residuation as indicated above, the quasi-inequality \eqref{eq: first approx anind definite} is in Ackermann shape relative to the $<_\Omega$-minimal variables.

For every $p\in\overline{p}$ and $q\in\overline{q}$, let us define the sets $\mathsf{Mv}(p)$ and $\mathsf{Mv}(q)$ by recursion on $<_\Omega$ as follows: 
\begin{itemize}
	\item
	$\mathsf{Mv}(p):=\{\mathsf{la}(\alpha_p)[x_k/u,\overline{\mathsf{mv}(p)}/\overline{p},\overline{\mathsf{mv}(q)}/\overline{q}], \mathsf{ra}(\beta_p)[y_h/u,\overline{\mathsf{mv}(p)}/\overline{p},\overline{\mathsf{mv}(q)}/\overline{q}]\mid 1\leq k\leq n_{i_1}, 1\leq h\leq n_{i_2}, \overline{\mathsf{mv}(p)}\in\overline{\mathsf{Mv}(p)},\overline{\mathsf{mv}(q)}\in\overline{\mathsf{Mv}(q)}  \}$
	\item $\mathsf{Mv}(q):=\{\mathsf{la}(\alpha_q)[x_h/u,\overline{\mathsf{mv}(p)}/\overline{p},\overline{\mathsf{mv}(q)}/\overline{q}], \mathsf{ra}(\beta_q)[y_k/u,\overline{\mathsf{mv}(p)}/\overline{p},\overline{\mathsf{mv}(q)}/\overline{q})\mid 1\leq h\leq m_{j_1}, 1\leq k\leq m_{j_2}, \overline{\mathsf{mv}(p)}\in\overline{\mathsf{Mv}(p)},\overline{\mathsf{mv}(q)}\in\overline{\mathsf{Mv}(q)}  \}$
\end{itemize}
where,  $n_{i_1}$ (resp.~$n_{i_2}$) is the number of occurrences of $p$ in $\alpha$s (resp.~in $\beta$s) for every $p\in\overline{p}$, and $m_{j_1}$ (resp.~$m_{j_2}$) is the number of occurrences of $q$ in $\alpha$s (resp.~in $\beta$s) for every $q\in\overline{q}$. 
By induction on $<_\Omega$, we can apply the Ackermann rule exhaustively so as to eliminate all variables $p$  and  $q$.  Then the antecedent of the quasi-inequality has the following form: % (recall that by assumption all formulas in  $\overline{\gamma}$ and $\overline{\delta}$ agree with  $\epsilon^\partial$):
 
\begin{equation}
\label{eq: after Ackermann anind}
\overline{(z\vdash \gamma^\lambda)_{\lambda}}\left[\overline{\bigvee\mathsf{Mv}(p)}/\overline{p}, \overline{\bigwedge\mathsf{Mv}(q)}/\overline{q}\right]\quad \quad \overline{(\delta^\mu\vdash w)_\mu}\left [\overline{\bigvee\mathsf{Mv}(p)}/\overline{p}, \overline{\bigwedge\mathsf{Mv}(q)}/\overline{q}\right].\end{equation}
Since by assumption $\varepsilon(p) = 1$ for every $p$ in $\overline{p}$ and $\varepsilon(q) = \partial$ for every $q$ in $\overline{q}$, recalling that  $\overline{\gamma^\lambda}$ and $\overline{\delta^\mu}$ agree with  $\epsilon^\partial$, and moreover every $\gamma^\lambda$  in $\overline{\gamma^\lambda}$ (resp.~$\delta^\mu$  in $\overline{\delta^\mu}$) is positive (resp.~negative) PIA, the following semantic equivalences hold for each $\gamma^\lambda$ in $\overline{\gamma^\lambda}$ and $\delta^\mu$ in $\overline{\delta^\mu}$:
\[\gamma^\lambda\left[\overline{\bigvee\mathsf{Mv}(p)}/\overline{p}, \overline{\bigwedge\mathsf{Mv}(q)}/\overline{q}\right] = \bigwedge\gamma^\lambda\left[\overline{\mathsf{mv}(p)}/\overline{p}, \overline{\mathsf{mv}(q)}/\overline{q}\right]\]
\[ \delta^\mu\left[\overline{\bigvee\mathsf{Mv}(p)}/\overline{p}, \overline{\bigwedge\mathsf{Mv}(q)}/\overline{q}\right] = \bigvee\delta^\mu\left[\overline{\mathsf{mv}(p)}/\overline{p}, \overline{\mathsf{mv}(q)}/\overline{q}\right].\]
Hence for every $\gamma^\lambda$  in $\overline{\gamma^\lambda}$ and $\delta^\mu$  in $\overline{\delta^\mu}$, the corresponding sequents  in \eqref{eq: after Ackermann anind} can be equivalently replaced by (at most) $\Sigma_{n, m}(n_i m_j)$ sequents of the form

\begin{equation}\label{last eq}
z\vdash \gamma^\lambda\left[\overline{\mathsf{mv}(p)}/\overline{p}, \overline{\mathsf{mv}(q)}/\overline{q}\right]\quad \quad \delta^\mu\left[\overline{\mathsf{mv}(p)}/\overline{p}, \overline{\mathsf{mv}(q)}/\overline{q}\right]\vdash w,
\end{equation}
 yielding a rule of the desired shape.
\end{proof}

\begin{example}
\label{example:computing rules anint}
Let us illustrate the procedure described in the lemma above by applying it to the sequents discussed in Example \ref{example:inductive and analytic inductive}.
{{ 
\begin{center}
\begin{tabular}{clr}
& ALBA-run computing the structural rule for $\wdia \wbox p \vdash \wbox \wdia p$: \\
\hline
    \ &$\wdia \wbox p\vdash \wbox \wdia p$ & \ \\
iff \ & $\forall p \forall x \forall w[x\vdash \wbox p\ \&\ \wdia p \vdash w \Rightarrow \wdia x \vdash \wbox w]$ & \ Instance of (\ref{eq: first approx anind}) \\
iff \ & $\forall p \forall x \forall w[\bdia x\vdash p\ \&\ \wdia p \vdash w \Rightarrow \wdia x \vdash \wbox w]$   & \ Instance of (\ref{eq: adj anind})  \\
iff \ & $ \forall x \forall w[\wdia \bdia x \vdash w \Rightarrow \wdia x \vdash \wbox w]$ &  Instance of (\ref{last eq})\\
\end{tabular}
\end{center}
}}
Then the analytic rule corresponding to  $\wdia \wbox p \vdash \wbox \wdia p$ is:
\[
\AX$\WDIA \BDIA X \fCenter W$
\RL{\fns$R_4$}
\UI$\WDIA X \fCenter \WBOX W$
\DP
\]
{{ 
\begin{center}
\begin{tabular}{clr}
& ALBA-run computing the structural rule for $p \ararr (q \ararr r) \le ((p\ararr q) \ararr (\wbox p \ararr r))\mor \Diamond r$: \\
\hline
    \ &$p \ararr (q \ararr r) \le ((p\ararr q) \ararr (\wbox p \ararr r))\mor \Diamond r$ & \ \\
iff \ & $\forall p \forall q \forall r \forall x_1 \forall x_2 \forall y_1 \forall y_2 \forall z$
\\
&$[x_1\vdash \wbox p \ \& \ x_2 \vdash p \ararr q \ \&\ \wdia r \vdash y_1\ \&\ r \vdash y_2 \ \&\ z \vdash p \ararr (q \ararr r) \Rightarrow z \vdash  (x_2\ararr(x_1 \ararr y_2))\mor y_1  ]$ &(\ref{eq: first approx anind}) \\
iff \ & $\forall p \forall q \forall r \forall x_1 \forall x_2 \forall y_1 \forall y_2 \forall z$
\\
&$[\bdia x_1\vdash  p \ \&p\ \aand  x_2 \vdash  q \ \&\  r \vdash \bbox y_1\ \&\ r \vdash y_2 \ \&\ z \vdash p \ararr (q \ararr r) \Rightarrow z \vdash( x_2\ararr(x_1 \ararr y_2))\mor y_1 ]$ &  (\ref{eq: adj anind}) \\
iff \ & $ \forall q \forall r \forall x_1 \forall x_2 \forall y_1 \forall y_2 \forall z$
\\
&$[\bdia x_1\aand  x_2 \vdash  q \ \&\  r \vdash \bbox y_1\ \&\ r \vdash y_2 \ \&\ z \vdash \bdia x_1 \ararr (q \ararr r) \Rightarrow z \vdash( x_2\ararr(x_1 \ararr y_2))\mor y_1 ]$ &   \\
iff \ & $\forall r \forall x_1 \forall x_2 \forall y_1 \forall y_2 \forall z$
\\
&$[ r \vdash \bbox y_1\ \&\ r \vdash y_2 \ \&\ z \vdash \bdia x_1 \ararr (\bdia x_1\aand  x_2  \ararr r) \Rightarrow z \vdash  x_2\ararr(x_1 \ararr y_2))\mor y_1 ]$ &  \\
iff \ & $\forall x_1 \forall x_2 \forall y_1 \forall y_2 \forall z$
\\
&$[  z \vdash \bdia x_1 \ararr (\bdia x_1\aand  x_2  \ararr \bbox y_1\aor y_2) \Rightarrow  z \vdash (x_2\ararr(x_1 \ararr y_2))\mor y_1]$ &  (\ref{eq: after Ackermann anind}) \\
iff \ & $\forall x_1 \forall x_2 \forall y_1 \forall y_2 \forall z$
\\
&$[  z \vdash \bdia x_1 \ararr (\bdia x_1\aand  x_2  \ararr \bbox y_1)\ \& \ z \vdash \bdia x_1 \ararr (\bdia x_1\aand  x_2  \ararr y_2)\Rightarrow  z \vdash (x_2\ararr(x_1 \ararr y_2))\mor y_1 ]$ &   (\ref{last eq})\\
\end{tabular}
\end{center}
}}
Then the analytic rule corresponding to  $p \ararr (q \ararr r) \le ((p\ararr q) \ararr (\wbox p \ararr r))\mor \Diamond r$ is:
\[
\AX$Z \fCenter \BDIA X_1 \ARARR (\BDIA X_1\AAND  X_2  \ARARR \BBOX Y_1)$
\AX$Z \fCenter \BDIA X_1 \ARARR (\BDIA X_1\AAND  X_2  \ARARR Y_2)$
\RL{\fns$R_5$}
\BI$Z \fCenter( X_2\ARARR(X_1 \ARARR Y_2))\MOR Y_1$
\DP
\]
{{ 
\begin{center}
\begin{tabular}{clr}
& ALBA-run computing the structural rule for $\wdia \wbox ( p \aand q) \vdash \wbox \wdia p \aor  \wbox \wdia q$: \\
\hline
    \ &$\wdia \wbox ( p \aand q) \vdash \wbox \wdia p \aor  \wbox \wdia q$ & \ \\
iff \ & $\forall p \forall x \forall w[ x \vdash \wbox ( p \aand q) \ \& \  \wdia p \vdash y \ \& \ \wdia q \vdash z
\Rightarrow \wdia x  \vdash \wbox y \aor  \wbox  z]$ & \ Instance of (\ref{eq: first approx anind}) \\

iff \ & $\forall p \forall x \forall w[ x \vdash \wbox ( p \aand q) \ \& \   p \vdash \bbox y \ \& \  q \vdash \bbox z \Rightarrow \wdia x  \vdash \wbox y \aor  \wbox  z]$
 & \ Instance of (\ref{eq: adj anind})  \\

iff \ & $\forall p \forall x \forall w[ x \vdash \wbox (  \bbox y  \aand  \bbox z) \Rightarrow \wdia x  \vdash \wbox y \aor  \wbox  z]$
 & \ Instance of (\ref{eq: after Ackermann anind})  \\

iff \ & $\forall p \forall x \forall w[ x \vdash \wbox \bbox y  \ \&\ x \vdash \wbox \bbox z \Rightarrow \wdia x  \vdash \wbox y \aor  \wbox  z]$ &  Instance of (\ref{last eq})\\
\end{tabular}
\end{center}
}}
Then the analytic rule corresponding to $\wdia \wbox ( p \aand q) \vdash \wbox \wdia p \aor  \wbox \wdia q$ is:
\[
\AX$X \fCenter \WBOX \BBOX Y$
\AX$X \fCenter \WBOX \BBOX Z$
\RL{\fns$R_6$}
\BI$\WDIA X \fCenter \WBOX Y \AOR \WBOX Z$
\DP
\]
\end{example}
\begin{thm} 
\label{prop:synt-compl}
If $(\varphi\fCenter\psi)[\oba, \orb,\overline{\bgamma}, \overline{\rdelta}]$ is an analytic inductive sequent, then  a cut-free derivation in pre-normal form exists of $(\varphi\fCenter\psi)[\oba, \orb,\overline{\bgamma}, \overline{\rdelta}]$ in $\mathrm{\cfDLE} + \mathcal{R}$ (resp.~$\mathrm{\cfDDLE} + \mathcal{R}$), where $\mathcal{R}$ denotes the finite set of   analytic structural rules  corresponding to $(\varphi\fCenter\psi)[\oba, \orb,\overline{\bgamma}, \overline{\rdelta}]$  as in Lemma \ref{lemma:  inductive corresponding rule anind}.
\end{thm}
\begin{proof}
 Recall that each $\gamma$ in $\overline{\gamma}$ (resp.~$\delta$ in $\overline{\delta}$) is a positive (resp.~negative) PIA formula. Let $\bigwedge_\lambda\gamma^\lambda$ (resp.~$\bigvee_\mu\delta^\mu$) denote the equivalent rewriting of $\gamma$ (resp.~$\delta$) as conjunction (resp.~disjunction) of definite positive (resp.~negative) PIA formulas as per Lemma \ref{lemma: reduction to definite}. Let us assume that $(\varphi\fCenter\psi)[\oba, \orb,\overline{\bgamma}, \overline{\rdelta}]$ is definite,%\footnote{If $(\varphi\fCenter\psi)[\obp, \orq,\overline{\bgamma}, \overline{\rdelta}]$ was not definite, then, by e.g.~applying the ALBA-preprocessing to it, we can equivalently transform $(\varphi\fCenter\psi)[\obp, \orq,\overline{\bgamma}, \overline{\rdelta}]$ into a conjunction of definite analytic inductive inequalities, which we can treat separately as shown in the proof. Then, a derivation of the original inequality can be obtained by suitably prolonging these derivations with applications of left-introduction of $\vee$ and right-introduction of $\wedge$ (plus right-and left-contraction rules in the distributive setting).}
\footnote{\label{footnote:non definite general}If $(\varphi\fCenter\psi)[\oba, \orb,\overline{\bgamma}, \overline{\rdelta}]$ is not definite, then, by Lemma \ref{lemma: from inductive to definite sequents}, it can  equivalently be transformed  into the conjunction of {\em definite} analytic inductive sequents $(\varphi_i\fCenter\psi_j)[\oba, \orb,\overline{\bgamma}, \overline{\rdelta}]$, where $\varphi$ is equivalent to $\bigvee_i\varphi_i$ and $\psi$ is equivalent to $\bigwedge_j\psi_j$,  which we can treat separately as shown in the proof. Then, a derivation of the original sequent can be obtained by applying the procedure indicated in the proof of Proposition \ref{prop: the thing needed for the non-definite} twice: by applying the procedure once, from derivations of $(\Phi_i\fCenter \Psi_j)[\oba, \orb, \overline{\bgamma}, \overline{\rdelta}]$ for every $i$ and $j$ we obtain derivations of $(\varphi\vdash\Psi_j)[\oba, \orb, \overline{\bgamma}, \overline{\rdelta}]$ for every $j$. Then, by applying the procedure again on these sequents, we obtain the required derivation of $(\varphi\fCenter \psi)[\oba, \orb, \overline{\bgamma}, \overline{\rdelta}]$.}
  and hence $ \mathcal{R}$ has only one element $\mathrm{R}$, with the following shape (cf.~Lemma \ref{lemma:  inductive corresponding rule anind}):
\begin{equation} %\label{eq:structural rule (definite) analytic-inductive}
\AXC{$\overline{(Z\vdash \Gamma^\lambda)_\lambda}[\overline{\mathsf{MV}(p)}/\overline{p}, \overline{\mathsf{MV}(q)}/\overline{q}]$}
\AXC{$\overline{(\Delta^\mu\vdash W)_\mu}[\overline{\mathsf{MV}(p)}/\overline{p}, \overline{\mathsf{MV}(q)}/\overline{q}]$}
\BIC{$(\Phi\vdash \Psi)[\overline{X}, \overline{Y}, \overline{Z}, \overline{W}]$}
\DP
\end{equation}
 Then, 
modulo application of display rules, we can apply left-introduction (resp.~right-introduction) rules to positive (resp.~negative) SLR-connectives bottom-up, so as to transform all Skeleton connectives into structural connectives:
\begin{equation}\label{inv anlytic}
\AX$(\Phi\fCenter \Psi)[\oba, \orb, \overline{\bgamma}, \overline{\rdelta}]$
\noLine
\UIC{$\vdots$}
\noLine
\UI$(\varphi\fCenter \psi)[\oba, \orb, \overline{\bgamma}, \overline{\rdelta}]$
\DP
\end{equation} 
Notice that $(\Phi\fCenter \Psi)[\oba, \orb, \overline{\bgamma}, \overline{\rdelta}]$ is an instance of the conclusion of $\mathrm{R}$. Hence we can apply $\mathrm{R}$ bottom-up and obtain:
\begin{center}
\AXC{$\overline{(\gamma\vdash \Gamma^\lambda)_\lambda}[\overline{\mathsf{MV}(p)}[\oba/\obx, \orb/\ory]/\obp, \overline{\mathsf{MV}(q)}[\oba/\obx, \orb/\ory]/\orq]$}
\AXC{$\overline{(\Delta^\mu\vdash \delta)_\mu}[\overline{\mathsf{MV}(p)}[\oba/\obx, \orb/\ory]/\obp, \overline{\mathsf{MV}(q)}[\oba/\obx, \orb/\ory]/\orq]$}
\BIC{$(\Phi\vdash \Psi)[\oba,\orb, \overline{\gamma}, \overline{\delta}]$}
\DP
\end{center} 
 To finish the proof that $(\varphi\fCenter \psi)[\oba, \orb, \overline{\bgamma}, \overline{\rdelta}]$
is derivable in in $\mathrm{\cfDLE} + \mathrm{R}$ (resp.~$\mathrm{\cfDDLE} + \mathrm{R}$), it is enough to show that  the sequents \[\overline{(\gamma\vdash \Gamma^\lambda)_\lambda}[\overline{\mathsf{MV}(p)}[\oba/\obx, \orb/\ory]/\obp, \overline{\mathsf{MV}(q)}[\oba/\obx, \orb/\ory]/\orq]\quad\text{  and }\quad\overline{(\Delta^\mu\vdash \delta)_\mu}[\overline{\mathsf{MV}(p)}[\oba/\obx, \orb/\ory]/\obp, \overline{\mathsf{MV}(q)}[\oba/\obx, \orb/\ory]/\orq]\] are derivable in $\mathrm{\cfDLE}$ (resp.~$\mathrm{\cfDDLE}$). Recalling that  each $\gamma$ in $\overline{\gamma}$ (resp.~$\delta$ in $\overline{\delta}$) is a positive (resp.~negative) PIA formula, %denoting $\bigwedge_\lambda\gamma^\lambda$ (resp.~$\bigvee_\mu\delta^\mu$) its equivalent rewriting as per Lemma \ref{lemma: reduction to definite}, 
 by  Proposition \ref{prop: generalized derivation of amlost identities}, it is enough to show that for every $p$ and $q$, the sequents $\mathsf{MV}(p)[\oba/\obx, \orb/\ory]\vdash p$ and $q\vdash \mathsf{MV}(q)[\oba/\obx, \orb/\ory]$ are derivable in $\mathrm{\cfDLE}$ (resp.~$\mathrm{\cfDDLE}$) for all formulas $\mathsf{mv}(p)\in \mathsf{Mv}(p)$ and $\mathsf{mv}(q)\in \mathsf{Mv}(q)$. Let us show this latter statement.
 Each sequent $\mathsf{MV}(p)[\oba/\obx, \orb/\ory]\vdash p$ is of either of the following forms: \[\mathsf{LA}(\alpha_p)[\ba[\obp/\obp, \orq/\orq]/!\bu, \overline{\mathsf{MV}(p)}/\overline{p}, \overline{\mathsf{MV}(q)}/\overline{q}]\vdash p    \quad    \mathsf{RA}(\beta_p)[\beta[\obp/\obp, \orq/\orq]/!u, \overline{\mathsf{MV}(p)}/\overline{p}, \overline{\mathsf{MV}(q)}/\overline{q}]\vdash p,\]
 where $\alpha_p$ (resp.~$\beta_p$) denotes the definite positive (resp.~negative) PIA formula which occurs as a conjunct (resp.~disjunct) of $\alpha$ (resp.~$\beta$) as per Lemma \ref{lemma: reduction to definite}, which contains the $\epsilon$-critical occurrence of $p$ as a subformula (cf.~discussion around \eqref{eq: first approx anind definite}).
By Corollary \ref{cor: generalized deriving la almost implies atom}, it is enough to show that $\mathsf{MV}(p')\vdash p'$ and $q'\vdash \mathsf{MV}(q')$ are derivable in $\mathrm{\cfDLE}$ (resp.~$\mathrm{\cfDDLE}$) for each $p', q'<_\Omega p$ (which is true by the induction hypothesis, while the basis of the induction holds by Corollary \ref{cor: deriving la implies atom with weakening}), and $p\vdash p$ is derivable in $\mathrm{\cfDLE}$ (resp.~$\mathrm{\cfDDLE}$), which is of course the case. Likewise, one shows that the sequents $q\vdash \mathsf{MV}(q)[\oba/\obx, \orb/\ory]$ are derivable in $\mathrm{\cfDLE}$ (resp.~$\mathrm{\cfDDLE}$), which completes the proof that $(\varphi\fCenter \psi)[\oba, \orb, \overline{\bgamma}, \overline{\rdelta}]$
is derivable in in $\mathrm{\cfDLE} + \mathrm{R}$ (resp.~$\mathrm{\cfDDLE} + \mathrm{R}$). Finally, the derivation so generated only consists of identity axioms, and applications of display rules, right-introduction rules for  negative PIA-connectives, 
and left-introduction rules for positive PIA-connectives 
(and weakening and exchange rules in the case of $\mathrm{\cfDDLE}$) before the application of a rule in $\mathcal{R}$ (cf.~Proposition \ref{prop: generalized derivation of amlost identities} and Corollaries \ref{cor: generalized deriving la almost implies atom} and \ref{cor: deriving la implies atom with weakening}); moreover, after applying a rule in $\mathcal{R}$, the only rules applied are display rules, left-introduction rules for positive Skeleton-connectives and right-introduction rules for negative Skeleton-connectives (plus contraction in the case of $\mathrm{\cfDDLE}$), cf.~Proposition \ref{prop: the thing needed for the non-definite} and Footnote \ref{footnote:non definite general}). This completes the proof that the cut-free derivation in $\mathrm{\cfDLE} + \mathcal{R}$ (resp.~$\mathrm{\cfDDLE} + \mathcal{R}$) of $(\varphi\fCenter\psi)[\oba, \orb,\overline{\bgamma}, \overline{\rdelta}]$ is in pre-normal form. 
\end{proof}

\begin{example}
\label{ex:derivations general}
Let us illustrate the procedure described in the proposition above by deriving the sequents
in Example \ref{example:inductive and analytic inductive} using the rules computed in Example \ref{example:computing rules anint}. In the last derivation below, the symbol $\hat{\slash}_{\mor}$  denotes the left residual  of $\MOR$ in its  first coordinate.

%\rule[-5mm]{0mm}{0.4cm}\rule{0pt}{4.5ex}

\begin{center}
{\fns
\begin{tabular}{ccccc}
\normalsize{\cfDLE-derivation of $\wdia  \wbox p \fCenter\wbox \wdia p$:}
 & & & & 
\rule[-1.8mm]{0mm}{0mm}
\normalsize{\cfDDLE-derivation of $\wdia \wbox ( p \aand q) \fCenter \wbox \wdia p \aor  \wbox \wdia q$:}
 \\
\cline{0-1}
\cline{4-5}
\rule[0mm]{0mm}{0mm}
\hspace{-0.5cm}
\begin{tikzpicture}		
\node at(0,0) {
\AX$p \fCenter p $
\LL{\fns$\wbox_L$}
\UI$\wbox p \fCenter \WBOX p $
\UI$\BDIA \wbox p \fCenter p$
\RL{\fns$\wdia_R$}
\UI$\WDIA \BDIA \wbox p\fCenter \wdia p$
\RL{\fns$R_4$}
\UI$\WDIA  \wbox p  \fCenter \WBOX \wdia p$
\LL{\fns$\wdia_L$}
\UI$\wdia  \wbox p  \fCenter \WBOX \wdia p$
\RL{\fns$\wbox_R$}
\UI$\wdia  \wbox p \fCenter\wbox \wdia  p$
\DP
};
 %\draw[help lines] (-4,-4) grid (4,4);

\node[rotate = -90] at (-1.3, -1) {$\underbrace{\hspace{1.3cm}}$};
\node at (-2,-1) {(\ref{inv anlytic})};
\node[rotate = 90] at (1.7,-0.3) {$\underbrace{\hspace{0.6cm}}$};
\node at (2.2,-0.3) {(\ref{eq:structural rule (definite) analytic-inductive})};
\end{tikzpicture}

 & &\ \ \ \ \ \ \ \ \ \ & & 

\hspace{-1cm}
\begin{tikzpicture}		
\node at(0,0) {
\AX$p \fCenter p $
\RL{\fns$\wdia_R$}
\UI$\WDIA p  \fCenter \wdia p $
\UI$p \fCenter \BBOX \wdia p $
\LL{\fns$W_L$}
\UI$p \AAND q \fCenter \BBOX \wdia p$
\LL{\fns$\aand_L$}
\UI$p \aand q \fCenter \BBOX \wdia p$
\LL{\fns$\wbox_L$}
\UI$\wbox (p \aand q) \fCenter \WBOX \BBOX \wdia p$
\AX$q \fCenter q$
\RL{\fns$\wdia_R$}
\UI$\WDIA q \fCenter \wdia q$
\UI$q \fCenter \BBOX \wdia q$
\LL{\fns$W_L$}
\UI$q \AAND p \fCenter \BBOX \wdia q$
\LL{\fns$E_L$}
\UI$p \AAND q \fCenter \BBOX \wdia q$
\LL{\fns$\aand_L$}
\UI$p \aand q \fCenter \BBOX \wdia q$
\LL{\fns$\wbox_L$}
\UI$\wbox (p \aand q) \fCenter \WBOX \BBOX \wdia q$
\RL{\footnotesize$R_6$}
\BI$\WDIA \wbox ( p \aand q) \fCenter \WBOX \wdia p \AOR \WBOX \wdia q$
\UI$\WDIA \wbox ( p \aand q) \ADLARR \WBOX \wdia q \fCenter \WBOX \wdia p$
\RL{\fns$\wbox_R$}
\UI$\WDIA \wbox ( p \aand q) \ADLARR \WBOX \wdia q \fCenter \wbox \wdia p$
\UI$\WDIA \wbox ( p \aand q) \fCenter \wbox \wdia p \AOR \WBOX \wdia q$
\UI$ \wbox\wdia p\ADRARR  \WDIA \wbox ( p \aand q) \fCenter \WBOX \wdia q$
\RL{\fns$\wbox_R$}
\UI$ \wdia p\ADRARR  \WDIA \wbox ( p \aand q) \fCenter \wbox \wdia q$
\UI$\WDIA \wbox ( p \aand q) \fCenter \wbox \wdia p \AOR \wbox \wdia q$
\RL{\fns$\aor_R$}
\UI$\WDIA \wbox ( p \aand q) \fCenter \wbox \wdia p \aor \wbox \wdia q$
\LL{\fns$\wdia_L$}
\UI$\wdia \wbox ( p \aand q) \fCenter \wbox \wdia p \aor \wbox \wdia q$
\DP
};
 %\draw[help lines] (-4,-4) grid (4,4);

\node[rotate = -90] at (-2.75, -1.7) {$\underbrace{\hspace{4.4 cm}}$};
\node at (-3.35,-1.7) {(\ref{inv anlytic})};
\node[rotate = 90] at (4,0.5) {$\underbrace{\hspace{0.6cm}}$};
\node at (4.5,0.5) {(\ref{eq:structural rule (definite) analytic-inductive})};
\end{tikzpicture}
 \\
\end{tabular}
 }
\end{center}

\begin{center}
{\fns
\begin{tabular}{@{}c@{}}
\normalsize{\cfDDLE-derivation of $p\ararr (q \ararr r) \fCenter ((p \ararr q) \ararr (\wbox p \ararr  r)) \mor \wdia r$:}
\rule[-1.8mm]{0mm}{0mm} \\
\hline
\hspace{-2.65cm}
\begin{tikzpicture}
\node at(0,0) {
\small
\AXC{$p\fCenter p$}
\AXC{$p \fCenter p$}
\AXC{$q\fCenter q$}
\AXC{$r \fCenter r$}
\LL{\fns$\ararr_L$}
\BIC{$q \ararr r \fCenter q \ARARR r$}
\LL{\fns$\ararr_L$}
\BIC{$p\ararr (q \ararr r) \fCenter p \ARARR(q \ARARR r)$}
\UIC{$p \AAND(p\ararr (q \ararr r) ) \fCenter q \ARARR r$}
\UIC{$p \fCenter (q \ARARR r) \ALARR (p\ararr (q \ararr r) )$}
\LL{\fns$\wbox_L$}
\UIC{$\wbox p \fCenter \WBOX ((q \ARARR r) \ALARR (p\ararr (q \ararr r) ))$}
\UIC{$\BDIA \wbox p \fCenter ((q \ARARR r) \ALARR (p\ararr (q \ararr r) ))$}
\UIC{$\BDIA \wbox p  \AAND(p\ararr (q \ararr r) )\fCenter q \ARARR r $}
\UIC{$ q\AAND (\BDIA \wbox p  \AAND(p\ararr (q \ararr r) )) \fCenter r $}
\UIC{$ q\fCenter r\ALARR (\BDIA \wbox p  \AAND(p\ararr (q \ararr r) ))  $}
\LL{\fns$\ararr_L$}
\BIC{$p \ararr q \fCenter p \ARARR (r\ALARR (\BDIA \wbox p  \AAND(p\ararr (q \ararr r) )))$}
\UIC{$p \AAND (p \ararr q) \fCenter r\ALARR (\BDIA \wbox p  \AAND(p\ararr (q \ararr r) ))$}
\UIC{$p  \fCenter (r\ALARR (\BDIA \wbox p  \AAND(p\ararr (q \ararr r) ))) \ALARR (p \ararr q)$}
\LL{\fns$\wbox_L$}
\UIC{$ \wbox p  \fCenter \WBOX (r\ALARR (\BDIA \wbox p  \AAND(p\ararr (q \ararr r) ))) \ALARR (p \ararr q)$}
\UIC{$ \BDIA \wbox p  \fCenter r\ALARR ((\BDIA \wbox p  \AAND(p\ararr (q \ararr r) ))) \ALARR (p \ararr q)$}
\UIC{$ \BDIA \wbox p \AAND (p \ararr q)  \fCenter r \ALARR ((\BDIA \wbox p  \AAND(p\ararr (q \ararr r) ))) $}
%\AX$ \BDIA \wbox p \AAND (p \ararr q)  \fCenter r\ALARR ((\BDIA \wbox p  \AAND(p\ararr (q \ararr r) ))) $
\UIC{$ (\BDIA \wbox p \AAND (p \ararr q))\AAND(\BDIA \wbox p  \AAND(p\ararr (q \ararr r) ))   \fCenter r$}
\UIC{$ \WDIA ((\BDIA \wbox p \AAND (p \ararr q))\AAND(\BDIA \wbox p  \AAND(p\ararr (q \ararr r) )))  \fCenter \wdia r$}
\UIC{$ (\BDIA \wbox p \AAND (p \ararr q))\AAND(\BDIA \wbox p  \AAND(p\ararr (q \ararr r) ))  \fCenter  \BBOX \wdia r$}
\UIC{$ p\ararr (q \ararr r)   \fCenter \BDIA \wbox p\ARARR (\BDIA \wbox p \AAND (p \ararr q) \ARARR\BBOX \wdia r)$}
\AXC{$p\fCenter p$}
\AXC{$p \fCenter p$}
\AXC{$q\fCenter q$}
\AXC{$r \fCenter r$}
\LL{\fns$\ararr_L$}
\BIC{$q \ararr r \fCenter q \ARARR r$}
\LL{\fns$\ararr_L$}
\BIC{$p\ararr (q \ararr r) \fCenter p \ARARR(q \ARARR r)$}
\UIC{$p \AAND(p\ararr (q \ararr r) ) \fCenter q \ARARR r$}
\UIC{$p \fCenter (q \ARARR r) \ALARR (p\ararr (q \ararr r) )$}
\LL{\fns$\wbox_L$}
\UIC{$\wbox p \fCenter \WBOX ((q \ARARR r) \ALARR (p\ararr (q \ararr r) ))$}
\UIC{$\BDIA \wbox p \fCenter ((q \ARARR r) \ALARR (p\ararr (q \ararr r) ))$}
\UIC{$\BDIA \wbox p  \AAND(p\ararr (q \ararr r) )\fCenter q \ARARR r $}
\UIC{$ q\AAND (\BDIA \wbox p  \AAND(p\ararr (q \ararr r) )) \fCenter r $}
\UIC{$ q\fCenter r\ALARR (\BDIA \wbox p  \AAND(p\ararr (q \ararr r) ))  $}
\LL{\fns$\ararr_L$}
\BIC{$p \ararr q \fCenter p \ARARR (r\ALARR (\BDIA \wbox p  \AAND(p\ararr (q \ararr r) )))$}
\UIC{$p \AAND (p \ararr q) \fCenter r\ALARR (\BDIA \wbox p  \AAND(p\ararr (q \ararr r) ))$}
\UIC{$p  \fCenter (r\ALARR (\BDIA \wbox p  \AAND(p\ararr (q \ararr r) ))) \ALARR (p \ararr q)$}
\LL{\fns$\wbox_L$}
\UIC{$ \wbox p  \fCenter \WBOX (r\ALARR (\BDIA \wbox p  \AAND(p\ararr (q \ararr r) ))) \ALARR (p \ararr q)$}
\UIC{$ \BDIA \wbox p  \fCenter r\ALARR ((\BDIA \wbox p  \AAND(p\ararr (q \ararr r) ))) \ALARR (p \ararr q)$}
\UIC{$ \BDIA \wbox p \AAND (p \ararr q)  \fCenter r \ALARR ((\BDIA \wbox p  \AAND(p\ararr (q \ararr r) ))) $}
%\AX$ \BDIA \wbox p \AAND (p \ararr q)  \fCenter r\ALARR ((\BDIA \wbox p  \AAND(p\ararr (q \ararr r) ))) $
\UIC{$ (\BDIA \wbox p \AAND (p \ararr q))\AAND(\BDIA \wbox p  \AAND(p\ararr (q \ararr r) )))  \fCenter r$}
\UIC{$ \BDIA \wbox p  \AAND(p\ararr (q \ararr r) )   \fCenter (\BDIA \wbox p \AAND (p \ararr q))\ARARR r$}
\UIC{$ p\ararr (q \ararr r)    \fCenter \BDIA \wbox p \ARARR((\BDIA \wbox p \AAND (p \ararr q))\ARARR r)$}
\RL{\fns$R_5$}
\BIC{$p\ararr (q \ararr r) \fCenter ((p \ararr q)\ARARR(\wbox p \ARARR  r)) \MOR \wdia r$}
\UIC{$(p\ararr (q \ararr r)) \ \hat{\slash}_{\mor}\ \wdia r \fCenter (p \ararr q) \ARARR (\wbox p \ARARR r)$}
\UIC{$(p \ararr q) \AAND ((p\ararr (q \ararr r)) \ \hat{\slash}_{\mor}\ \wdia r) \fCenter \wbox p \ARARR r$}
\RL{\fns$\ararr_R$}
\UIC{$(p \ararr q) \AAND ((p\ararr (q \ararr r)) \ \hat{\slash}_{\mor}\ \wdia r) \fCenter \wbox p \ararr r$}
\UIC{$(p\ararr (q \ararr r)) \ \hat{\slash}_{\mor}\ \wdia r \fCenter (p \ararr q) \ARARR (\wbox p \ararr r)$}
\RL{\fns$\ararr_R$}
\UIC{$(p\ararr (q \ararr r)) \ \hat{\slash}_{\mor}\ \wdia r \fCenter (p \ararr q) \ararr (\wbox p \ararr r)$}
\UIC{$p\ararr (q \ararr r) \fCenter ((p \ararr q) \ararr (\wbox p \ararr  r)) \MOR \wdia r$}
\RL{\fns$\mor_R$}
\UIC{$p\ararr (q \ararr r) \fCenter ((p \ararr q) \ararr (\wbox p \ararr  r)) \mor \wdia r$}
\DP
 };
 \node[rotate = -90] at (-1.5, -4.72) {$\underbrace{\hspace{3.7 cm}}$};
\node at (-2.15,-4.72) {(\ref{inv anlytic})};
\node[rotate = 90] at (9.35,-2.8) {$\underbrace{\hspace{0.6cm}}$};
\node at (9.85,-2.8) {(\ref{eq:structural rule (definite) analytic-inductive})};
\end{tikzpicture}
\end{tabular}
 }
\end{center}
\end{example}

\section{Conclusions}
\label{sec:conclusions}
%%%%%%%%%%%%%%%%%%%%%%%%

\paragraph{Main contribution}
In this article we showed that, for any {\em properly displayable} (D)LE-logic $\mathsf{L}$ (i.e.~a (D)LE-logic axiomatized by {\em analytic inductive} axioms, cf.~Definition \ref{def:type5}), the proper display calculus for $\mathsf{L}$---i.e.~the calculus obtained by adding the analytic structural rules corresponding to the axioms of $\mathsf{L}$ to the basic calculus $\mathrm{D.LE}$ (resp.~$\mathrm{D.DLE}$)---derives all the theorems of $\mathsf{L}$. This is what we refer to as the {\em syntactic completeness} of the proper display calculus for $\mathsf{L}$ with respect to $\mathsf{L}$. We achieved this result by providing an {\em effective procedure} for generating a derivation---which is not only {\em cut-free} but also in {\em pre-normal form} (cf.~Definitions \ref{def:canonical form nonDist} and \ref{def:canonical form Dist})---of any analytic inductive axiom in any (D)LE-signature in the proper display calculus $\mathrm{D.LE}$ (resp.~$\mathrm{D.DLE}$) augmented with the analytic structural rule(s) corresponding to the given axiom. 

\paragraph{Scope} Since (D)LE-logics encompass a wide family of well known logics (modal, intuitionistic, substructural), and since analytic inductive axioms provide a formulation of the notion of analyticity based on the syntactic shape of formulas/sequents, the results of the present paper directly apply to all logical settings for which analytic (proper display) calculi have been defined, such as those of \cite{GMPTZ,Belnap,Wa98,Wan02,Kracht,CiRa14}. Moreover, in the present paper we have worked in a single-type environment, mainly for ease of exposition. However,  all the results mentioned above  straightforwardly apply also to properly displayable {\em multi-type calculi}, which have been recently introduced to extend the scope and benefits of proper display calculi also to a wide range of logics that for various reasons do not fall into the scope of the analytic inductive definition. These logics crop up in various areas of the literature and include well known logics such as linear logic \cite{GrecPal17}, dynamic epistemic logic \cite{Multitype}, semi De Morgan logic \cite{GreLiaMosPal17,GreLianMosPal2020}, bilattice logic \cite{GreLiaPalRiv19}, inquisitive logic \cite{inquisitive}, non normal modal logics \cite{CheGrePalTzi19}, the logics of classes of rough algebras \cite{GreLiaManPal19,GrecJipManPalTzi19}. Interestingly, the multi-type framework can be also usefully deployed to introduce logics specifically designed to describe and reason about the interaction of entities of different types, as done e.g.~in  \cite{LoRC,PDL}. 

%\paragraph{Related results}

\paragraph{The syntax-semantics interface on analytic calculi} The main insight  developed in the research line to which the present paper pertains is that there is a close connection between {\em semantic} results pertaining to correspondence theory and the {\em syntactic} theory of analytic calculi. This close connection, which has been observed and also exploited by several authors in various proof-theoretic settings (cf.~e.g.~\cite{Kracht,negri2005proof,ciabattoni2008axioms}), gave rise in \cite{GMPTZ} to the notion of analytic inductive inequalities as the uniform and independent identification, across signatures, of the syntactic shape (semantically motivated by the order-theoretic properties of the algebraic interpretation of the logical connectives) which guarantees the desiderata of analyticity. In this context, the core of the ``syntax-semantic interface'' is the algorithm ALBA, which serves to compute {\em both} the first-order correspondent of analytic inductive axioms {\em and} their corresponding analytic structural rules. In this paper, we saw the analytic structural rules computed by ALBA at work as the key cogs of the machinery of proper display calculi to derive the axioms that had generated them. This result can be understood as the purely syntactic counterpart of the proof that ALBA preserves semantic equivalence on complete algebras (cf.~\cite{ConPal12,ConPalZha19,ConPal20}), which has been used in  \cite{GMPTZ} to motivate the semantic equivalence of any given analytic inductive axiom with its corresponding ALBA-generated structural rules. This observation paves the way to various questions, among which, whether  information about the derivation of a given analytic inductive axiom can be extracted directly from its successful ALBA-run, or conversely, whether information about (optimal) ALBA-runs of analytic inductive  axioms can be extracted from its derivation in pre-normal form, or whether the recent independent topological characterization of analytic inductive inequalities established in \cite{DeRPal20} can be exploited for proof-theoretic purposes.

%\bibliographystyle{abbrv}
%\bibliography{reference}

%%%%%%%%%%%%%%%%%%%%%%%%
	
\end{document}